\documentclass[11pt]{article}   
\usepackage[letterpaper,hmargin=2.1cm,vmargin=3cm]{geometry}
\usepackage{graphicx}
\usepackage{amsmath}            
\usepackage{amssymb}            
\usepackage{amsthm}             
\usepackage{amsfonts}           
\usepackage{url}                
\usepackage{color}              
\usepackage{multirow}           
\usepackage{cite}               
\usepackage{hyperref}           
\hypersetup{colorlinks=true,linkcolor=[rgb]{0.75,0,0},citecolor=[rgb]{0,0,0.75}}%
\usepackage{epstopdf}           
\usepackage{thmtools}			
\usepackage{thm-restate}
\declaretheorem[name=Theorem]{theorem}
\declaretheorem[name=Lemma,numberwithin=section]{lemma}

\newcommand{\floor}[1]{\left\lfloor #1\right\rfloor}
\newcommand{\ceil}[1]{\left\lceil #1\right\rceil}
\newcommand{\ang}[1]{\langle #1\rangle}
\newcommand{\RE}{\mathbb{R}}    
\newcommand{\eps}{\varepsilon}  

\newcommand{\SoftOh}{\widetilde{O}}
\newcommand{\SoftOmega}{\widetilde{\Omega}}
\newcommand{\OO}[1]{O\kern-2pt\left(#1\right)}  

\newcommand{\inv}[1]{\frac{1}{#1}}
\newcommand{\alg}{\textrm{SplitReduce}}
\newcommand{\diam}{\mathrm{diam}}
\newcommand{\radius}{\mathrm{radius}}
\newcommand{\vol}{\mathrm{vol}}
\newcommand{\area}{\mathrm{area}}
\newcommand{\conv}{\mathrm{conv}}
\newcommand{\Vor}{\mathrm{Vor}}
\newcommand{\etal}{\textit{et al.}}
\renewcommand{\P}{\kern+1pt}    
\newcommand{\N}{\kern-2pt}      
\newcommand{\NN}{\kern-4pt}     
\newcommand{\polar}[1]{\mathrm{polar}(#1)}            
\newcommand{\polarX}[2]{\mathrm{polar}_{#1}(#2)}      

\begin{document}

\title{Approximate Polytope Membership Queries%
\thanks{Preliminary results of this paper appeared in ``Approximate Polytope Membership Queries'', in
Proc. ACM Sympos. Theory Comput. (STOC), 2011, 579--586 and ``Polytope Approximation and the Mahler Volume'', in Proc. ACM-SIAM Sympos. Discrete Algorithms (SODA), 2012, 29--42.}
}

\author{%
	Sunil Arya\thanks{Research supported by the Research Grants Council of Hong Kong, China under project numbers 610108 and 16200014.}\\
		Department of Computer Science and Engineering \\
		Hong Kong University of Science and Technology \\
		Clear Water Bay, Kowloon, Hong Kong\\
		arya@cse.ust.hk \\
		\and
	Guilherme D. da Fonseca\thanks{Research supported by CNPq and FAPERJ grants.}\\
		Universit\'e Clermont Auvergne and LIMOS \\
		Clermont-Ferrand, France\\
		fonseca@isima.fr \\
		\and
	David M. Mount\thanks{Research supported by NSF grant CCF--1618866.}\\
		Department of Computer Science and \\
		Institute for Advanced Computer Studies \\
		University of Maryland \\
		College Park, Maryland 20742 \\
		mount@cs.umd.edu \\
}

\date{}

\maketitle

\begin{abstract}
In the polytope membership problem, a convex polytope $K$ in $\RE^d$ is given, and the objective is to preprocess $K$ into a data structure so that, given any query point $q \in \RE^d$, it is possible to determine efficiently whether $q \in K$. We consider this problem in an approximate setting. Given an approximation parameter $\eps$, the query can be answered either way if the distance from $q$ to $K$'s boundary is at most $\eps$ times $K$'s diameter. We assume that the dimension $d$ is fixed, and $K$ is presented as the intersection of $n$ halfspaces. Previous solutions to approximate polytope membership were based on straightforward applications of classic polytope approximation techniques by Dudley (1974) and Bentley {\etal} (1982). The former is optimal in the worst-case with respect to space, and the latter is optimal with respect to query time. 
 
We present four main results. First, we show how to combine the two above techniques to obtain a simple space-time trade-off. Second, we present an algorithm that dramatically improves this trade-off. In particular, for any constant $\alpha \ge 4$, this data structure achieves query time roughly $O\big(1/\eps^{(d-1)/\alpha}\big)$ and space roughly $O\big(1/\eps^{(d-1)(1 - \Omega(\log \alpha)/\alpha)}\big)$. We do not know whether this space bound is tight, but our third result shows that there is a convex body such that our algorithm achieves a space of at least $\Omega\big( 1/\eps^{(d-1)(1-O(\sqrt{\alpha})/\alpha} \big)$. Our fourth result shows that it is possible to reduce approximate Euclidean nearest neighbor searching to approximate polytope membership queries. Combined with the above results, this provides significant improvements to the best known space-time trade-offs for approximate nearest neighbor searching in $\RE^d$. For example, we show that it is possible to achieve a query time of roughly $O(\log n + 1/\eps^{d/4})$ with space roughly $O(n/\eps^{d/4})$, thus reducing by half the exponent in the space bound. 
\end{abstract}


\section{Introduction} \label{sec:intro}

Convex polytopes are key structures in many areas of mathematics and computation. In this paper, we consider a fundamental search problem related to convex polytopes. Let $K$ denote a convex body in $\RE^d$, that is, a closed, convex set of bounded diameter that has a nonempty interior. We assume that $K$ is presented as the intersection of $n$ closed halfspaces. (Our results generally hold for any representation that satisfies the access primitives given at the start of Section~\ref{sec:split-reduce}.) The \emph{polytope membership problem} is that of preprocessing $K$ so that it is possible to determine efficiently whether a given query point $q \in \RE^d$ lies within $K$. Throughout, we assume that the dimension $d$ is a fixed constant that is at least $2$. 

It follows from standard results in projective duality that polytope membership is equivalent to answering halfspace emptiness queries for a set of $n$ points in $\RE^d$. In dimension $d \leq 3$, it is possible to build a data structure of linear size that can answer such queries in logarithmic time~\cite{textbook}. In higher dimensions, however, all known exact data structures with roughly linear space have a query time
of $\SoftOmega\big(n^{1-1/\floor{d/2}}\big)$%
\footnote{Throughout, we use $\SoftOh(\cdot)$ and $\SoftOmega(\cdot)$ as variants of $O(\cdot)$ and $\Omega(\cdot)$, respectively, that ignore logarithmic factors. We use ``$\lg$'' to denote base-2 logarithm.}
\cite{Mat92}, which is unacceptably high for many applications. Polytope membership is a special case of polytope intersection queries \cite{ChD87,DoK83,BaL15}. Barba and Langerman \cite{BaL15} showed that for any fixed $d$, it is possible to preprocess polytopes in $\RE^d$ so that given two such polytopes that have been translated and rotated, it can be determined whether they intersect each other in time that is logarithmic in their total combinatorial complexity. However, the preprocessing time and space grow as the combinatorial complexity of the polytope raised to the power $\floor{d/2}$.

The lack of efficient exact solutions motivates the question of whether polytope membership queries can be answered approximately. Let $\eps$ be a positive real parameter, and let $\diam(K)$ denote $K$'s diameter. Given a query point $q \in \RE^d$, an \emph{$\eps$-approximate polytope membership query} returns a positive result if $q \in K$, a negative result if the distance from $q$ to its closest point in $K$ is greater than $\eps \cdot \diam(K)$, and it may return either result otherwise (see Figure~\ref{fig:apxmembership}(a)). Polytope membership queries, both exact and approximate, arise in many application areas, such as linear-programming and ray-shooting queries~\cite{lp-Chan, lp-Chan2, lp-Ramos, lp-Mat, lp-Mat2}, nearest neighbor searching and the computation of extreme points~\cite{hull-Chan,Clarkson-ANN}, collision detection~\cite{collision2}, and machine learning~\cite{SVM}. 

\begin{figure}[htbp]
  \centerline{\includegraphics[scale=0.40]{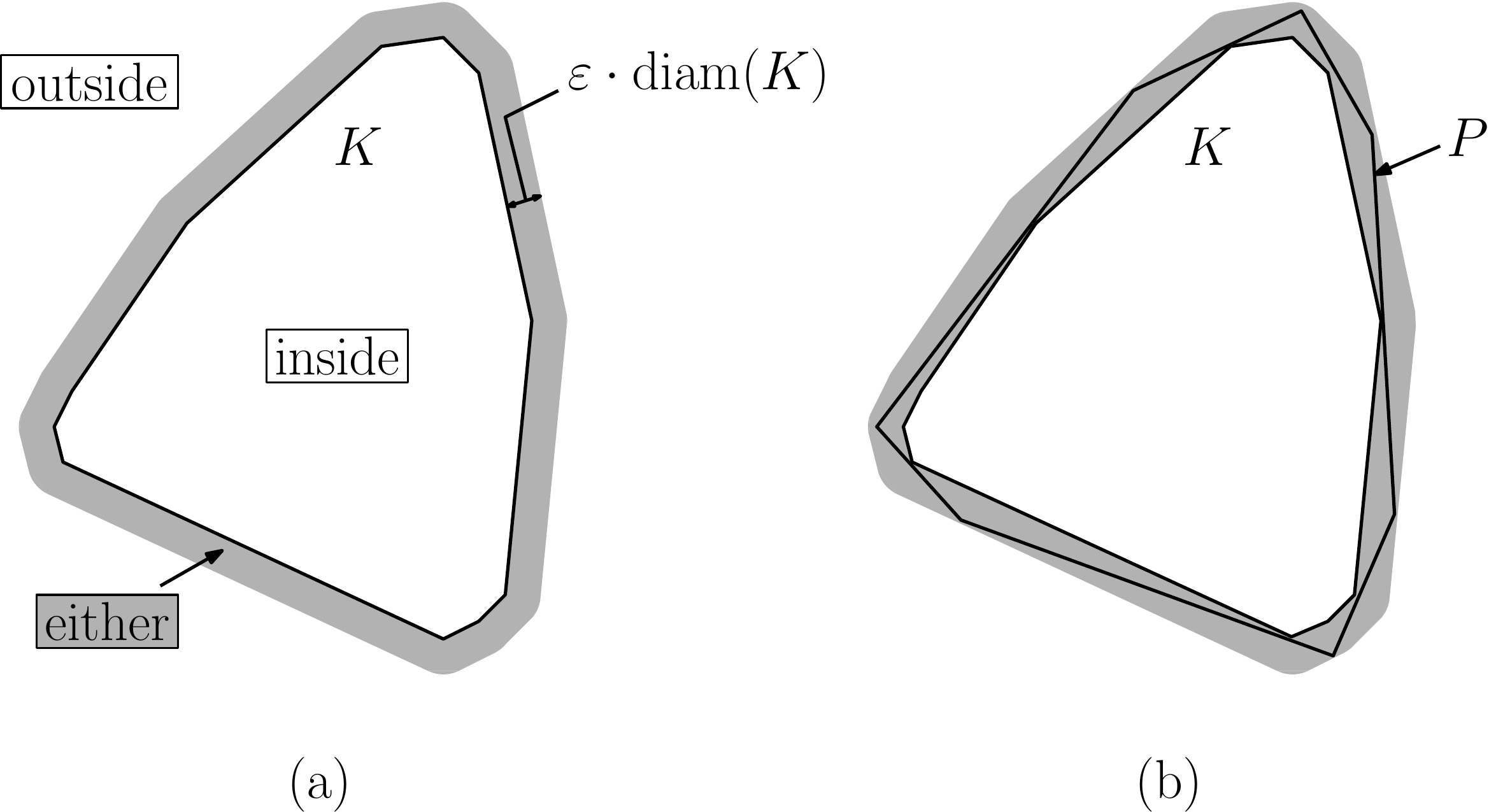}}
  \caption{Approximate polytope membership: (a) problem formulation, (b) outer $\eps$-approximation.}
  \label{fig:apxmembership}
\end{figure}

Existing solutions to approximate polytope membership queries have been based on straightforward applications of classic polytope approximation techniques. We say that a polytope $P$ is an \emph{outer $\eps$-approximation} of $K$ if $K \subseteq P$, and the Hausdorff distance between $P$ and $K$ is at most $\eps \cdot \diam(K)$ (see Figure~\ref{fig:apxmembership}(b)). An \emph{inner $\eps$-approximation} is defined similarly but with $P \subseteq K$. Dudley~\cite{Dudley} showed that there exists an outer $\eps$-approximating polytope for any bounded convex body in $\RE^d$ formed by the intersection of $O\big(1/\eps^{(d-1)/2}\big)$ halfspaces, and Bronshteyn and Ivanov~\cite{BrIv} proved an analogous bound on the number of vertices needed to obtain an inner $\eps$-approximation. Both bounds are known to be asymptotically tight in the worst case (see, e.g., \cite{Bro08}). These results have been applied to a number of problems, for example, the construction of coresets~\cite{kernel-survey}. By checking that a given query point lies within each of the halfspaces of Dudley's approximation, $\eps$-approximate polytope membership queries can be answered with space and query time of $O\big(1/\eps^{(d-1)/2}\big)$.

\medskip

The principal contribution of this paper is to show that it is possible to achieve nontrivial space-time trade-offs for approximate polytope membership. In order to motivate our methods, in Section~\ref{sec:prelim} we present a simple space-time trade-off (stated in the following theorem), which is based on a straightforward combination of the approximations of Dudley~\cite{Dudley} and Bentley~{\etal}~\cite{BFP}. Throughout, we will treat $n$ and $\eps$ as asymptotic quantities, while the dimension $d$ is a constant.

\begin{theorem}[Simple Trade-off] \label{thm:simple-trade-off}
Given a convex polytope $K$ in $\RE^d$, a positive approximation parameter $\eps$, and a real parameter $\alpha \geq 2$, there is a data structure for $\eps$-approximate polytope membership queries that achieves
\[
\hbox{Query time:~} O\left(1/\eps^{\frac{d-1}{\alpha}}\right)
\qquad
\textrm{Space:~} O\left( 1/\eps^{(d-1)\left(1 - \inv{\alpha} \right)} \right).
\]
The constant factors in the space and query time depend only on $d$ (not on $K$, $\alpha$, or $\eps$). 
\end{theorem}

We will strengthen this trade-off significantly in Sections~\ref{sec:split-reduce} and~\ref{sec:firstbound}. We will show that it is possible to build a data structure with $O\big(1/\eps^{(d-1)/2}\big)$ space that allows polytope membership queries to be answered in roughly $O\big(1/\eps^{(d-1)/8}\big)$ time, thus reducing the exponent in the query time of Theorem~\ref{thm:simple-trade-off} (for $\alpha=2$) by $1/4$. Further, we will show that by iterating a suitable generalization of this construction, we can obtain the following trade-offs.

\begin{theorem} \label{thm:membership-ub}
Given a convex polytope $K$ in $\RE^d$, an approximation parameter $0 < \eps \le 1$, and a real constant $\alpha \geq 4$, there is a data structure for  $\eps$-approximate polytope membership queries that achieves
\[
\hbox{Query time:~} O\left(({\textstyle \log \inv{\eps}})/\eps^{\frac{d-1}{\alpha}}\right)
\qquad
\textrm{Space:~} \OO{1/\eps^{(d-1)\left(1 - \frac{2\floor{\lg \alpha} - 2}{\alpha} \right)}}.
\]
The constant factors in the space and query time depend only on $d$ and $\alpha$ (not on $K$ or $\eps$).
\end{theorem}

The above space bound is a simplification, and the exact bound is given in Lemma~\ref{lem:trade-off-ub}. Both bounds are piecewise linear in $1/\alpha$ (with breakpoints at powers of two), but the bounds of Lemma~\ref{lem:trade-off-ub} are continuous as a function of $\alpha$. The resulting space-time trade-off is illustrated below in Figure~\ref{fig:trade-offs}(a). (The plot reflects the more accurate bounds.) 

The above theorem is intentionally presented in a purely existential form. This is because our construction algorithm assumes the existence of a procedure that computes an $\eps$-approximating polytope whose number of bounding hyperplanes is at most a constant factor larger than optimal. Unfortunately, we know of no efficient solution to this problem. In Lemma~\ref{lem:preproc-time} we will show that, if the input polytope is expressed as the intersection of $n$ halfspaces, it is possible to build such a structure in time $O\big(n + 1/\eps^{O(1)}\big)$, such that the space and query times of the above theorem increase by an additional factor of $O(\log \inv{\eps})$.

Note that, in contrast to many complexity bounds in the area of convex approximation, which hold only in the limit as $\eps$ approaches zero (see, e.g.,~\cite{Gru93,Bor00}), Theorems~\ref{thm:simple-trade-off} and~\ref{thm:membership-ub} hold for any positive $\eps \le 1$. The data structure of Theorem~\ref{thm:membership-ub} is quite simple. It is based on a quadtree subdivision of space in which each cell is repeatedly subdivided until the combinatorial complexity of the approximating polytope within the cell is small enough to achieve the desired query time. 

We do not know whether the upper bounds presented in Theorem~\ref{thm:membership-ub} are tight for our algorithm. In Section~\ref{sec:lb}, we establish the following lower bound on the trade-off achieved by this algorithm.

\begin{theorem} \label{thm:lb}
In any fixed dimension $d \ge 2$ and for any constant $\alpha \ge 4$, there exists a polytope such that for all sufficiently small positive $\eps$, the data structure described in Theorem~\ref{thm:membership-ub} when generated to achieve query time $O\big(1/\eps^{(d-1)/\alpha}\big)$ has space
\[
	\Omega\left(1/\eps^{(d-1) \left(1 - \frac{2\sqrt{2\alpha} - 3}{\alpha}\right) - 1} \right).
\]
\end{theorem}

Although $\alpha$ is not an asymptotic quantity, for the sake of comparing the upper and lower bounds, let us imagine that it is. For roughly the same query time, the $\alpha$ dependencies appearing in the exponents of the upper bounds on space are $\big(1 - \inv{\alpha}\big)$ for Theorem~\ref{thm:simple-trade-off}, $\big( 1 - \frac{\Omega(\log \alpha)}{\alpha}\big)$ for Theorem~\ref{thm:membership-ub}, and the lower bound of Theorem~\ref{thm:lb} is roughly $\big(1 - \frac{O(\sqrt{\alpha})}{\alpha}\big)$. The trade-offs provided in these theorems are illustrated in Figure~\ref{fig:trade-offs}(a).

\begin{figure}[htbp]
	\centerline{\begin{tabular}{ccc}
		  \includegraphics[width=6.0cm]{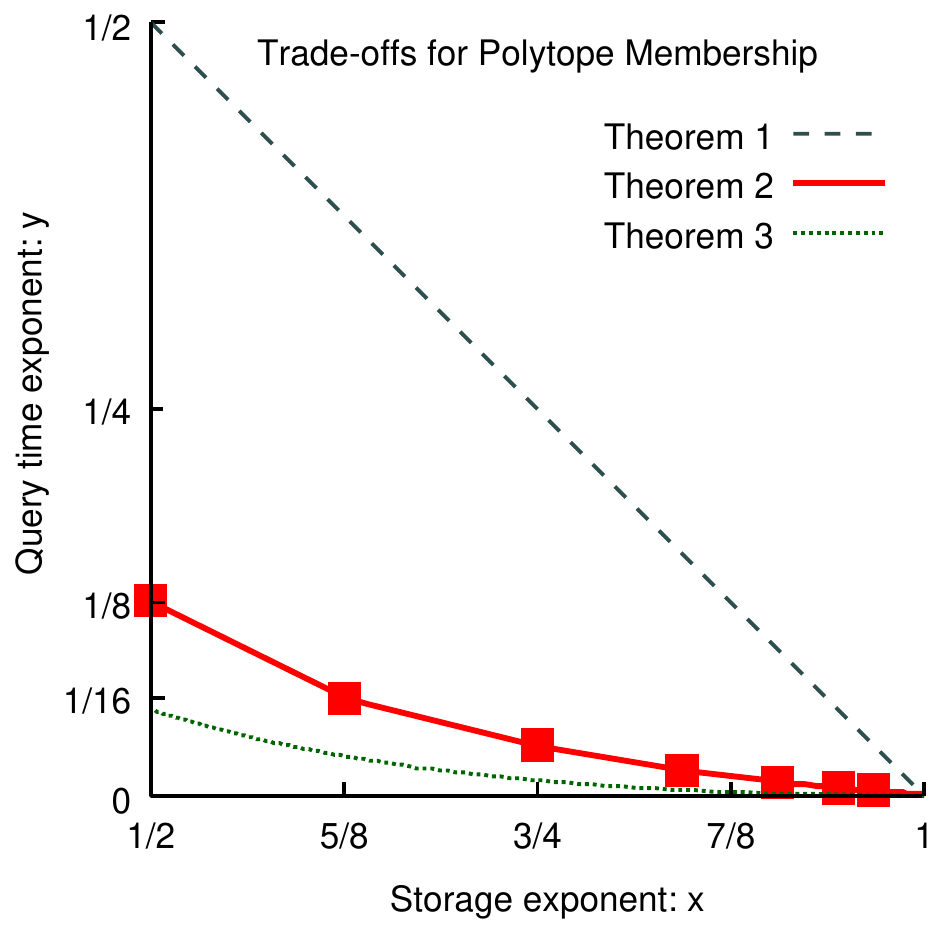} & \hspace*{0cm} &
		  \includegraphics[width=6.0cm]{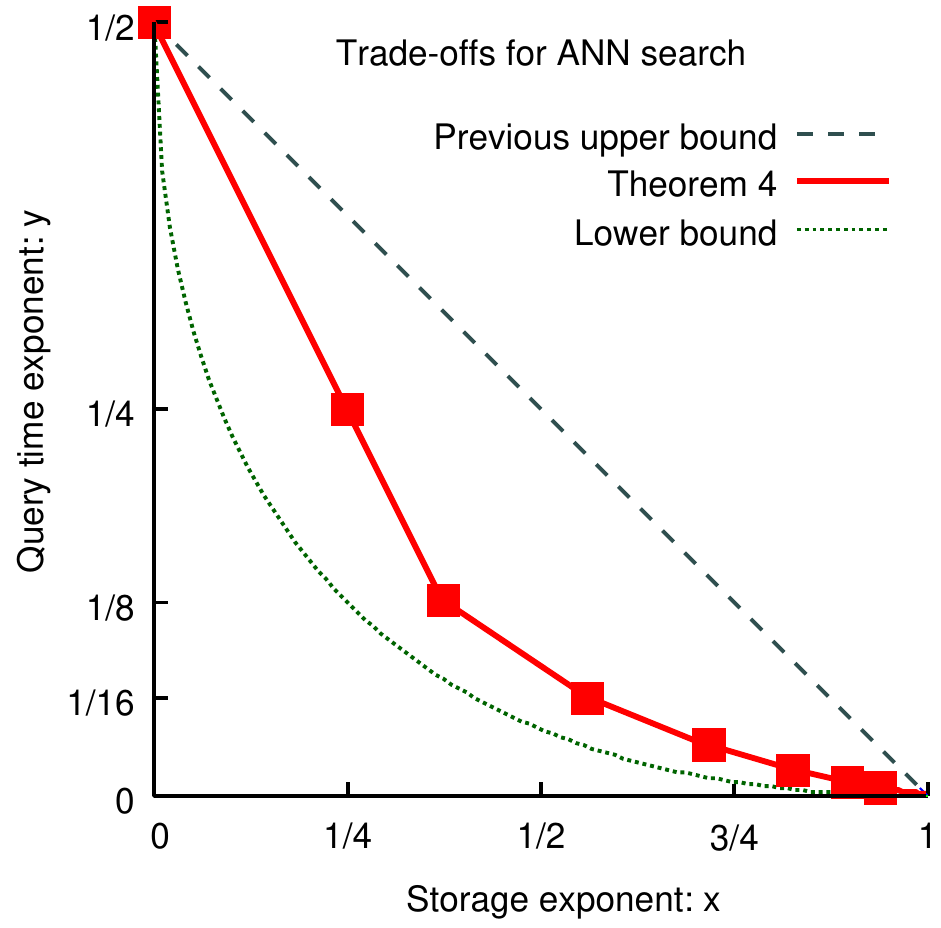} \\
		  (a) & \hspace{0cm} & (b)
	\end{tabular}}
  \caption{The multiplicative factors in the exponent of the $1/\eps$ terms for (a) polytope membership queries and (b) approximate nearest neighbor (ANN) queries. Each point $(x,y)$ represents a term of $1/\eps^{x(d \P \pm \P O(1))}$ for storage and $1/\eps^{y(d \P \pm \P O(1))}$ for query time, where the $O(1)$ term does not depend on $d$.}
  \label{fig:trade-offs}
\end{figure}

\medskip

The second major contribution of this paper is to demonstrate that our trade-offs for approximate polytope membership queries imply significant improvements to the best known space-time trade-offs for approximate nearest neighbor searching (ANN). We are given a set $X$ of $n$ points in $\RE^d$. Given any $q \in \RE^d$, an \emph{$\eps$-approximate nearest neighbor} of $q$ is any point of $X$ whose distance from $q$ is at most $(1+\eps)$ times the distance to $q$'s closest point in $X$. The objective is to preprocess $X$ in order to answer such queries efficiently. Data structures for approximate nearest neighbor searching (in fixed dimensions) have been proposed by several authors~\cite{Chan-ANN02, DGK01, HP-AVD, Clarkson-ANN, SSS06}. The best space-time trade-offs~\cite{AVD-JACM} have query times roughly $O\big(1/\eps^{d/\alpha}\big)$ with storage roughly $O\big(n/\eps^{d(1-2/\alpha)}\big)$, for $\alpha \geq 2$ (see the dashed line in Figure~\ref{fig:trade-offs}(b)). 

These results are based on a data structure called an \emph{approximate Voronoi diagram} (or AVD). In general, a data structure for approximate nearest neighbor searching is said to be in the \emph{AVD model} if it has the general form of decomposition of space (generally a covering) by hyperrectangles of bounded aspect ratio, each of which is associated with a set of representative points. Given any hyperrectangle that contains the query point, at least one of these representatives is an $\eps$-approximate nearest neighbor of the query point~\cite{AVD-JACM}. The AVD model is of interest because it is possible to prove lower bounds on the performance of such a data structure. In particular, the lower bounds proved in \cite{AVD-JACM} are shown in the dotted curve in Figure~\ref{fig:trade-offs}(b). By violating the AVD model, small additional improvements were obtained in~\cite{proximity}.

Our improvements to ANN searching are given in the following theorem.

\begin{theorem} \label{thm:ann-ub}
Let $0 < \eps \leq 1$ be a real parameter, $\alpha \geq 1$ be a real constant, and $X$ be a set of $n$ points in $\RE^d$. There is a data structure in the AVD model for approximate nearest neighbor searching that achieves
\begin{eqnarray*}
\text{Query time:} & & O\big(\log n + (1/\eps^{d/2 \alpha}) \cdot  {\textstyle \log^2 \inv{\eps}}\big) \\
\text{Space:} & & 
  \OO{n \cdot \max\left({\textstyle \log \inv{\eps}}, 1/\eps^{d\left(\inv{2} -\inv{2\alpha}\right)}\right)}\kern-2pt , \text{~for $1 \leq \alpha < 2$ and,} \\
  & &
  \OO{n/\eps^{d \left(1 - \frac{\floor{\lg\alpha}}{\alpha} - \inv{2 \alpha}\right)}}\kern-2pt , \text{~for $\alpha \geq 2$}.
\end{eqnarray*}
The constant factors in the space and query time depend only on $d$ and $\alpha$ (not on $\eps$). 
\end{theorem}

The above space bound is a simplification of the more accurate bound given in Lemma~\ref{lem:ann-ub}. (Also see the remarks following the proof of this lemma for further minor improvements achievable by forgoing the AVD model.) As before, both bounds are piecewise linear in $1/\alpha$ (with breakpoints at powers of two), but the bounds of Lemma~\ref{lem:ann-ub} are continuous as a function of $\alpha$. The resulting space-time trade-off is illustrated in Figure~\ref{fig:trade-offs}(b). (The plot reflects the more accurate bounds of Lemma~\ref{lem:ann-ub}.)

As an example of the strength of the improvement that this offers, observe that in order for the existing AVD-based results to yield a query time of $\SoftOh\big(1/\eps^{d/4}\big)$ the required space would be roughly $\Omega\big(n/\eps^{d/2}\big)$. The exponent in the space bound is nearly twice that given by Theorem~\ref{thm:ann-ub}, which arises by setting $\alpha = 2$. The connection between the polytope membership problem and ANN has been noted before by Clarkson~\cite{Clarkson-ANN}. Unlike Clarkson's, our results hold for point sets with arbitrary aspect ratios.


\medskip

Our data structure is based on a simple quadtree-based decomposition of space. Let $t$ denote the desired query time. We begin by preconditioning $K$ so that it is fat and has at most unit diameter. We then employ a quadtree that hierarchically subdivides space into hypercube cells. The decomposition stops whenever we can declare that a cell is either entirely inside or outside of $K$, or (if it intersects $K$'s boundary) it is locally approximable by at most $t$ halfspaces. This procedure, called {\alg} is presented in Section~\ref{sec:split-reduce}. Queries are answered by descending the quadtree to determine the leaf cell containing the query point, and (if not inside or outside) testing whether the query point lies within the approximating halfspaces. 

Although the algorithm itself is very simple, the analysis of its space requirements is quite involved. In Section~\ref{sec:firstbound}, we begin with a simple analysis, which shows that it is possible to obtain a significant improvement over the Dudley-based approach (in particular, reducing the exponent in the query time by  half with no increase in space). While this simple analysis introduces a number of useful ideas, it is neither tight nor does it provide space-time trade-offs. 

Our final analysis requires a deeper understanding of the local structure of the convex body's boundary. In Section~\ref{sec:cap} we introduce local surface patches of $K$'s boundary, called $\eps$-dual caps. We relate the data structure's space requirements to the existence of a low cardinality hitting set of the dual caps. We present a two-pronged strategy for generating such a hitting set, one focused on dual caps of large surface area (intuitively corresponding to boundary regions of low curvature) and the other focused on dual caps of small surface area (corresponding to boundary regions of high curvature). We show that simple random sampling suffices to hit dual caps of high surface area, and so the challenge is to hit the dual caps of low surface area. To do this, we show that dual caps of low surface area generate Voronoi patches on a hypersphere enclosing $K$ of large surface area. We refer to this result as the \emph{area-product bound}, which is stated in Lemma~\ref{lem:dual-basic}. This admits a strategy based on sampling points randomly on this hypersphere, and then projecting them back to their nearest neighbor on the surface of $K$. 

The area-product bound is proved with the aid of a classical concept from the theory of convexity, called the \emph{Mahler volume}~\cite{BoMi,Santalo}. The Mahler volume of a convex body is a dimensionless quantity that involves the product of the body's volume and the volume of its polar body. We demonstrate that dual caps and their Voronoi patches exhibit a similar polar relationship. The proof of the area-product bound is quite technical and is deferred to Section~\ref{sec:proof}. 

Armed with the area-product bound, in Section~\ref{sec:membership} we establish our final bound on the space-time trade-offs of {\alg}, which culminates in the proof of Theorem~\ref{thm:membership-ub}. In Section~\ref{sec:preproc} we present details on how the data structure is built and discuss preprocessing time. In Section~\ref{sec:lb} we establish the lower bound result, which is stated in Theorem~\ref{thm:lb}. 

Finally, in Section~\ref{sec:ann} we show how these results can be applied to improve the performance of approximate nearest neighbor searching in Euclidean space. It is well known that (exact) nearest neighbor searching can be reduced to vertical ray shooting to a polyhedron that results by lifting points in dimension $d$ to tangent hyperplanes for a paraboloid in dimension $d+1$ \cite{ray-shooting-NN,edels}. We show how to combine approximate vertical ray shooting (based on approximate polytope membership) with approximate Voronoi diagrams to establish Theorem~\ref{thm:ann-ub}.

\section{Preliminaries} \label{sec:prelim}

Throughout, we will use asymptotic notation to eliminate constant factors. In particular, for any positive real $x$, let $O(x)$ denote a quantity that is at most $c \P x$, for some constant $c$. Define $\Omega(x)$ and $\Theta(x)$ analogously. We will sometimes introduce constants within a local context (e.g., within the statement of single lemma). To simplify notation, we will often use the same symbol ``$c$'' to denote such generic constants. Recall that we use ``$\lg$'' to denote the base-2 logarithm. We will use ``$\log$'' when the base does not matter. Some of our search algorithms involve integer grids, and for these we assume a model of computation that supports integer division. 

Let $K$ denote a full-dimensional convex body in $\RE^d$, and let $\partial K$ denote its boundary. For concreteness, we assume that $K$ is represented as the intersection of $n$ closed halfspaces. Our data structure can generally be applied to any representation that supports access primitives (i)--(iii) given at the start of Section~\ref{sec:split-reduce}.

\subsection{Absolute and Relative Approximations} \label{sec:prelim-abs}

Earlier, we defined approximation relative to $K$'s diameter, but it will be convenient to define the approximation error in absolute terms. Given a positive real $r$, define $K \oplus r$ to be set of points that lie within Euclidean distance $r$ of $K$. We say that a polytope $P$ is an \emph{absolute $\eps$-approximation} of a convex body $K$ if 
\[
	K ~\subseteq~ P ~\subseteq~ K \oplus \eps.
\]
When we wish to make the distinction clear, we refer to the definition in the introduction as a \emph{relative approximation}. Henceforth, unless otherwise stated approximations are in the absolute sense.

In order to reduce the general approximation problem into a more convenient absolute form, we will transform $K$ into a ``fattened'' body of bounded diameter. Given a parameter $0 < \gamma \le 1$, we say that a convex body $K$ is \emph{$\gamma$-fat} if there exist concentric Euclidean balls $B$ and $B'$, such that $B \subseteq K \subseteq B'$, and $\radius(B)/\radius(B') \ge \gamma$. We say that $K$ is \emph{fat} if it is $\gamma$-fat for a constant $\gamma$ (possibly depending on $d$, but not on $n$ or $\eps$). The following lemma shows that $K$ can be fattened without significantly altering the approximation parameter. Let $Q_0^{(d)}$ denote the $d$-dimensional axis-aligned hypercube of unit diameter centered at the origin. When $d$ is clear, we refer to this as $Q_0$.

\begin{lemma} \label{lem:fat}
Given a convex body $K$ in $\RE^d$ and $0 < \eps \le 1$, there exists an affine transformation $T$ such that $T(K)$ is $(1/d)$-fat and $T(K) \subseteq Q_0$. If $P$ is an absolute $(\eps/d \sqrt{d})$-approximation of $T(K)$, then $T^{-1}(P)$ is a relative $\eps$-approximation of $K$.
\end{lemma}

We omit the proof of this lemma for now, since it is subsumed by Lemma~\ref{lem:precondition-1} below. Our approach will be to map $K$ to $T(K)$, set $\eps' \gets \eps/d \sqrt{d}$, and then apply an absolute $\eps'$-approximation algorithm to $T(K)$ (or more accurately, to the result of applying $T$ to each of $K$'s defining halfspaces). Since $\eps'$ is within a constant factor of $\eps$, the asymptotic complexity bounds that we will prove for the absolute case will apply to the original (relative) approximation problem case as well.

\subsection{Concepts from Quadtrees} \label{sec:prelim-quadtree}

By the above reduction, it suffices to consider the problem of computing an absolute $\eps$-approxi\-mation to a fat convex body $K$ that lies within $Q_0$. Our construction will be based on a quadtree decomposition of $Q_0$. More formally, we define a \emph{quadtree cell} by the following well known recursive decomposition. $Q_0$ is a quadtree cell, and given any quadtree cell $Q$, each of the $2^d$ hypercubes that result by bisecting each of $Q$'s sides by an axis-orthogonal hyperplane is also a quadtree cell. A cell $Q'$ that results from subdividing $Q$ is a \emph{child} of $Q$. Clearly, the child's diameter is half that of its parent. The subdivision process defines a $(2^d)$-ary tree whose nodes are quadtree cells, and whose leaves are cells that are not subdivided.

It will be useful to define a notion of approximation that is local to a quadtree cell $Q$. An obvious definition would be to approximate $K \cap Q$. The problem with this is that a point $p \in Q$ that is close to $K$ need not be close to $K \cap Q$ (see Figure~\ref{fig:extended-region}(a)). To remedy this we say that a polytope $P$ is an \emph{$\eps$-approximation of $K$ within} $Q$ if
\[
	K \cap Q
		~ \subseteq ~ P \cap Q
		~ \subseteq ~ (K \oplus \eps) \cap Q
\]
(see Figure~\ref{fig:extended-region}(b)). This definition implies that for any query point $q \in Q$, we can correctly answer $\eps$-approximate polytope membership queries with respect to $K$ by checking whether $q \in P$. We do not care what happens outside of $Q$, and indeed $P$ may even be unbounded.

\begin{figure}[htbp]
  \centerline{\includegraphics[scale=0.40]{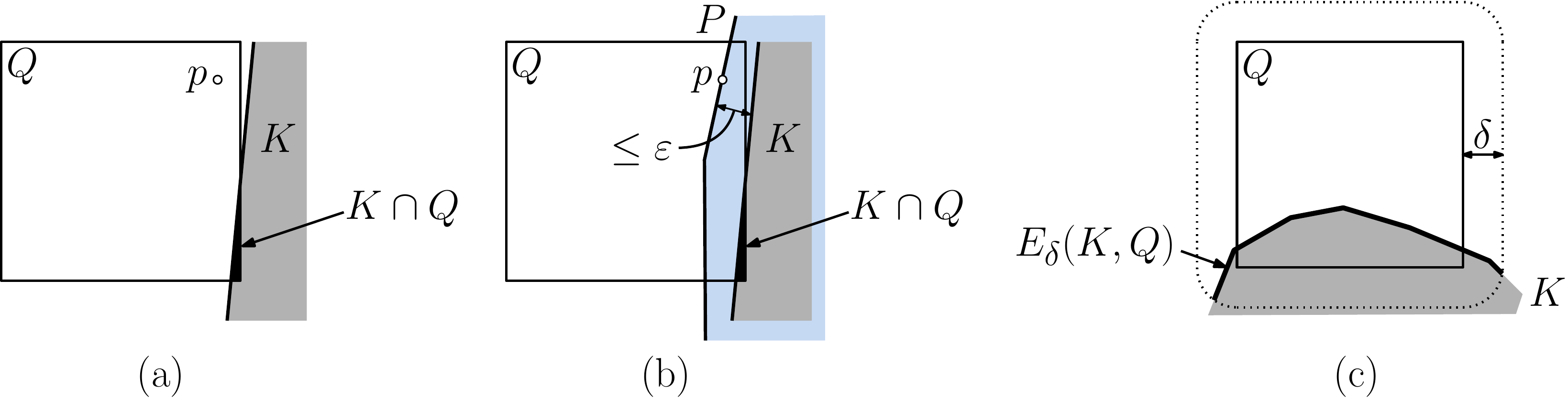}}
  \caption{(a) and (b): $\eps$-approximation of $K$ within $Q$ and (c): $E_{\delta}(K,Q)$.}
  \label{fig:extended-region}
\end{figure}

As we shall see later, computing an $\eps$-approximation of $K$ within a quadtree cell $Q$, will generally require that we consider $\partial K$ in a region that extends slightly beyond $Q$. We define $E_{\delta}(K,Q)$ to be the portion of $\partial K$ that lies within distance $\delta$ of $Q$ (see Figure~\ref{fig:extended-region}(c)). Because $\delta = \sqrt{\eps}$ will be of particular interest, we use $E(K,Q)$ as a shorthand for $E_{\sqrt{\eps}}(K,Q)$.

In order to apply constructions on quadtree cells of various sizes, it will be convenient to transform all such constructions into a common form. Given a quadtree cell $Q$, we define \emph{standardization} to be the application of an affine transformation that uniformly scales and translates space so that $Q$ is aligned with the standard quadtree cell $Q_0$. We transform $K$ using this same transformation, and apply the same scale factor to $\eps$. Although we assume that the input body is contained within $Q_0$, after standardization, the transformed image of $K$ need not be contained within $Q_0$.

\subsection{Polarity and the Mahler Volume} \label{sec:prelim-polar}

Some of our analysis will involve the well known concept of \emph{polarity}. Let us recall some general facts (see, e.g., Eggleston~\cite{Egg58}). Given vectors $u, v \in \RE^d$, let $\ang{u,v}$ denote their inner product, and let $\|v\| = \sqrt{\ang{v,v}}$ denote $v$'s Euclidean length. Given a convex body $K \in \RE^d$ define its \emph{polar} to be the convex set
\[
	\polar{K}
		~ = ~ \{ u : \ang{u,v} \le 1, \hbox{~for all $v \in K$} \}.
\]
If $K$ contains the origin then $\polar{K}$ is bounded. Given $v \in \RE^d$, $\polar{v}$ is simply the closed halfspace that contains the origin whose bounding hyperplane is orthogonal to $v$ and at distance $1/\|v\|$ from the origin (on the same side of the origin as $v$). The polar has the inclusion-reversing property that $v$ lies within $\polar{u}$ if and only if $u$ lies within $\polar{v}$. We may equivalently define $\polar{K}$ as the intersection of $\polar{v}$, for all $v \in K$. 

Generally, given $r > 0$, define 
\[
	\polarX{r}{K}
		~ = ~ \{ u : \ang{u,v} \le r^2, \hbox{~for all $v \in K$} \}.
\]
It is easy to see that for any $v \in \RE^d$, $\polarX{r}{v}$ is the closed halfspace at distance $r^2/\|v\|$ (see Figure~\ref{fig:polar}(a)). Thus, $\polarX{r}{K}$ is a uniform scaling of $\polar{K}$ by a factor of $r^2$. In particular, if $B$ is a Euclidean ball of radius $x$ centered at the origin, then $\polarX{r}{B}$ is a concentric ball of radius $r^2/x$.

\begin{figure}[htbp]
  \centerline{\includegraphics[scale=0.40]{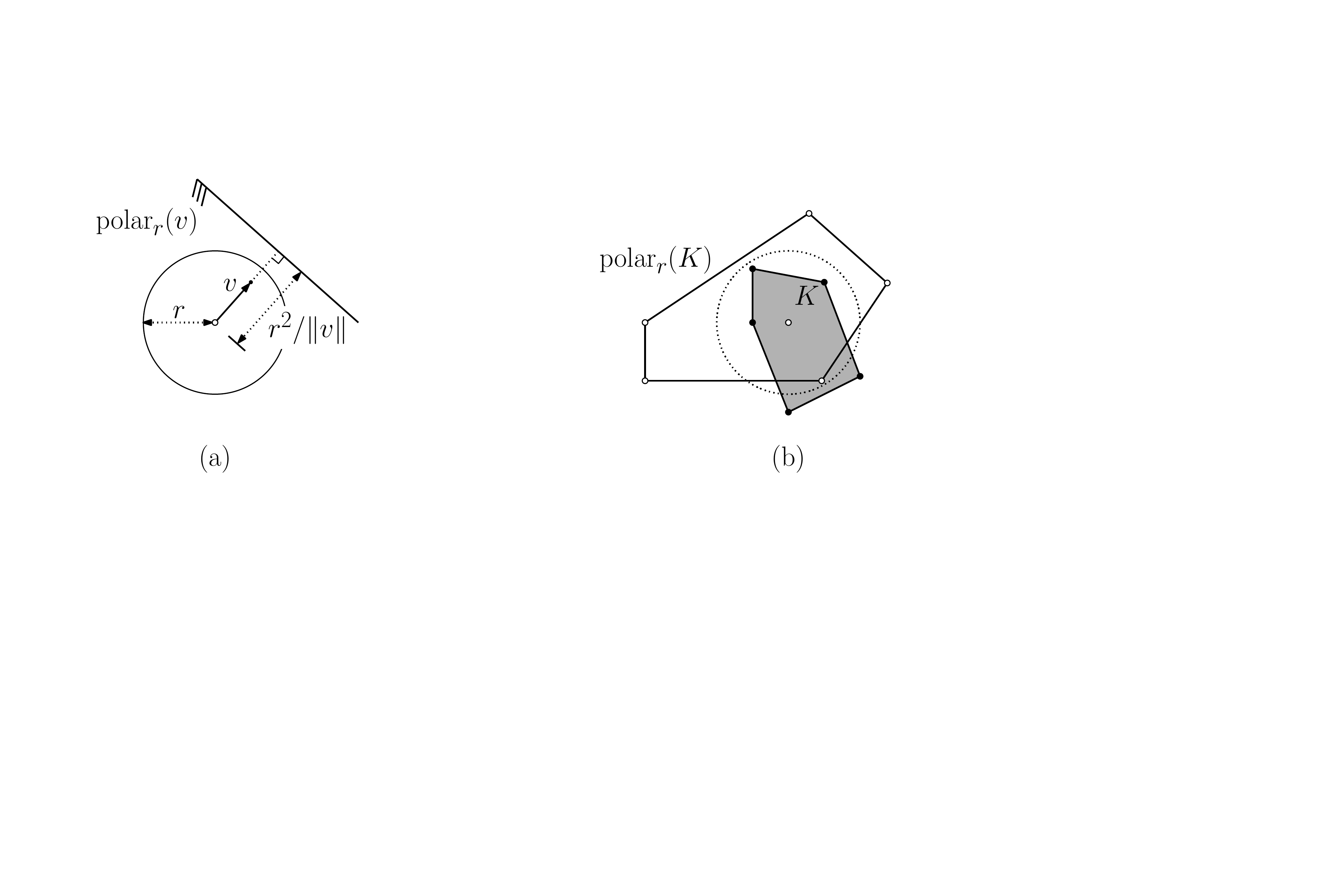}}
  \caption{The generalized polar transform and polar body.}
  \label{fig:polar}
\end{figure}

An important concept related to polarity is the \emph{Mahler volume}, which is defined to be the product of the volumes of a convex body and its polar. There is a large literature on the Mahler volume, mostly for centrally symmetric bodies. Later in the paper we will make use of the following bound on the Mahler volume for arbitrary convex bodies (see, e.g., Kuperberg~\cite{Kuperberg}). Given a convex body $K$ in $\RE^d$, let $\vol(K)$ denote its volume, or more formally, its $d$-dimensional Hausdorff measure.

\begin{lemma}[Mahler volume] \label{lem:mahler}
There is a constant $c_m$ depending only on $d$, such that given a convex body $K$ in $\RE^d$, $\vol(K) \cdot \vol(\polar{K}) \ge c_m$. More generally, given $r > 0$, $\vol(K) \cdot \vol(\polarX{r}{K}) \ge c_m \P r^{2 d}$.
\end{lemma}

\subsection{Simple Approximation Trade-off} \label{sec:prelim-simple-approx}

Before presenting our results, it will be illuminating to see how to obtain simple data structures for approximate polytope membership by combining two existing approximation methods. Let us begin by describing Dudley's approximation. Assuming that $K$ is contained within $Q_0$, let $S$ denote the $(d-1)$-dimensional sphere of radius $3$ centered at the origin, which we call the \emph{Dudley hypersphere}. (The value $3$ is not critical; any sufficiently large constant suffices.) For $\delta > 0$, a set $\Sigma$ of points on $S$ is said to be \emph{$\delta$-dense} if every point of $S$ lies within distance $\delta$ of some point of $\Sigma$. Let $\Sigma$ be a $\sqrt{\eps}$-dense set of points on $S$ (see Figure~\ref{fig:dudley-bentley}(a)). By a simple packing argument there exists such a set of cardinality $\Theta\big(1/\eps^{(d-1)/2}\big)$. For each point $x \in \Sigma$, let $x_0$ be its nearest point on $K$'s boundary. For each such point $x_0$, consider the halfspace containing $K$ that is defined by the supporting hyperplane passing through $x_0$ that is orthogonal to the line segment $\overline{x x}_0$. Dudley shows that the intersection of these halfspaces is an outer $\eps$-approximation of $K$. We can answer approximate membership queries by testing whether $q$ lies within all these halfspaces (by brute force). This approach takes $O\big(1/\eps^{(d-1)/2}\big)$ query time and space.

\begin{figure}[htbp]
  \centerline{\includegraphics[scale=0.33]{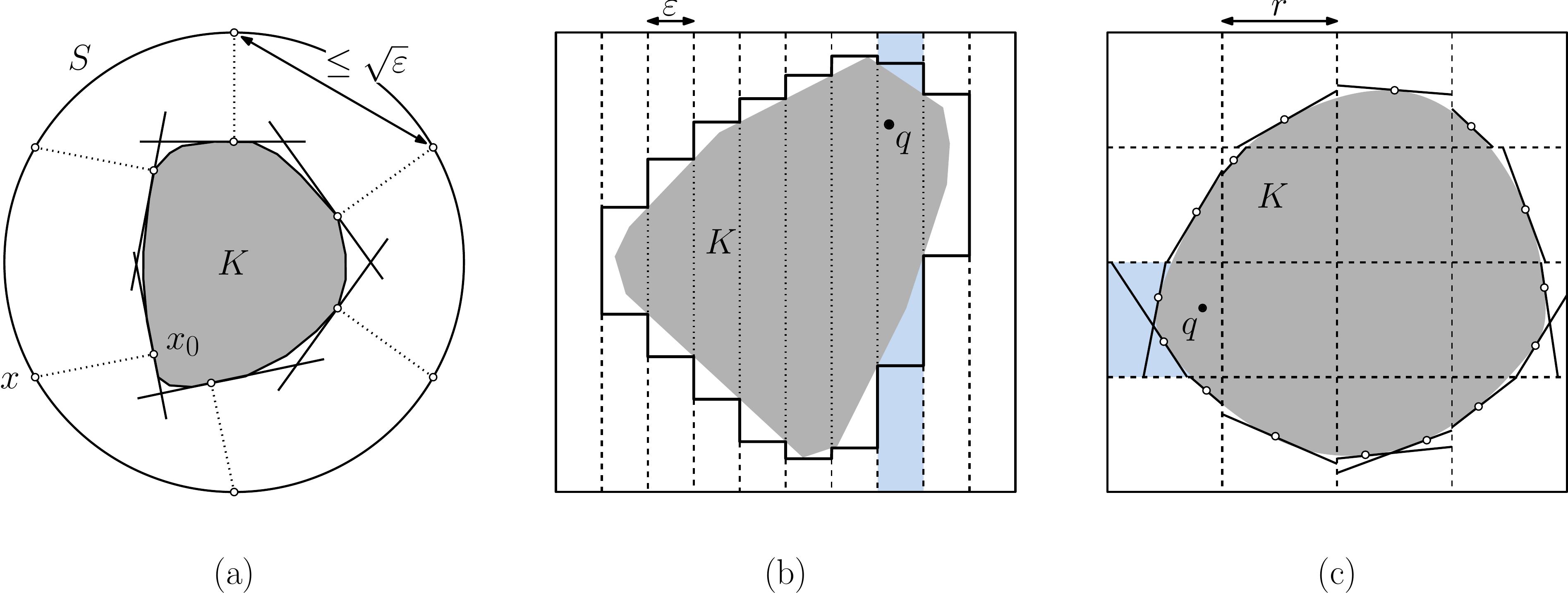}}
  \caption{The $\eps$-approximations of (a)~Dudley and (b)~Bentley {\etal}, and (c) the simple trade-off. (Not drawn to scale.)}
  \label{fig:dudley-bentley}
\end{figure}

An alternative solution is related to a grid-based approximation by Bentley~{\etal}~\cite{BFP}. Again, we assume that $K$ is contained within $Q_0$. For the sake of illustration, let us think of the $d$-th coordinate axis as pointing upwards. Partition the upper facet of $Q_0$ into a $(d-1)$-dimensional square grid with cells of diameter $\eps$. A packing argument implies that the number of cells is $O(1/\eps^{d-1})$. Extend each of these cells downwards to form a subdivision of $Q_0$ into vertical columns (see Figure~\ref{fig:dudley-bentley}(b)). Trim each column at the highest and lowest points at which it intersects $K$. Together, these trimmed columns define a collection of hyperrectangles whose union contains $K$. The resulting data structure has $O(1/\eps^{d-1})$ space. Given a query point $q$, in $O(1)$ time we can determine the vertical column containing $q$ (assuming a model of computation that supports integer division), and we then test whether $q$ lies within the trimmed column. In contrast to the method based on Dudley's construction, this method provides a better query time of $O(1)$ but with higher space of $O(1/\eps^{d-1})$.

It is possible to combine these two solutions into a simple unified approach that achieves a trade-off between space and query time. Given a parameter $r$, where $\eps \le r \le 1$, subdivide $Q_0$ into a grid of hypercube cells each of diameter $\Theta(r)$. For each cell $Q$ that intersects $K$'s boundary, apply Dudley's approximation to this portion of the polytope. By a straightforward packing argument, the number of grid cells that intersect $K$'s boundary is $O(1/r^{d-1})$ (see, for example, Lemma~{3} of \cite{ARS}). We apply standardization to $Q$ (thus mapping $Q$ to $Q_0$ and scaling $\eps$ to $\Omega(\eps/r)$) and apply Dudley's construction. By Dudley's results, the number of halfspaces needed to approximate $K$ within $Q$ is $O\big((r/\eps)^{(d-1)/2}\big)$. To answer a query, in $O(1)$ time we determine which hypercube of the grid contains the query point (assuming a model of computation that supports integer division). We then apply brute-force search to determine whether the query point lies within all the associated halfspaces in time $O\big((r/\eps)^{(d-1)/2}\big)$. The query time is dominated by this latter term. The space is dominated by the total number of halfspaces, which is $O\big((1/r^{d-1}) \cdot (r/\eps)^{(d-1)/2}\big) = O\big(1/(\eps r)^{(d-1)/2}\big)$. If we express $r$ in terms of a parameter $\alpha$, where $r = \eps^{1 - 2/\alpha}$, then Theorem~\ref{thm:simple-trade-off} follows as an immediate consequence. Note that the resulting trade-off interpolates nicely between the two extremes for $\eps \le r \le 1$.

\section{The Data Structure and Construction} \label{sec:split-reduce}

In this section we show how to improve the approach from the previous section by replacing the grid with a quadtree. The data structure is constructed by the recursive procedure, called {\alg}, whose inputs consist of a convex body $K$ and a quadtree cell $Q$. We are also given the approximation parameter $0 < \eps \le 1$, and a parameter $t \ge 1$ that controls the query time. Although we assume that $K$ is presented as the intersection of $n$ halfspaces, this procedure can be applied to any representation that supports the following access primitives:
\begin{enumerate}
\item[(i)] Determine whether $Q$ is disjoint from $K$.

\item[(ii)] Determine whether $Q$ is contained within $K \oplus \eps$.

\item[(iii)] Determine whether there exists a set of at most $t$ halfspaces whose intersection $\eps$-approximates $K$ within $Q$, and if so generate such a set.
\end{enumerate}

Recall that we assume that $K$ has been transformed so it is $(1/d)$-fat and lies within $Q_0$ (the hypercube of unit diameter centered at the origin). The data structure is built by the call $\alg(K,Q_0)$. In general, $\alg(K,Q)$ checks whether any of the above access primitives returns a positive result, and if so it terminates the decomposition and $Q$ is declared a \emph{leaf cell}. Otherwise, it makes a recursive call on the children of $Q$ (see Figure~\ref{fig:split-reduce}(a)). On termination, each leaf cell is labeled as either ``inside'' or ``outside'' or is associated with a set of at most $t$ approximating halfspaces (see Figure~\ref{fig:split-reduce}(b)).

\begin{figure}[htbp]
  \centerline{\includegraphics[scale=0.40]{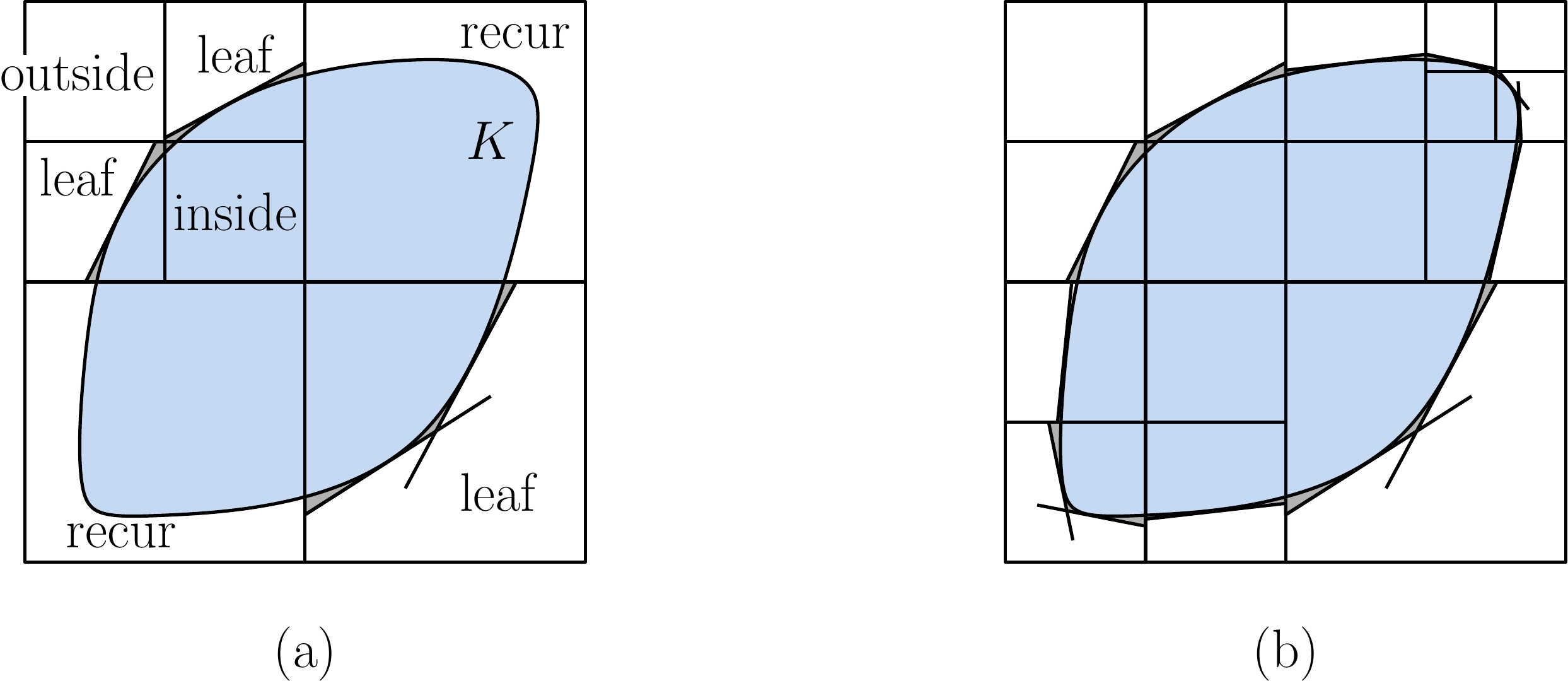}}
  \caption{(a) Cases arising in {\alg} for $t = 2$ and (b) the final subdivision.}
  \label{fig:split-reduce}
\end{figure}

\newcommand{\stepout}{1}  
\newcommand{\stepin}{2}
\newcommand{\stepapx}{3}
\newcommand{\steprecur}{4}
%
  \noindent{\alg}$(K,Q)$:
  \begin{enumerate}
  \setlength{\itemsep}{-0.5ex}%
  \setlength{\parsep}{0pt}%
	\item[(\stepout)] If $Q \cap K = \emptyset$,  label $Q$ as ``outside.''
	
	\item[(\stepin)] If $Q \subseteq K \oplus \eps$, label $Q$ as ``inside.''
  
	\item[(\stepapx)] If there exists a set at most $t$ halfspaces whose intersection provides an $\eps$-approximation to $K$ within $Q$, associate $Q$ with such a set $P(Q)$ of minimum size.
	
  \item[(\steprecur)] Otherwise, split $Q$ into $2^d$ quadtree cells and recursively invoke {\alg} on each. 
  \end{enumerate}

For the sake of our space-time trade-offs, we will usually assume that $t$ is reasonably large, say, $t = \Omega(\log \inv{\eps})$. Under our assumption that $t \ge 1$, steps (\stepout) and (\stepin) are not needed, since it is possible to $\eps$-approximate any cell satisfying these conditions with a single halfspace. This assumption on the value of $t$ is mainly a convenience to simplify the formulas of our mathematical analysis. Observe that even if $t = 0$, the procedure will terminate and provide a correct answer once the cell diameter falls below $\eps$.

It is easy to see that the recursion must terminate as soon as $\diam(Q) \le \eps$ (since, irrespective of whether it intersects $\partial K$, any such cell can be labeled either as ``inside'' or ``outside''). Of course, it may terminate much sooner. Since $Q_0$ is of unit diameter, it follows that the height of the quadtree is $O(\log \inv{\eps})$. The total space used by the data structure is the sum of the space needed to store the quadtree and the space needed to store the approximating halfspaces for the cells that intersect $K$'s boundary. Our next lemma shows that if the query time is sufficiently large, the latter quantity dominates the space asymptotically. For each leaf cell $Q$ generated by Step~(\stepapx), define $t(Q) = |P(Q)|$, and define $t(Q) = 1$ for all the other leaf cells.

\begin{lemma} \label{lem:total-space}
Given a convex body $K \subseteq Q_0$ in $\RE^d$. If $t$ is $\Omega(\log \inv{\eps})$, then the total space of the data structure produced by $\alg(K,Q_0)$ on $K$ for query time $t$ is asymptotically dominated by the sum of $t(Q)$ over all the leaf cells $Q$ that intersect $K$'s boundary.
\end{lemma}

\begin{proof}
Let $T$ denote the quadtree produced by running {\alg} on $K$. As mentioned above, $T$ is of height $O(\log \inv{\eps} )$. By our hypothesis that $t$ is $\Omega(\log \inv{\eps} )$, there exists a constant $c$ such that $\textrm{height}(T) \le c \P t$. Let $L$ denote the set of leaves of $T$ that intersect the boundary of $K$, and let $M$ denote the internal nodes of $T$ that have the property that all their children are leaves. (These are the lowest internal nodes of the tree.) Let $t(L)$ denote the sum of $t(Q)$ over all $Q \in L$.

The fact that each node $u \in M$ was subdivided by {\alg} implies that more than $t$ halfspaces are needed to approximate $K$ within $u$'s cell. Therefore, the children of $u$ that intersect $K$'s boundary together require at least $t$ halfspaces. In addition to $P(Q)$, each quadtree leaf $Q$ can (implicitly) contribute its $2 d$ bounding hyperplanes to the approximation. Therefore, $t(L) + 2 d |L|$ hyperplanes suffice to approximate $K$ in all the cells of $M$, implying that $t(L) + 2 d |L| \ge t \cdot |M|$. Since $t(Q) \ge 1$, we have $t(L) \ge |L|$, and thus $(1 + 2 d) t(L) \ge t \cdot |M|$. 

Each internal node of $T$ is either in $M$ or is an ancestor of a node in $M$. Thus, the total number of internal nodes of $T$ is at most $|M| \cdot \textrm{height}(T)$. Since each internal node of a quadtree has $2^d$ children, the total number of nodes in the tree, excluding the root, is at most
\[
	2^d \cdot |M| \cdot \textrm{height}(T)
		~ \le ~ 2^d \cdot |M| \cdot (c \P t)
		~ \le ~ 2^d c  (1 + 2 d) t(L)
		~  =  ~ O(t(L)).
\]
Each internal node of $T$ and each leaf node that does not intersect $K$'s boundary contributes only a constant amount to the total space. Therefore, the space contribution of the nodes other than those of $L$ is at most a constant factor larger than the total number of nodes of $T$, which we have shown is $O(t(L))$. Therefore, the total space is $O(t(L))$, as desired.
\end{proof}

A query is answered by performing a point location in the quadtree to determine the leaf cell containing the query point. If the leaf cell is not labeled as being ``inside'' or ``outside'', we test whether the query point lies within all the associated halfspaces, and if so, we declare the point to be inside $K$. Otherwise it is declared to be outside. Clearly, the query time is $O(\log \inv{\eps} + t)$.

The algorithm is correct provided that the set of halfspaces $P(Q)$ computed in Step~(\stepapx) defines any $\eps$-approximation of $K$ within $Q$, but our analysis of the data structure's space requirements below (see the proof of Lemma~\ref{lem:restricted-dudley}) relies on the assumption that the size of $P(Q)$ is within a constant factor of the minimum number of halfspaces of any $\eps$-approximating polytope within $Q$. Unfortunately, we know of no constant-factor approximation to the problem of computing such a polytope. Thus, strictly speaking, the bounds stated in Theorem~\ref{thm:membership-ub} are purely existential. In Section~\ref{sec:preproc} we will show that through a straightforward modification of the greedy set-cover heuristic, it is possible to compute an approximation in which the number of defining halfspaces exceeds the optimum (for slightly smaller approximation parameter) by a factor of at most $\rho = O(\log \inv{\eps})$. From the following result it follows that this increases our space and query time bounds by $O(\log \inv{\eps})$.

\begin{lemma} \label{lem:weak-membership}
Given any $\rho \ge 1$ and any constant $0 < \beta \le 1$, if the number of halfspaces of $P(Q)$ computed in Step~(\stepapx) of {\alg} is within a factor $\rho$ of the minimum number of facets of any $(\beta \eps)$-approximating polytope within $Q$, then Theorems~\ref{thm:membership-ub} and~\ref{thm:ann-ub} hold but with the asymptotic space and query time bounds larger by a factor of $\rho$.
\end{lemma}

\begin{proof}
Let us refer to the hypothesized version of {\alg} whose Step~(\stepapx) is suboptimal as ${\alg}'$. Consider an execution of ${\alg}'$ using $\rho \kern1pt t$ as the desired query time and $\eps$ as the approximation parameter, and let us compare this to an execution of {\alg} using $t$ and $\beta \eps$, respectively. Since $\beta$ is a constant, the asymptotic dependencies on $\eps$ are unaffected, and therefore the space and query times stated in Theorems~\ref{thm:membership-ub} and~\ref{thm:ann-ub} apply without modification to the execution of {\alg}. In this execution, if the subdivision declares some quadtree cell $Q$ to be a leaf, then $t$ halfspaces suffice to $(\beta \eps)$-approximate $K$ within $Q$, and so by our hypothesis in the corresponding execution of ${\alg}'$, Step~(\stepapx) returns at most $\rho \kern1pt t$ halfspaces, implying this execution also declares $Q$ to be a leaf. Therefore, the tree generated by ${\alg}'$ is a subtree of the tree generated by {\alg}, but each leaf node may contain up to a factor of $\rho$ more halfspaces. Thus, the asymptotic space and query time bounds for ${\alg}'$ are larger than those of ${\alg} $ by this same factor.
\end{proof}

\section{Simple Upper Bound} \label{sec:firstbound}

In this section, we present a simple upper bound of $O\big(1/\eps^{(d-1)/2}\big)$ on the storage of the data structure obtained by the {\alg} algorithm for any given query time $t \geq 1/\eps^{(d-1)/4}$. The tools developed in this section will be useful for the more comprehensive upper bounds, which will be presented in subsequent sections.

Throughout this section we do not necessarily assume that $K$ has been scaled to lie within $Q_0$ and may generally be much larger. Recall that $S$ denotes a hypersphere of radius $3$ centered at the origin. Let $X$ denote a surface patch of $K$ that lies within $S$. Let $\Vor(X)$ denote the set of points exterior to $K$ whose closest point on $\partial K$ lies within $X$. We refer to the surface patch $\Vor(X) \cap S$ (the points of $S$ whose closest point on $\partial K$ lies within $X$) as the \emph{Voronoi patch} of $X$. Voronoi patches are related to Dudley's construction. In particular, a sample point $x \in S$ from Dudley's construction generates a supporting halfspace at a point of $X$ if and only if $x \in \Vor(X) \cap S$. The following two lemmas are straightforward adaptations of Dudley's analysis~\cite{Dudley}. The first is just a restatement of Dudley's result.

\begin{lemma} \label{lem:Dudley1}
Given a convex body $K$ in $\RE^d$ that lies within $Q_0$ and $0 < \eps \le 1$, there exists an $\eps$-approximating polytope $P$ bounded by at most $c / \eps^{(d-1)/2}$ facets, where $c$ is a constant depending only on $d$.
\end{lemma}

The second lemma is a technical result that is implicit in Dudley's analysis. Given two points $x, y \in \RE^d$, let $\overline{x y}$ denote the segment between them, and let $\|x y\|$ denote the Euclidean length of this segment.

\begin{lemma} \label{lem:Dudley}
Let $K$ be a convex body, let $0 < \eps \le 1$, and let $z$ and $x$ be two points of $S$ such that $\|z x\| \le \sqrt{\eps} / 4$. Let $z_0$ and $x_0$ be the points on $\partial K$ that are closest to $z$ and $x$, respectively. If $z_0$ is within unit distance of the origin, then
\begin{enumerate}
\item[$(i)$] $\|z_0 x_0\| \leq \sqrt{\eps} / 4$ and

\item[$(ii)$] the supporting hyperplane at $x_0$ orthogonal to the segment $\overline{x x}_0$ intersects segment $\overline{z z}_0$ at distance less than $\eps$ from $z_0$ (see Figure~\ref{fig:dudley-lemma}).
\end{enumerate}
\end{lemma}

\begin{figure}[htbp]
  \centerline{\includegraphics[scale=0.40]{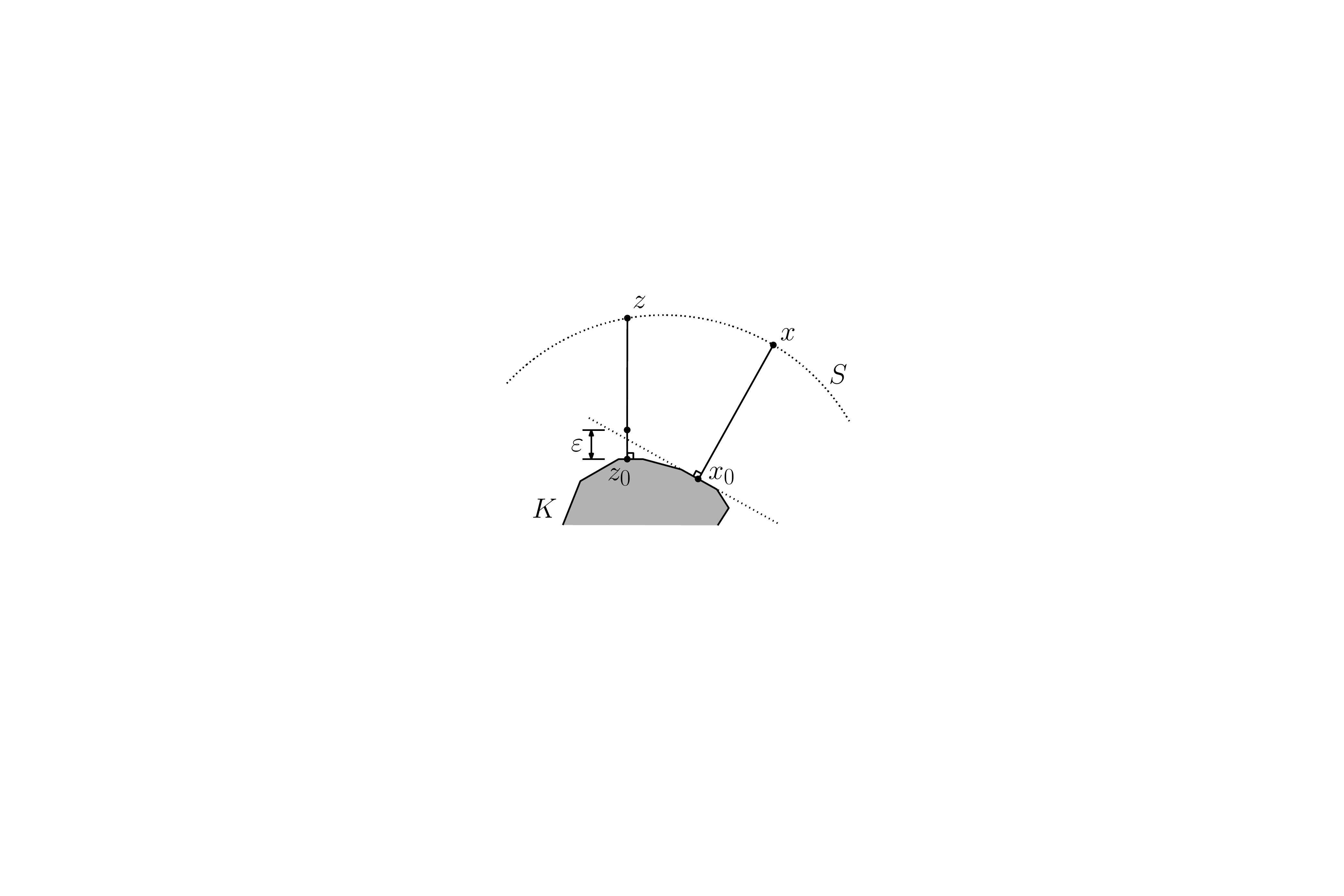}}
  \caption{Lemma~\ref{lem:Dudley}.}
  \label{fig:dudley-lemma}
\end{figure}

The following lemma is an extension of Dudley's results, which allows us to bound the complexity of an $\eps$-approximation of $K$ within a quadtree cell $Q$. Recall from Section~\ref{sec:prelim-quadtree} that $E(K,Q)$ denotes portion of $\partial K$ that lies within distance $\sqrt{\eps}$ of $Q$. 

\begin{lemma} \label{lem:restricted-dudley}
Let $K$ be a convex body, $Q \subseteq Q_0$ be a quadtree cell that intersects $\partial K$, and $0 < \eps \le 1/2$. Let $\Sigma$ denote a set of $(\sqrt{\eps}/4)$-dense points on the Dudley sphere $S$. Then $t(Q) \le |\Sigma \cap \Vor(E(K,Q))|$ (see Figure~\ref{fig:restricted-dudley}(a)).
\end{lemma}

\begin{proof}
We construct an approximating polytope $P$ by the following local variant of Dudley's construction. For each point $x \in \Sigma \cap \Vor(E(K,Q))$, let $x_0$ be its nearest point on the boundary of $K$. (Note that $x_0 \in E(K,Q)$.) For each point $x_0$, take the supporting halfspace to $K$ passing through $x_0$ that is orthogonal to the segment $\overline{x x}_0$. Let $P$ be the (possibly unbounded) intersection of these halfspaces.

\begin{figure}[htbp]
  \centerline{\includegraphics[scale=0.40]{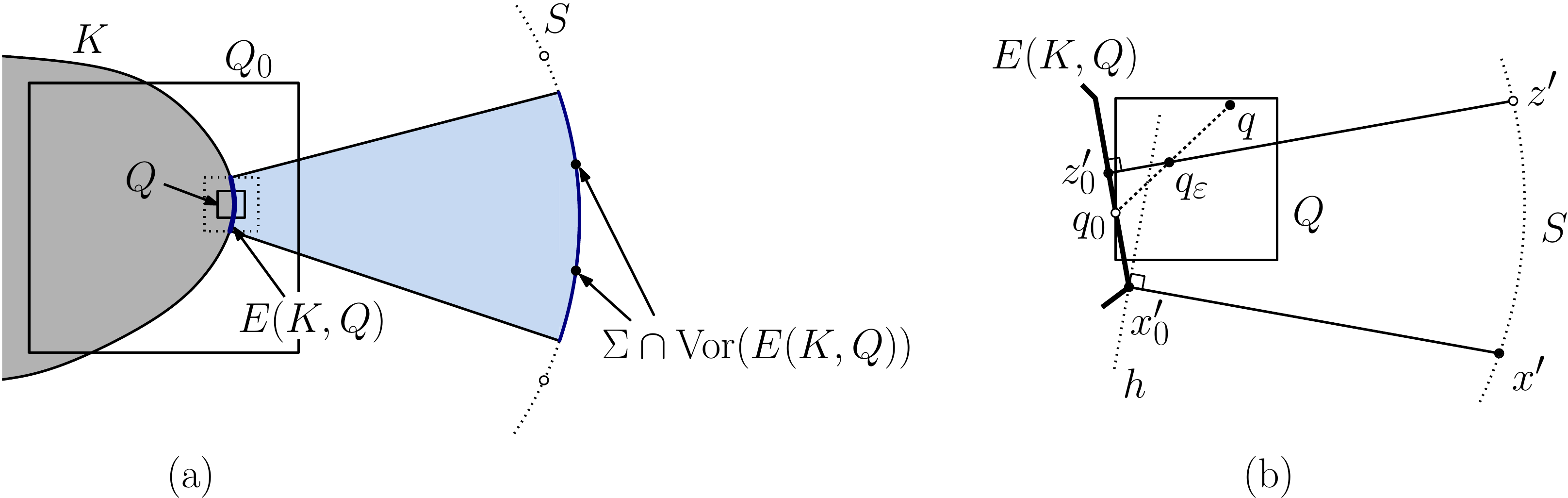}}
  \caption{Lemma~\ref{lem:restricted-dudley}. (Not drawn to scale.)}
  \label{fig:restricted-dudley}
\end{figure}

First, we show that $\Sigma \cap \Vor(E(K,Q))$ is nonempty. Consider any point $z_0$ on $\partial K \cap Q$. Let $z$ denote any point of $S \cap \Vor(E(K,Q))$ whose closest point on $\partial K$ is $z_0$. By definition of $\Sigma$, there is a point $x \in \Sigma$ whose distance from $z$ is at most $\sqrt{\eps}/4$. Letting $x_0$ denote $x$'s closest point on $\partial K$, by Lemma~\ref{lem:Dudley}(i), $\|z_0 x_0\| \le \sqrt{\eps}/4 < \sqrt{\eps}$. Thus, $x_0$ lies within $E(K,Q)$, which implies that $x \in \Sigma \cap \Vor(E(K,Q))$. It follows that $P$ is bounded by at least one halfspace.

We now show that $P$ is an (outer) $\eps$-approximation of $K$ within $Q$. Since $P$ is defined by supporting hyperplanes, $K$ is contained within $P$. Consider any $q \in Q$ that is at distance greater than $\eps$ from $K$. It suffices to show that $q \notin P$, that is, there exists a bounding hyperplane for $P$ that separates $q$ from $K$. Let $q_0$ denote the point of $K \cap Q$ that is closest to $q$ (see Figure~\ref{fig:restricted-dudley}(b)). Note that $q_0$ is constrained to lie within $Q$, and hence this may not be the closest point to $q$ on $\partial K$. By continuity, there must be a point on the segment $\overline{q q}_0$ that is at distance exactly $\eps$ from $\partial K$, which we denote by $q_{\eps}$. Since $Q$ is convex, this segment is contained in $Q$, and, hence, so is $q_{\eps}$. 

Let $z'_0$ be the point on $\partial K$ that is closest to $q_{\eps}$. (Note that $z'_0$ need not lie within $Q$.) Because $Q_0$ is centered at the origin, $z'_0$'s distance from the origin is at most $\diam(Q_0)/2 + \|q_{\eps} z'_0\| \le 1/2 + \eps \le 1$. Let $z'$ denote the point of intersection with the Dudley hypersphere $S$ of the ray emanating from $z'_0$ and passing through $q_{\eps}$. Let $x'$ be a point of $\Sigma$ that lies within distance $\sqrt{\eps} / 4$ of $z'$, and let $x'_0$ be its closest point on $\partial K$. By Lemma~\ref{lem:Dudley}(i) $\|x'_0 z'_0\| \le \sqrt{\eps} / 4$, and (ii) the supporting hyperplane $h$ at $x'_0$ orthogonal to the segment $\overline{x' x'_0}$ intersects segment $\overline{z' z'_0}$ at distance less than $\eps$ from $z'_0$. Thus, $h$ separates $q_{\eps}$ from $K$, and therefore it separates $q$ from $K$.

To complete the proof that $q \notin P$, it suffices to show that $h$ is indeed included in our construction of $P$. By the triangle inequality and our assumption that $\eps \le 1/2$, the distance from $x'_0$ to $Q$ is at most
\[
	\|x'_0 z'_0\| + \|z'_0 q_{\eps}\|
		~ \le ~ \frac{\sqrt{\eps}}{4} + \eps
		~ \le ~ \sqrt{\eps}.
\]
It follows that $x'_0 \in E(K,Q)$, and so $h$ is included in the construction of $P$. By our hypothesis that the set $P(Q)$ constructed in Step~(\stepapx) of ${\alg}$ is the minimum-sized set of halfspaces needed to $\eps$-approximate $K$ within $Q$, we have $t(Q) = |P(Q)| \le |P| = |\Sigma \cap \Vor(E(K,Q))|$. (Note that this works even if $Q$ is an ``inside'' cell that intersects $K$'s boundary. In such a case $t(Q) = 1$ by definition, and as argued above, $\Sigma \cap \Vor(E(K,Q))$ is nonempty.) This completes the proof.
\end{proof}

Next, we prove a useful technical lemma, which bounds the total complexity of a set of leaves whose cells are of a given minimum size. Recalling the definition of $\Sigma$ from the previous lemma, we may assume that $|\Sigma| = \Theta\big(1/\eps^{(d-1)/2}\big)$.

\begin{lemma} \label{lem:space-bound}
Let $K$ be a convex body in $\RE^d$, let $0 < \eps \le 1/2$, and let $L$ denote a set of disjoint quadtree cells contained within $Q_0$ such that each intersects $\partial K$ and is of diameter $\Omega(\sqrt{\eps})$. Then $\sum_{Q \in L} t(Q) = O\big(1/\eps^{(d-1)/2}\big)$.
\end{lemma}

\begin{proof}
By applying Lemma~\ref{lem:restricted-dudley} to each $Q \in L$ we have
\[
  \sum_{Q \in L} t(Q) ~\le~ \sum_{Q \in L} |\Sigma \cap \Vor(E(K,Q))|.
\]
Since $|\Sigma| = O\big(1/\eps^{(d-1)/2}\big)$, to complete the proof it suffices to show that each $x \in \Sigma$ lies within $\Vor(E(K,Q))$ for at most constant number of $Q \in L$. To see this, let $x_0$ be the point on $\partial K$ that is closest to $x$. Since each cell $Q \in L$ has size at least $\Omega(\sqrt{\eps})$, by disjointness and a packing argument it follows that at most a constant number (depending on dimension) of such cells can lie within distance $\sqrt{\eps}$ of $x_0$, which establishes the claim.
\end{proof}

Combining the above results, we obtain the main result of this section.

\begin{lemma} \label{lem:firstboundunit}
Let $K$ be a convex body in $\RE^d$ and $0 < \eps \le 1/2$. The output of $\alg(K,Q_0)$ for $t \geq 1/\eps^{(d-1)/4}$ has total space $O\big(1/\eps^{(d-1)/2}\big)$.
\end{lemma}

\begin{proof}
Let $c_2$ be the constant of Lemma~\ref{lem:Dudley1}, and define $c_1 = (1/c_2)^{2/(d-1)}$. We may assume that $\eps \le c_1^2$, for otherwise $\eps = \Omega(1)$ and clearly ${\alg}$ will not generate more than a constant number of cells. 

Let $T$ denote the quadtree produced by the algorithm, and let $L$ denote the set of leaf cells of $T$ that intersect the boundary of $K$. Recall from Lemma~\ref{lem:total-space} that the data structure's total space is asymptotically bounded by the sum of $t(Q)$ for all $Q \in L$. Thus, it suffices to prove that
\[
  \sum_{Q \in L} t(Q) ~ = ~ O\big(1/\eps^{(d-1)/2}\big).
\]

Towards this end, we first prove a lower bound on the size of any leaf cell $Q$. We assert that the cell $Q$ associated with any internal node has diameter at least $\delta = c_1 \sqrt{\eps}$. It will then follow that each leaf cell has diameter at least $\delta/2$. Suppose to the contrary that $\diam(Q) < \delta$. Recall the standardization transformation from Section~\ref{sec:prelim-quadtree}, which maps $Q$ to $Q_0$ and scales $\eps$ to at least $\eps/\delta = \sqrt{\eps}/c_1$. Let us denote this value by $\eps'$. Since $\eps \le c_1^2$, we have $\eps' \le 1$. By applying Lemma~\ref{lem:Dudley1} to the transformed body (with $\eps'$ playing the role of $\eps$), it follows that the polytope $K \cap Q$ can be $\eps$-approximated by a polytope $P$ defined by the intersection of at most 
\[
  \frac{c_2}{(\eps')^{(d-1)/2}} 
	~  =  ~ c_2 \left( \frac{c_1}{\sqrt{\eps}} \right)^{\kern-2pt (d-1)/2}
	~ \le ~ \inv{\eps^{(d-1)/4}}
\] 
halfspaces. Since $K \cap Q \subseteq P$, it is easy to see that $P$ is an $\eps$-approximation of $K$ within $Q$. Since $t \ge 1/\eps^{(d-1)/4}$, the termination condition of our algorithm implies that such a cell is not further subdivided, contradicting our hypothesis that this is an internal node. Therefore, the cells of $L$ satisfy the conditions of Lemma~\ref{lem:space-bound}. The desired bound follows by applying this lemma.
\end{proof}

It is useful to contrast this with the Dudley-based approach described in Section~\ref{sec:prelim-simple-approx}. For $t = 1/\eps^{(d-1)/4}$, we obtain the same $O(1/\eps^{(d-1)/2})$ space in each case, but the exponent in the query time of {\alg} is only half that of the Dudley-based approach. Later, in Lemma~\ref{lem:baseunit}, we will present a more refined analysis showing that it is possible to reduce this further, achieving a query time of only $\SoftOh\big(1/\eps^{(d-1)/8}\big)$.

It will be useful in later sections to generalize the above lemma to quadtree cells of arbitrary size. By a direct application of standardization, we obtain the following.

\begin{lemma} \label{lem:firstbound}
Let $K$ be a convex body in $\RE^d$, $Q$ be a quadtree cell contained within $Q_0$, and let $0 < \eps \le \diam(Q)/2$. The output of the call $\alg(K,Q)$ for $t \geq (\diam(Q)/\eps)^{(d-1)/4}$ has total space $O\big((\diam(Q)/\eps)^{(d-1)/2}\big)$.
\end{lemma}

\section{Dual Caps and Approximation} \label{sec:cap}

The bounds proved in the previous section apply to query times $t \geq 1/\eps^{(d-1)/4}$. In Section~\ref{sec:membership} we will show how to obtain good space bounds for smaller query times. This will involve analyzing the local geometry about small boundary patches of the convex body. In this section, we introduce the principal geometric underpinnings that will be needed for this more refined analysis. In particular, we discuss the concepts of dual caps and restricted dual caps and their role in polytope approximation.

Although we do not assume that $K$ is smooth, it will simplify the presentation to imagine that each boundary point has a unique supporting hyperplane and a unique normal vector. To achieve this, we employ an \emph{augmented representation} of the boundary points of $K$. In particular, each boundary point $p \in \partial K$ will be expressed as a pair $(p,h)$, where $h$ is a supporting hyperplane at $p$. We will often refer to $h$ as $h(p)$. When $h$ is clear from context or unimportant, we avoid explicit reference to it.

We first observe that computing an outer $\eps$-approximation of a convex body $K$ by halfspaces can be reduced to a hitting-set problem. Consider any point $p_{\eps}$ that is external to $K$ at distance $\eps$ from its boundary, and let $(p,h)$ denote the augmented boundary point consisting of the closest point $p \in \partial K$ to $p_{\eps}$ and the supporting hyperplane through $p$ that is orthogonal to the segment $p p_{\eps}$ (see Figure~\ref{fig:dual-cap}(a)). We define the \emph{$\eps$-dual cap} of $p$, denoted $D(p)$, to be the set of augmented boundary points $(q,h')$ such that the supporting hyperplane $h'$ through $q$ intersects the closed line segment $p p_{\eps}$. (Equivalently, these are the points of $\partial K$ that are visible to $p_{\eps}$.) 

\begin{figure}[htbp]
  \centerline{\includegraphics[scale=0.38]{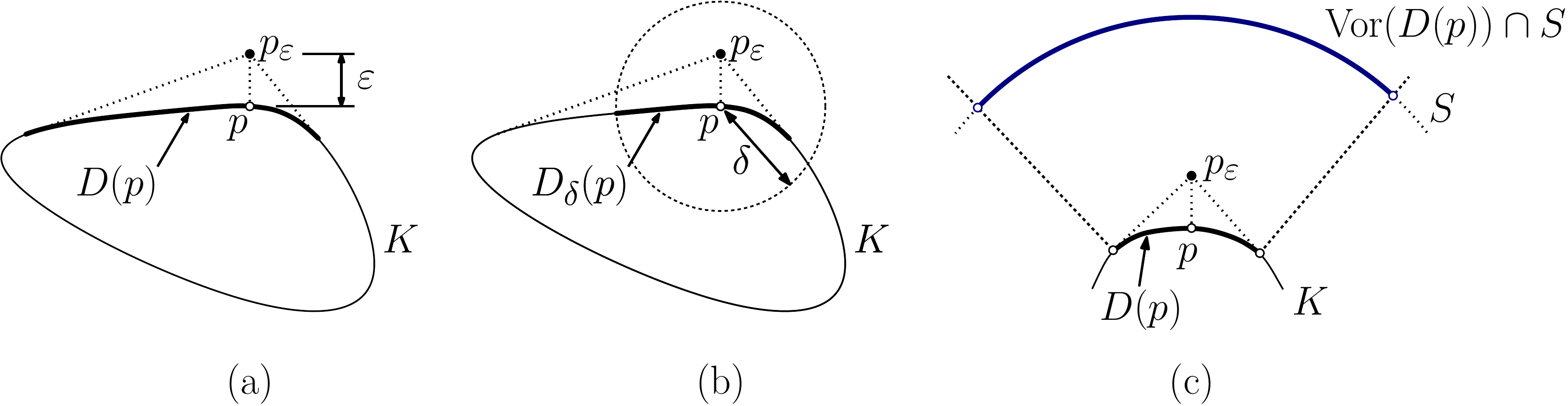}}
  \caption{(a) Dual caps, (b) restricted dual caps, and (c) the Voronoi patch of a dual cap.}
  \label{fig:dual-cap}
\end{figure}

Any outer $\eps$-approximation of $K$ by halfspaces must contain at least one halfspace that separates $p$ from $p_{\eps}$, and this can be achieved by including $h'$ for any pair $(q,h')$ within $D(p)$. A set of augmented points $\Sigma \subseteq \partial K$ is said to be an \emph{$\eps$-hitting set} for $K$ if for every $p \in \partial K$, $\Sigma \cap D(p) \ne \emptyset$. It follows directly that the intersection of the supporting halfspaces for any $\eps$-hitting set is an outer $\eps$-approximation of $K$. This observation will be formalized within our quadtree-based context in our next lemma. Before stating the lemma, we need to introduce one additional concept. In order to approximate $K$ within a given quadtree cell $Q$, we are interested only in the geometry of $K$'s boundary that lies close to $Q$. For this reason, it will be desirable to limit the diameter of dual caps. Given $\delta > 0$, let $B_{\delta}(p)$ denote the closed Euclidean ball of radius $\delta$ centered at $p$. Define the \emph{$\delta$-restricted dual cap}, denoted $D_{\delta}(p)$, to be the intersection of $D(p)$ with $B_{\delta}(p)$ (see Figure~\ref{fig:dual-cap}(b)).

\begin{lemma} \label{lem:patchouter}
Let $K$ be a convex body, $Q \subseteq Q_0$ be a quadtree cell that intersects $\partial K$, and $0 < \eps \le 1/2$. Let $\Sigma$ be any set of augmented points on $E(K,Q)$ that hits the set of all $\sqrt{\eps}$-restricted $\eps$-dual caps whose defining point is in $E(K,Q)$ (see Figure~\ref{fig:patchouter}(a)). Then there is a polytope $P$ defined as the intersection of $|\Sigma|$ halfspaces that $\eps$-approximates $K$ within $Q$.
\end{lemma}

\begin{proof}
Let $P$ be the polytope defined by the intersection of the supporting halfspaces associated with each augmented point of $\Sigma$ (see Figure~\ref{fig:patchouter}(b)). Clearly, $K \subseteq P$. Consider any point $q \in Q$ that is at distance greater than $\eps$ from $K \cap Q$. It suffices to show that $q \notin P$, that is, there exists a bounding hyperplane for $P$ that separates $q$ from $K$. 

\begin{figure}[htbp]
  \centerline{\includegraphics[scale=0.40]{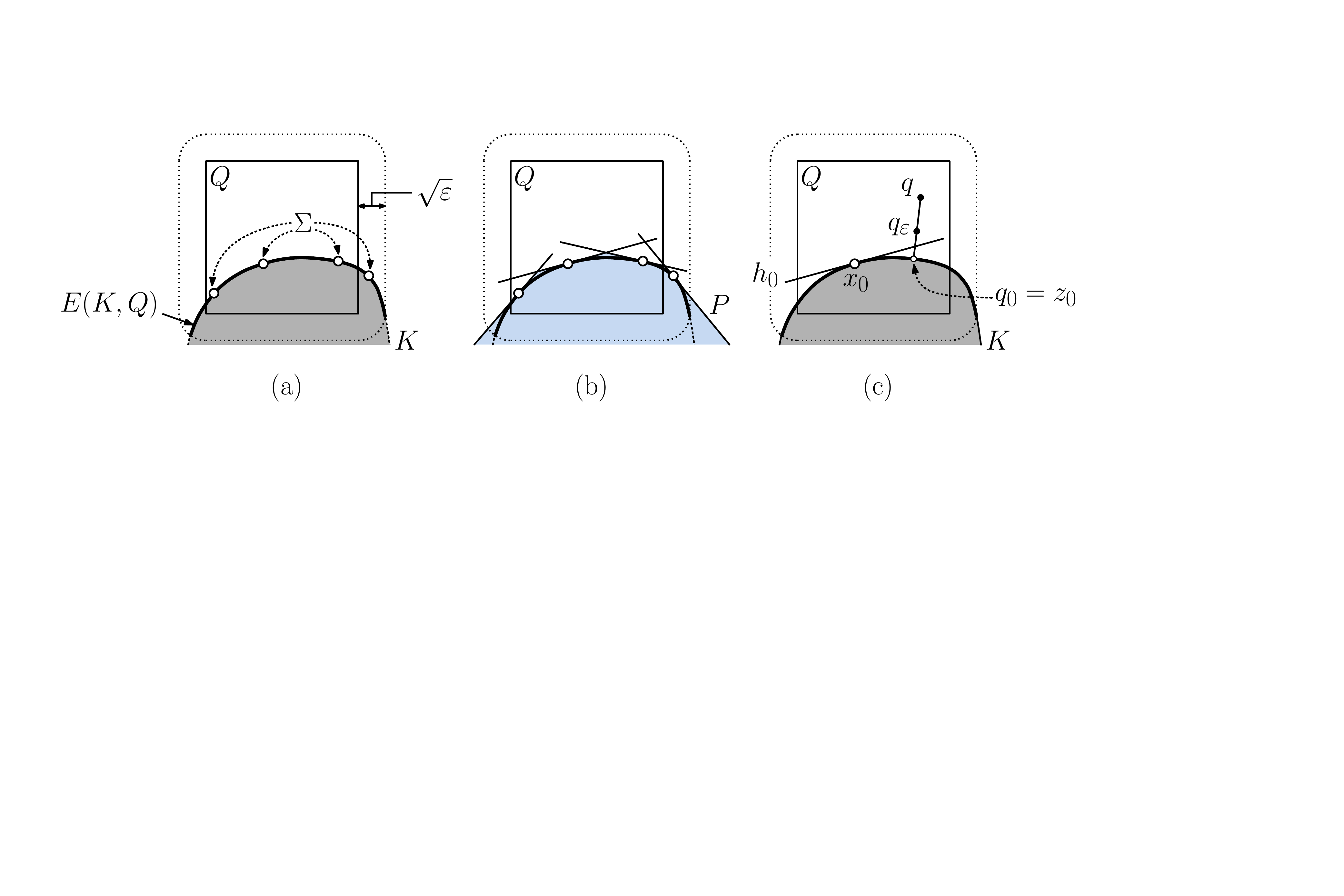}}
  \caption{Lemma~\ref{lem:patchouter}.}
  \label{fig:patchouter}
\end{figure}

We apply a similar argument to the one that we used in the proof of Lemma~\ref{lem:restricted-dudley}. Consider any $q \in Q$ that is at distance greater than $\eps$ from $K$ (see Figure~\ref{fig:patchouter}(c)). It suffices to show that there exists a bounding hyperplane for $P$ that separates $q$ from $K$. Let $q_0$ denote the point of $K \cap Q$ that is closest to $q$. By continuity, there must be a point on the segment $\overline{q q}_0$ that is at distance exactly $\eps$ from $\partial K$, which we denote by $q_{\eps}$. Since $Q$ is convex, this segment must be contained in $Q$, and, hence, so is $q_{\eps}$. 

Let $z_0$ be the point on $\partial K$ that is closest to $q_{\eps}$. (In our figure $z_0 = q_0$, but generally $z_0$ need not lie within $Q$.) Since $\eps \le 1$, we have $\|q_{\eps} z_0\| = \eps \le \sqrt{\eps}$. It follows that $z_0 \in E(K,Q)$. Therefore, there exists an augmented point $(x_0,h_0) \in \Sigma$ that hits the $\sqrt{\eps}$-restricted $\eps$-dual cap defined by $z_0$ (whose apex is at $q_{\eps}$). The supporting hyperplane $h_0$ separates $q_{\eps}$ (and therefore $q$) from $K$, as desired.
\end{proof}

Our analysis of the space bounds of {\alg} is based on the combined sizes of the $\eps$-hitting sets for $K$ within each quadtree cell $Q$. Dudley's construction can be viewed as one method of computing $\eps$-hitting sets. Unfortunately, Dudley's construction does not lead to the best bounds because it tends to over-sample in regions of very low or very high curvature. Our analysis will be based on a more refined, area-based approach to bounding the sizes of hitting sets. The key geometric observation is that the product of the areas of any $\eps$-dual cap and its associated Voronoi patch on the Dudley sphere $S$ must be large. Intuitively, if the surface area of an $\eps$-dual cap is small, then the total curvature of the patch must be high, and so the associated Voronoi patch must have relatively large area (see Figure~\ref{fig:dual-cap}(c)). More precisely, we show that (under certain conditions) the product of the areas of an $\eps$-dual cap and its Voronoi patch is $\Omega(\eps^{d-1})$. This result is stated formally in Lemma~\ref{lem:dual-basic} below. Given a $(d-1)$-dimensional manifold, let $\area(Y)$ denote its $(d-1)$-dimensional Hausdorff measure. Given a convex body $X$ in $\RE^d$, we use $\area(X)$ as a shorthand for $\area(\partial X)$. 

\begin{lemma}[Area-Product Bound] \label{lem:dual-basic}
Let $K$ be a convex body in $\RE^d$, let $0 < \eps \le 1/8$. Consider a pair $(p,h(p))$, where $p \in \partial K$ and $h(p)$ is a supporting hyperplane passing through $p$. Let $D$ denote the $\sqrt{\eps}$-restricted $\eps$-dual cap whose defining point is $p$. If $K$ is fat and of diameter at least $2 \eps$, there exists a constant $c_a$ (depending only on $d$) such that if $p$ lies within a unit ball centered at the origin, then $\area(D) \cdot \area(\Vor(D) \cap S) > c_a \cdot \eps^{d-1}$.
\end{lemma}

The proof of the lemma is quite technical and will be deferred to Section~\ref{sec:proof}. The geometric basis of the proof involves the Mahler volume, which was introduced in Section~\ref{sec:prelim-polar}. The bound stated in the lemma holds if $K$ is $\gamma$-fat for any $\gamma$ in the interval $(0,1]$ under the assumption that $\gamma$ does not depend on $\eps$. In particular, the proof will reveal that $c_a = \Omega(\gamma^{d-1})$.

We will exploit this observation to demonstrate the existence of smaller $\eps$-hitting sets than those given by Dudley's construction. We will hit the restricted $\eps$-dual caps that have large surface area by sampling points randomly on the boundary of $K$, and we will hit those with small surface area by sampling points randomly on the Dudley hypersphere and then selecting their nearest neighbors on $\partial K$. In order to prove that such a random sampling strategy works to stab all the dual caps, we need to establish bounds on the VC-dimension of an appropriate range space based on dual caps. This is not surprising given that dual caps and restricted dual caps are defined by a constant number of parameters. The result is stated in the following lemma. The proof involves a straightforward application of basic geometric principles and appears in the appendix.

\begin{restatable}{lemma}{VCLemmaStmt}\label{lem:VC}
Let $K$ be a convex body in $\RE^d$ that lies within $Q_0$, and let $\eps$ and $\delta$ be positive real parameters. The following range spaces $(X_i,R_i)$ have constant VC-dimension (where the constant depends only on $d$):
\begin{enumerate}
	\setlength{\itemsep}{-0.5ex}%
	\setlength{\parsep}{0pt}%
	\item[$(1)$] $X_1 = \partial K$ and $R_1$ is the set of $\eps$-dual caps.
	\item[$(2)$] $X_2 = S$ and $R_2$ is the set of Voronoi patches of the $\eps$-dual caps.
	\item[$(3)$] $X_3 = \partial K$ and $R_3$ is the set of $\delta$-restricted $\eps$-dual caps.
	\item[$(4)$] $X_4 = S$ and $R_4$ is the set of Voronoi patches of the $\delta$-restricted $\eps$-dual caps.
\end{enumerate}
\end{restatable}

In the next section we will exploit this result to establish the existence of small $\eps$-nets for these range spaces. Note that the range spaces defined in this lemma are defined over $\partial K$, a domain of infinite cardinality. However, for our purposes, it suffices to consider dual caps and restricted dual caps whose defining points are drawn from any sufficiently dense set of points on $\partial K$ (depending on $\eps$), and therefore the domains of the range spaces can be treated as finite sets.

\section{Final Upper Bound} \label{sec:membership}

In this section, we use the tools developed in Sections~\ref{sec:firstbound} and~\ref{sec:cap} to obtain better upper bounds for approximate polytope membership. In particular, we present a proof of Theorem~\ref{thm:membership-ub}. We will first show how to apply the area-based techniques described in the previous section to improve the simple upper bound from Lemma~\ref{lem:firstboundunit} at the low-space end of the trade-off spectrum. (This will be presented in Lemma~\ref{lem:baseunit}.) We will then apply this improvement repeatedly in an inductive manner to establish trade-offs throughout the spectrum. For technical reasons, many of the lemmas of this section assume constant upper bounds on the value of $\eps$. There is no loss of generality in doing so, since it is easy to show that if $\eps$ is bounded below by any fixed constant, the asymptotic space and query times of {\alg} are both $O(1)$.

Throughout this section, recall that $E_{\delta}(K,Q)$ is the portion of $\partial K$ that is within distance $\delta$ of $Q$, and $E(K,Q) = E_{\sqrt{\eps}}(K,Q)$. Also, define $E^+(K,Q) = E_{2\sqrt{\eps}}(K,Q)$. We will assume that $\diam(K) \ge 2 \eps$, for otherwise it is trivial to compute an $\eps$-approximation of constant size. Our first result establishes an area-based bound on the number of halfspaces needed to approximate $K$ within a quadtree cell $Q$.

\begin{lemma} \label{lem:Guilherme}
Let $K$ be a fat convex body in $\RE^d$, let $0 < \eps \le 1/8$, and let $Q \subseteq Q_0$ be a quadtree cell that intersects $\partial K$. Letting $c_a$ denote the constant of Lemma~\ref{lem:dual-basic}, define
\[
	r 
		~ = ~ \left(\frac{\area(E^+(K,Q)) \cdot \area\left(\Vor(E^+(K,Q)) \cap S\right)}{c_a \cdot \eps^{d-1}} \right)^{\N 1/2}.
\]
There is a polytope $P$ defined as the intersection of $O(r \log r)$ halfspaces that is an $\eps$-approximation of $K$ within $Q$.
\end{lemma}

\begin{proof}
Letting $A_K = \area(E^+(K,Q))$ and $A_S = \area\left(\Vor(E^+(K,Q)) \cap S\right)$, we can express the value of $r$ more succinctly as $\big(A_K A_S/c_a \P\eps^{d-1}\big)^{1/2}$. First, we assert that $r = \Omega(1)$. To see this, we consider two cases. First, if $K$ lies entirely within distance $2 \sqrt{\eps}$ of $Q$, then $\Vor(E^+(K,Q)) \cap S = S$, which implies that $A_S = \Omega(1)$. Since $K$ is fat and by our assumption that $\diam(K) \ge 2 \eps$, it follows that $A_K = \Omega(\eps^{d-1})$. Therefore, $r = \Omega(1)$. On the other hand, if some part of $K$ lies at distance greater than $2 \sqrt{\eps}$ from $Q$, $E^+(K,Q)$ is a boundary patch of $K$ of diameter $\Omega(\sqrt{\eps})$. Since both $K$ and $Q$ are fat, it follows that $A_K = \Omega\big(\eps^{(d-1)/2}\big)$. By convexity, as we go from a boundary patch on $K$ to its Voronoi cell on $S$, distances cannot decrease. Therefore $A_S \ge A_K$, and again we have $r = \Omega(1)$. Through a minor adjustment to constant factor $c_a$ in $r$'s definition, we may assume that $\log r \ge 1$.

By Lemma~\ref{lem:patchouter}, in order to show the existence of an $\eps$-approximating polytope $P$ for $K$ within $Q$, it suffices to show that it is possible to hit all $\sqrt{\eps}$-restricted $\eps$-dual caps whose defining point lies in $E(K,Q)$ (not to be confused with $E^+(K,Q)$) using $O(r \log r)$ points. To do this, we distinguish between two types of such restricted dual caps. A restricted dual cap $D$ is of \emph{type~1}, if $\area(D) \ge \big(c_a \P \eps^{d-1} A_K / A_S\big)^{1/2}$, and otherwise it is of \emph{type~2}.

By assertions~(3) and~(4) of Lemma~\ref{lem:VC}, we know that $\sqrt{\eps}$-restricted $\eps$-dual caps and their Voronoi patches both have constant VC-dimension. The VC-dimension is no larger if we restrict the domain of the range space. Therefore, by standard machinery (see, e.g., \cite{probabilistic}) we can build a $(1/r)$-net for any restriction of these range spaces of size $O(r \log r)$ each by random sampling. 

For type-1 dual caps, consider the restriction $(E^+(K,Q),R_3)$ of the range space given in Lemma~\ref{lem:VC}(3). Let $\Sigma_1$ denote a $(1/r)$-net. Consider any type-1 dual cap $D$. Since $D$'s defining point lies within $E(K,Q)$ and it is $\sqrt{\eps}$-restricted, it lies entirely within $E^+(K,Q)$. Thus, we have
\begin{eqnarray*}
	\frac{\area(D \cap E^+(K,Q))}{\area(E^+(K,Q))} 
		&  =  & \frac{\area(D)}{\area(E^+(K,Q))} 
		~ \ge ~ \frac{\big(c_a \cdot \eps^{d-1} A_K / A_S\big)^{1/2}}{A_K} \\
		&  =  & \left(\frac{c_a \cdot \eps^{d-1}}{A_K A_S}\right)^{\N 1/2}
		~  =  ~ \inv{r}.
\end{eqnarray*}
Therefore $D$ contains at least one point of $\Sigma_1$. It follows that $\Sigma_1$ hits all type-1 dual caps.

For type-2 dual caps, let us consider the restriction $(\Vor(E^+(K,Q)) \cap S,R_4)$ of the range space of Lemma~\ref{lem:VC}(4). Let $\Sigma_2$ denote a $(1/r)$-net. Because $\eps \le 1/8$ and $Q \subseteq Q_0$, $E(K,Q)$ lies within a ball centered at the origin of radius $\diam(Q_0)/2 + \sqrt{\eps} \le 1$. Given any type-2 dual cap $D$ whose defining (augmented) point lies in $E(K,Q)$, we may apply Lemma~\ref{lem:dual-basic} to obtain
\[
	\area(\Vor(D) \cap S)
		~ \ge ~ \frac{c_a \cdot \eps^{d-1}}{\area(D)}
		~ \ge ~ \frac{c_a \cdot \eps^{d-1}}{\big(c_a \cdot \eps^{d-1} A_K / A_S\big)^{1/2}}
		~  =  ~ \left(\frac{c_a \cdot \eps^{d-1} A_S}{A_K}\right)^{\N 1/2}.
\]
As before, since $D$'s defining point lies within $E(K,Q)$, $D \subseteq E^+(K,Q)$. From this we have
\begin{eqnarray*}
	\frac{\area(\Vor(D \cap E^+(K,Q)) \cap S)}{\area(\Vor(E^+(K,Q)) \cap S)} 
		&  =  & \frac{\area(\Vor(D) \cap S)}{\area(\Vor(E^+(K,Q)) \cap S)}  \\
		& \ge & \frac{\big(c_a \cdot \eps^{d-1} A_S/A_K\big)^{1/2}}{A_S} 
		~  =  ~ \left(\frac{c_a \cdot\eps^{d-1}}{A_K A_S}\right)^{\N 1/2} \\
		&  =  & \inv{r}.
\end{eqnarray*}
Therefore $\Vor(D) \cap S$ contains at least one point of $\Sigma_2$, implying that $\Sigma_2$ hits the Voronoi patches of all type-2 dual caps. For each point of $\Sigma_2$, we select its nearest neighbor on $\partial K$, obtaining a set $\Sigma'_2 \subset E^+(K,Q)$. It follows directly that the set $\Sigma'_2$ hits all type-2 dual caps. Therefore, the union $\Sigma_1 \cup \Sigma'_2$ forms the desired set of size $O(r \log r)$ that hits all $\sqrt{\eps}$-restricted $\eps$-dual caps whose defining point lies within $E(K,Q)$.
\end{proof}

In order to establish our storage bounds, we analyze the behavior of the algorithm at a particular level of the decomposition. Given the query-time parameter $t$, recall that we stop the subdivision process in $\alg(K,Q)$ if the number of hyperplanes needed to approximate $K$ within $Q$ falls below $t$. Also recall that $t(Q)$ denotes the number of approximating halfspaces associated with $Q$. Let us consider the state of the subdivision process when the cell sizes reach roughly $\sqrt{\eps}$. Cells that have stopped subdividing by this point are ``good,'' since we can bound the total space requirements for all such cells by appealing to Lemma~\ref{lem:space-bound}. For the remaining ``bad'' cells, we will bound their space requirements on a cell-by-cell basis by using the simple upper bound from Lemma~\ref{lem:firstbound}. For our approach to work well, it is crucial to obtain a good bound on the number of such bad cells. We exploit the area bound of Lemma~\ref{lem:Guilherme} for this purpose. Whenever {\alg} subdivides a cell of size $O(\sqrt{\eps})$, we can infer that more than $t$ hyperplanes are required to approximate $K$ within this cell. Since the portion of $\partial K$ lying within this cell is small, the area of its Voronoi patch on the Dudley sphere must be large. A packing argument applied on the Dudley sphere will be used to limit the number of these bad cells.

In order to formalize the notion of good and bad cells, let $T$ denote the quadtree produced by $ \alg(K,Q_0)$, and let $T'$ denote the subtree of $T$ induced by cells of diameter at least $\sqrt{\eps}/2$. For the remainder of this section, let $L_1$ denote the (good) leaf cells of $T'$ that are not subdivided further by the algorithm, and let $L_2$ be the remaining (bad) leaf cells of $T'$. The cells of $L_1$ and $L_2$ are all of diameter $\Omega(\sqrt{\eps})$. Each cell in $L_1$ can be approximated using at most $t$ halfspaces, and those in $L_2$ require more. In our next lemma, we bound the total number of approximating halfspaces over all the good leaf cells and the total number of bad leaf cells. 

\begin{lemma} \label{lem:aux2}
Let $K$ be a fat convex body in $\RE^d$, let $0 < \eps \le 1/8$. Let $T$ denote the quadtree produced by $\alg(K,Q_0)$, for $t \ge 1$, and let $L_1$ and $L_2$ be as defined above. Then 
\begin{enumerate}
\item[$(i)$] $\sum_{Q  \in L_1} t(Q)  = O\big( 1/\eps^{(d-1)/2} \big)$,

\item[$(ii)$] $|L_2| = O\big( (1 + \log t) / t)^2 (1/\eps)^{(d-1)/2} \big)$.
\end{enumerate}
\end{lemma}

\begin{proof}
Because the cells of $L_1$ are disjoint and each is of diameter $\Omega(\sqrt{\eps})$, assertion~(i) follows as a direct consequence of Lemma~\ref{lem:space-bound}. Thus, it remains to prove assertion~(ii). Let $Q$ be any cell of $L_2$. Since any child of a cell of $L_2$ is of diameter smaller than $\sqrt{\eps}/2$ and $Q$'s diameter is twice this, we have $\sqrt{\eps}/2 \le \diam(Q) < \sqrt{\eps}$. Recall that $E^+(K,Q) = E_{2\sqrt{\eps}}(K,Q)$. Also, let $A_K(Q)$ and $A_S(Q)$ denote the values of $A_K$ and $A_S$, respectively, from the proof of Lemma~\ref{lem:Guilherme}, when applied to $Q$. 

Because $\diam(Q) \le \sqrt{\eps}$ and $E^+(K,Q)$ involves a boundary patch of $K$ that intersects $Q$ and includes an additional expansion by distance $2 \sqrt{\eps}$, it follows that this boundary patch has diameter $O(\sqrt{\eps})$. Therefore, $A_K(Q) = O\big(\eps^{(d-1)/2}\big)$. By applying Lemma~\ref{lem:Guilherme} (and recalling the constant $c_a$ from Lemma~\ref{lem:dual-basic}), we have $t(Q) = O(r \log r)$, where
\begin{eqnarray*}
	r 
		& = & \NN\left(\frac{\area(E^+(K,Q)) \cdot \area\left(\Vor(E^+(K,Q)) \cap S\right)}{c_a \cdot \eps^{d-1}} \right)^{\N 1/2} \NN \\
		& = & \left( \frac{A_K(Q) A_S(Q)}{c_a \cdot \eps^{d-1}}\right)^{\N 1/2} \NN
		~ = ~ \OO{ \sqrt{\frac{A_S(Q)}{\eps^{(d-1)/2}}} }.
\end{eqnarray*}
In Lemma~\ref{lem:Guilherme} we showed that (after a suitable adjustment to $c_a$), we have $\log r \ge 1$. Since $Q$ is subdivided further, we know that $t(Q) > t$, which implies that $t = O(r \log r)$. Because $t \ge 1$, by simple manipulations we have $t/(1 + \log t) = O(r)$. By combining this with the upper bound on $r$ from above, we obtain $A_S(Q) = \Omega\big( (t/(1 + \log t))^{2} \eps^{(d-1)/2}\big)$, which yields the lower bound
\[
	\sum_{Q \in L_2} A_S(Q)
		~  =  ~ |L_2| \cdot \Omega\left( \left( \frac{t}{1 + \log t} \right)^{2} \eps^{\frac{d-1}{2}}\right).
\]

As shown in the proof of Lemma~\ref{lem:space-bound}, given any set of disjoint quadtree cells of diameter $\Omega(\sqrt{\eps})$ a point of $S$ can be in $\Vor(E^+(K,Q))$ for at most a constant number of these cells. Since the quadtree cells of $L_2$ satisfy these conditions,
\[
	\sum_{Q \in L_2} A_S(Q)
		~  =  ~ \sum_{Q \in L_2} \area(\Vor(E^+(K,Q)) \cap S)
		~  =  ~ O(\area(S)).
\]
Combining this with our lower bound, we have
\[
	|L_2|
		~  =  ~ \OO{ \area(S) \cdot \left(\frac{1 + \log t}{t}\right)^{\N 2} \cdot \left( \inv{\eps} \right)^{\NN \frac{d-1}{2}} }.
\]
Since $S$ is a hypersphere of constant radius, its area is bounded, and assertion~(ii) follows immediately.
\end{proof}

Recall that we showed in Lemma~\ref{lem:firstboundunit} that it is possible to answer approximate membership queries in $1/\eps^{(d-1)/4}$ time using space $O\big( 1/\eps^{(d-1)/2} \big)$. By using the above lemma, we show next that we can improve this to achieving query time roughly $O\big( 1/\eps^{(d-1)/8} \big)$ for the same space.

\begin{lemma} \label{lem:baseunit}
Let $K$ be a fat convex body in $\RE^d$, and let $0 < \eps \le 1/16$. For $t \geq (\lg \inv{\eps})/\eps^{(d-1)/8}$, the output of $\alg(K,Q_0)$ has total space $O\big(1/\eps^{(d-1)/2}\big)$.
\end{lemma}

\begin{proof}
Let $T$ denote the quadtree produced by the algorithm. By Lemma~\ref{lem:total-space}, the data structure's total space is dominated by the space needed to store the hyperplanes in the leaf cells. Thus, it suffices to show that the sum of $t(Q)$ over all leaf cells $Q$ of $T$ is $O\big(1/\eps^{(d-1)/2}\big)$. Let $T'$, $L_1$, and $L_2$ be as defined just prior to Lemma~\ref{lem:aux2}. By Lemma~\ref{lem:aux2}(i), the total contribution of $t(Q)$ for all cells in $L_1$ is $O\big( 1/\eps^{(d-1)/2} \big)$. So, it suffices to bound the contribution due to $L_2$.

Let $Q$ be any cell of $L_2$. Recall from the proof of Lemma~\ref{lem:aux2} that $\sqrt{\eps} / 2 \le \diam(Q) \le \sqrt{\eps}$. Since $t \geq 1 / \eps^{(d-1)/8}$, it follows that $t \geq (\diam(Q) / \eps)^{(d-1)/4}$. Because $\eps \le 1/16$, we have $\eps \le \sqrt{\eps}/4 \le \diam(Q)/2$. By Lemma~\ref{lem:firstbound}, the output of $\alg(K,Q)$ has total space at most 
\[
	O\left(\left(\frac{\diam(Q)}{\eps}\right)^{\N \frac{d-1}{2}}\right)
		~ = ~ O\left(\left(\inv{\eps}\right)^{\N \frac{d-1}{4}}\right).
\]
By Lemma~\ref{lem:aux2}(ii), $|L_2| = O\big( ((1 + \log t) / t)^2 (1/\eps)^{(d-1)/2} \big)$. Since $t \geq (\lg \inv{\eps})/\eps^{(d-1)/8}$, we have $|L_2| = O\big(1/\eps^{(d-1)/4}\big)$. Summing up the space contributions of all $Q \in L_2$, the total space for these cells is 
\[
	|L_2| \cdot O\big(1/\eps^{(d-1)/4}\big)
	~ = ~ O\big(1/\big(\eps^{(d-1)/4} \cdot \eps^{(d-1)/4}\big)\big)
	~ = ~ O\big(1/\eps^{(d-1)/2}\big),
\]
as desired. 
\end{proof}

In order to extend the space-time trade-off to other query times, we will apply the previous result as the basis case in an induction argument. The induction will be controlled by a parameter $\alpha$, which we assume to be a constant. The proof is rather technical, but it involves a straightforward application of the earlier results of this section.

\begin{lemma} \label{lem:trade-off-ub}
Let $K$ be a fat convex body in $\RE^d$, and let $0 < \eps \le 1/16$. Let $\alpha \ge 4$ be a real-valued constant. For $t \ge (\lg \inv{\eps})/\eps^{(d-1)/\alpha}$, the output of $\alg(K,Q_0)$ has total space
\[
	\OO{1/\eps^{(d-1)\left(1 - 2 \left(\frac{\floor{\lg \alpha} - 2}{\alpha} + \inv{2^{\floor{\lg\alpha}}} \right) \right)}}.
\]
\end{lemma}

\begin{proof}
Define $k = \floor{\lg\alpha}$, which implies that $k \ge 2$, and $2^{k} \leq \alpha < 2^{k+1}$. Expressed as a function of $k$, the desired space bound can be expressed as
\begin{equation} \label{eq:alt-bound}
	c_k \cdot \left( 1/\eps^{(d-1) \left(1 - 2 \left(\frac{k - 2}{\alpha} + \inv{2^k}\right) \right)} \right),
\end{equation}
for a constant $c_k$ (depending on $k$ but not on $\eps$).

We begin exactly as in the proof of the previous lemma. Let $T$ denote the quadtree produced by the algorithm, and by Lemma~\ref{lem:total-space}, it suffices to bound the sum of $t(Q)$ over all leaf cells of $T$. Given $T'$, $L_1$, and $L_2$ defined prior to Lemma~\ref{lem:aux2}, the space contribution due to the cells of $L_1$ is $O\big( 1/\eps^{(d-1)/2} \big)$. To see that this satisfies our space bound, observe that since $k \ge 2$ and $\alpha \ge 2^k$, we have
\[
	\inv{4}
		~ \ge ~ \frac{k-1}{2^{k}}
		~  =  ~ \frac{k-2}{2^{k}} + \inv{2^{k}}
		~ \ge ~ \frac{k-2}{\alpha} + \inv{2^{k}}.
\]
Therefore, the total contribution of $t(Q)$ for all cells in $L_1$ is
\begin{equation} \label{eq:ind-bound}
	O\big( 1/\eps^{(d-1)/2} \big)
		~  =  ~ \OO{1/\eps^{(d-1)\left(1 - 2\left(\inv{4}\right)\right)}}
		~ \le ~ \OO{1/\eps^{(d-1)\left(1 - 2\left(\frac{k-2}{\alpha} + \inv{2^{k}}\right)\right)}},
\end{equation}
which matches the desired bound given in Eq.~(\ref{eq:alt-bound}).

It remains to bound the contribution to the space of the cells of $L_2$. We do this by induction on $k$. For the basis case $k=2$, we have $4 \le \alpha < 8$. Therefore $t > (\lg \inv{\eps})/\eps^{(d-1)/8}$. By applying Lemma~\ref{lem:baseunit}, the total space of the data structure (which includes the contribution of $L_2$) is $O\big(1/\eps^{(d-1)/2}\big)$. It follows from Eq.~(\ref{eq:ind-bound}) (for the case $k = 2$) that this satisfies our storage bound.

For the induction step, we assume that the lemma holds for $k-1$ (that is, $2^{k-1} \leq \alpha/2 < 2^{k}$), and our objective is to prove it for $k$. It will be convenient to express the induction hypothesis in a form that holds for an arbitrary quadtree cell $Q \subseteq Q_0$. By applying standardization to $Q$ (thus mapping $Q$ to $Q_0$ and scaling $\eps$ to $\eps/\diam(Q)$), the induction hypothesis states that for
\begin{equation} \label{eq:time-hyp}
  0 ~ < ~ \eps ~ \le ~ \frac{\diam(Q)}{16} ~~\hbox{and}~~
	t ~ \ge ~  \left(\lg \frac{\diam(Q)}{\eps} \right) \cdot \left( \frac{\diam(Q)}{\eps} \right)^{\N \frac{d-1}{\alpha/2}},
\end{equation}
there is a constant $c_{k-1}$ such that the output of $\alg(K,Q)$ has total space at most
\begin{equation} \label{eq:sp-hyp}
	c_{k-1} \cdot (\diam(Q)/\eps)^{(d-1) \left(1 - 2 \left(\frac{k - 3}{\alpha/2} + \inv{2^{k-1}}\right) \right)}.
\end{equation}

Let $Q$ be any cell of $L_2$. In the proof of Lemma~\ref{lem:aux2} we showed that $\sqrt{\eps}/2 \le \diam(Q) < \sqrt{\eps}$. By the bound on $t$ from the statement of this lemma, we have
\[
	t
		~ \ge ~ \left( \lg \inv{\eps} \right) \left( \inv{\eps} \right)^{\NN \frac{d-1}{\alpha}} \NN
		~ \ge ~ \left( \lg \inv{\sqrt{\eps}} \right) \left( \inv{\sqrt{\eps}} \right)^{\NN \frac{2(d-1)}{\alpha}} \NN
		~ \ge ~ \left( \lg \frac{\diam(Q)}{\eps} \right) \left( \frac{\diam(Q)}{\eps} \right)^{\N \frac{d-1}{\alpha/2}} \NN ,
\]
implying that $t$ satisfies Eq.~(\ref{eq:time-hyp}). If $\eps$ is at most $\diam(Q)/16$, we may apply the induction hypothesis, yielding the space bound given in Eq.~(\ref{eq:sp-hyp}). Since $\diam(Q) < \sqrt{\eps}$, this can be simplified to
\begin{equation} \label{eq:sp-bound}
	c_{k-1} \cdot (1/\sqrt{\eps})^{(d-1) \left(1 - 2 \left(\frac{k - 3}{\alpha/2} + \inv{2^{k-1}}\right) \right)}
		~ = ~
	c_{k-1} \cdot 1/\eps^{(d-1) \left(\inv{2} - \frac{2(k - 3)}{\alpha} - \inv{2^{k-1}}\right)}.
\end{equation}

By combining Lemma~\ref{lem:aux2}(ii) with the lower bound on $t$ given in the statement of this lemma, the number of cells in $L_2$ satisfies
\begin{equation} \label{eq:l2-size}
	|L_2|
		~ = ~ \OO{\left(\frac{\lg t}{t}\right)^{\N 2} \left( \inv{\eps} \right)^{\NN \frac{d-1}{2}}}
		~ = ~ \OO{\eps^{\frac{2(d-1)}{\alpha}}  \left( \inv{\eps} \right)^{\NN \frac{d-1}{2}}}
		~ = ~ \OO{\left( \inv{\eps} \right)^{\N (d-1)\left(\inv{2} - \frac{2}{\alpha}\right)}}.
\end{equation}
The total contribution to the space by the cells of $L_2$ is the product of the space requirements for each cell of $L_2$, given in Eq.~(\ref{eq:sp-bound}), and the number of such cells, given in Eq.~(\ref{eq:l2-size}). There exists a constant $c_k$ (depending on $k$ but not on $\eps$) such the total space is at most
\[
	c_k \cdot \left( 1/\eps^{(d-1) \left( \left(\inv{2} - \frac{2(k - 3)}{\alpha} - \inv{2^{k-1}}\right) + \left(\inv{2} - \frac{2}{\alpha}\right) \right)} \right)	
		~ = ~
	c_k \cdot \left( 1/\eps^{(d-1) \left(1 - 2 \left(\frac{k - 2}{\alpha} + \inv{2^k}\right) \right)} \right).
\]
On the other hand, if $\eps$ exceeds $\diam(Q)/16$, then since $\diam(Q) \ge \sqrt{\eps}/2$ it follows that $\eps$ is $\Omega(1)$, and we can adjust to $c_k$ to satisfy this bound. In either case, we achieve the bound in Eq.~(\ref{eq:alt-bound}). 
\end{proof}

Observe that the exponent in the space bound in the preceding lemma is a piecewise linear function in $1/\alpha$, whose breakpoints coincide with powers of two. It is easily verified that the exponent is a continuous function of $\alpha$. (In particular, observe that $\lim_{\delta \rightarrow 0} f(2^{k-\delta}) = f(2^{k})$, where $f(\alpha) = 1/2^{\floor{\lg \alpha}} + (\floor{\lg \alpha}-2)/\alpha$.)

\medskip

We can now present the proof of Theorem~\ref{thm:membership-ub}. Recall that $K$ is a convex polytope in $\RE^d$. By Lemma~\ref{lem:fat}, we can precondition $K$ so that it is $(1/d)$-fat and is contained within $Q_0$, thus allowing us to approximate $K$ absolutely. Also, if $1/16 < \eps \le 1$, we set $\eps = 1/16$. (Both of these changes result in a constant factor decrease to $\eps$, which will not affect the asymptotic bounds.) We then set $t = (\lg \inv{\eps})/\eps^{(d-1)/\alpha}$ and invoke $\alg(K, Q_0)$. Let $T$ denote the resulting data structure. Given the preconditioning of $K$ and the alteration of $\eps$, we may apply Lemma~\ref{lem:trade-off-ub} to show that the total space for $T$ is 
\[
	\OO{1/\eps^{(d-1)\left(1 - 2 \left(\frac{\floor{\lg \alpha} - 2}{\alpha} + \inv{2^{\floor{\lg\alpha}}}\right) \right)}}.
\]
Using the fact that $1/2^{\floor{\lg\alpha}} \ge 1/\alpha$, this is
\[
	\OO{1/\eps^{(d-1)\left(1 - 2 \left(\frac{\floor{\lg \alpha} - 2}{\alpha} + \inv{\alpha}\right) \right)}}
		~ = ~
	\OO{1/\eps^{\left(d-1\right)\left(1 - \frac{2\floor{\lg \alpha} - 2}{\alpha} \right)}},
\]
which matches the space bound of Theorem~\ref{thm:membership-ub}.

Recall that a query is answered by locating the leaf node of $T$ that contains the query point, followed by an inspection of the (at most) $t$ halfspaces stored in this leaf node. By our remarks following the presentation of {\alg}, $T$ is of height $O(\log \inv{\eps})$, which implies that the query time is dominated by the value of $t$. This completes the proof of Theorem~\ref{thm:membership-ub}.

\section{Preprocessing} \label{sec:preproc}

Our principal focus so far has been in establishing the existence of trade-offs between space and query time, without considering how to construct the data structure. In this section we discuss preprocessing issues. We first discuss the preconditioning of $K$ as described in Lemma~\ref{lem:fat} and then discuss the implementation of the access primitives (i)--(iii) needed for {\alg} as presented at the start of Section~\ref{sec:split-reduce}. We assume that the input convex body $K$ is presented as the intersection of a set $\mathcal{H}$ of $n$ halfspaces in $\RE^d$. Throughout, let $Q$ denote an arbitrary quadtree cell.

Let $t$ denote the query-time parameter in {\alg}. As observed in Section~\ref{sec:split-reduce}, under our assumption that $t \ge 1$, Steps~{\stepout} and {\stepin} are not needed, since we can rely entirely on Step~{\stepapx}, and therefore access primitives~(i) and~(ii) are not needed.%
\footnote{If we wished to we could implement access primitive~(i) in linear time by linear programming. Also, by testing the membership of each of $Q$'s vertices in $K$, we could implement a stronger version of access primitive~(ii), namely that of determining whether $Q \subseteq K$ (as opposed to $K \oplus \eps$).}
The remainder of this section will be focused on preconditioning (Section~\ref{sec:precondition}) and the implementation of access primitive~(iii), which locally approximates $K$ within $Q$ (Section~\ref{sec:apx-cover}).

\subsection{Preconditioning} \label{sec:precondition}

Recall that we assume that $K$ is a (full-dimensional) convex polytope in $\RE^d$ that is presented as the intersection of a set of $n$ closed halfspaces. Also recall that $Q_0$ is the axis-aligned hypercube of unit diameter that is centered at the origin. Our objective is to precondition $K$ by computing an affine transformation that both fattens $K$ and maps it to lie within $Q_0$. $Q_0$ has a side length of $1/\sqrt{d}$, and therefore it contains a ball of radius $1/2\sqrt{d}$ centered at the origin. Let $B_0$ denote this ball, and let $r_0$ denote its radius. For $0 < \gamma \le 1$, let $\gamma B_0$ denote the concentric ball of radius $\gamma \P r_0 = \gamma/2\sqrt{d}$. We say that a polytope is in \emph{$\gamma$-canonical position} if it is nested between $\gamma B_0$ and $B_0$ (see Figure~\ref{fig:canonical}). Clearly, a polytope that is in canonical position is contained within $Q_0$ and is $\gamma$-fat. The following lemma shows that $K$ can be efficiently mapped into this form, and furthermore an absolute approximation to the transformed body can be easily mapped to a relative approximation of $K$. (Lemma~\ref{lem:fat} of Section~\ref{sec:prelim-abs} follows as an immediate consequence of this.) Such fattening operations are commonplace in geometric approximation algorithms (see, e.g., \cite{AHV-coreset,Chan-coreset,HP-book,BHP-bbox}), and we employ the standard approach based on minimum enclosing volumes, the John Ellipsoid in particular. 

\begin{figure}[htbp]
  \centerline{\includegraphics[scale=0.40]{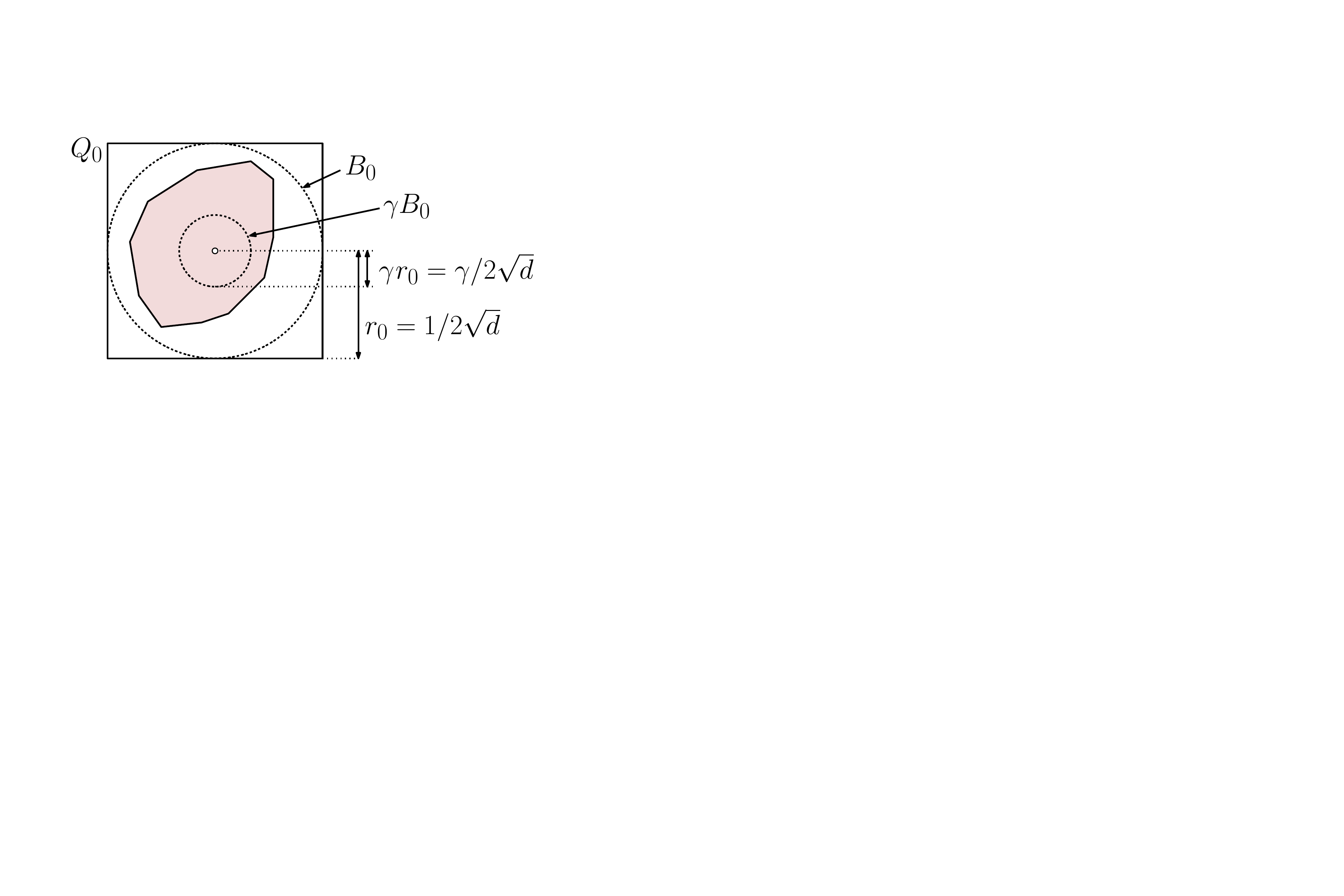}}
  \caption{A polytope in $\gamma$-canonical position.}
  \label{fig:canonical}
\end{figure}

\begin{lemma} \label{lem:precondition-1}
Let $K$ be a convex polytope in $\RE^d$ defined as the intersection of a set $\mathcal{H}$ of $n$ halfspaces, and let $0 < \eps \le 1$. There is an algorithm that, given $\mathcal{H}$ and $\eps$, in $O(n)$ time computes an affine transformation $T$ that maps $K$ into $(1/d)$-canonical position, such that if $P$ is an absolute $\eps/(d \sqrt{d})$-approximation of $T(K)$, then $T^{-1}(P)$ is a relative $\eps$-approximation of $K$.
\end{lemma}

\begin{proof}
Chazelle and Matou\v{s}ek~\cite{ChM96} show that in any fixed dimension, there exists an $O(n)$ time algorithm that, given a convex polytope $K$ presented as the intersection of $n$ halfspaces, computes an ellipsoid $E$ of maximum volume contained within $K$, also known as the \emph{John Ellipsoid}~\cite{Bal97}. (At the expense of an increase in the constant factors, we can apply the simpler construction by Har-Peled and Barequet~\cite{BHP-bbox}.) Since $d$ is fixed, in constant time we can compute an affine transformation $T$ that maps $E$ to the ball $B_0/d$ (that is, $\inv{d}B_0$). (Because $K$ is full dimensional, both $T$ and its inverse $T^{-1}$ are well-defined.) It is well known from John's Theorem (see, e.g., \cite{Bal97}) that $K$ is contained within a uniform scaling of $E$ by a factor of $d$, which we denote by $d \P E$. Therefore, we have $E \subseteq K \subseteq d \P E$, which implies that $B_0/d \subseteq T(K) \subseteq B_0$.

Let $\sigma = d \sqrt{d} \cdot \diam(E)$. Because $T$ maps an ellipse of diameter $\diam(E)$ to a ball of diameter $2 r_0/d = 1/d \sqrt{d}$, it follows that $T$ maps any vector $v$ to a vector of length at least $\|v\|/\sigma$. (The principal axis aligned with $E$'s diameter is scaled by exactly $1/\sigma$, and all other principal axes, which form a basis for the space, are scaled by at least this much.) Therefore, $\|T^{-1}(v)\| \le \sigma \|v\|$. 

If $P$ is any absolute $\eps/(d \sqrt{d})$-approximation to $T(K)$, then by definition 
\begin{equation}
	T(K) 
		~ \subseteq ~ P 
		~ \subseteq ~ T(K) \oplus \frac{\eps}{d \sqrt{d}}. \label{eqn:inclusion}
\end{equation}
To show that $T^{-1}(P)$ is a relative $\eps$-approximation to $K$, observe first that by applying $T^{-1}$ to the first inclusion of Eq.~(\ref{eqn:inclusion}), we have $K \subseteq T^{-1}(P)$. Also, by the second inclusion of Eq.~(\ref{eqn:inclusion}) we know that for any $p \in P$, there exists $q \in T(K)$ such that the vector $p-q$ is of length at most $\eps/(d \sqrt{d})$. Therefore, 
\[
	\|T^{-1} (p-q)\|
		~ \le ~ \sigma \|p-q\|
		~ \le ~ \sigma \frac{\eps}{d \sqrt{d}}
		~  =  ~ \eps \cdot \diam(E)
		~ \le ~ \eps \cdot \diam(K).
\]
We conclude that
\[
	K
		~ \subseteq ~ T^{-1}(P) 
		~ \subseteq ~ K \oplus (\eps \cdot \diam(K)),
\]
and therefore $T^{-1}(P)$ is a relative $\eps$-approximation of $K$.
\end{proof}

In order to make subsequent processing more efficient, we adapt a standard coreset construction to reduce the number of halfspaces to a function depending only on $\eps$ and $d$. The process will involve some further scaling, which will slightly modify the parameters.

\begin{lemma} \label{lem:precondition-2}
Let $K$ be a convex polytope in $\RE^d$ defined as the intersection of a set $\mathcal{H}$ of $n$ halfspaces, and let $0 < \eps \le 1$. There is an algorithm that, given $\mathcal{H}$ and $\eps$, in $O\big(n + 1/\eps^{d-1}\big)$ time computes an affine transformation $T'$ and a subset $\mathcal{H}' \subseteq \mathcal{H}$ of size $O\big(1/\eps^{(d-1)/2}\big)$ such that: 
\begin{enumerate}
\setlength{\itemsep}{-0.5ex}%
\setlength{\parsep}{0pt}%
\item[$(i)$] applying $T'$ to the intersection of $\mathcal{H}'$ results in a convex polytope $K'$ that is in $(1/2d)$-canonical position;

\item[$(ii)$] furthermore, if $P$ is an absolute $\eps/(4 d \sqrt{d})$-approximation of $K'$, then ${T'}^{-1}(P)$ is a relative $\eps$-approximation of $K$.
\end{enumerate}
\end{lemma}

\begin{proof}
Given $\mathcal{H}$, we begin by computing the transformation $T$ of Lemma~\ref{lem:precondition-1} in $O(n)$ time. Let $T(K)$ denote the resulting polytope, which is in $(1/d)$-canonical position (see Figure~\ref{fig:precondition-2}(a)).

Given a set $S$ of points $\RE^d$, the \emph{extent measure} associates each unit vector $u \in \RE^d$ with the minimum distance between two hyperplanes orthogonal to $u$ that contain $S$ between them (see Figure~\ref{fig:coreset}(a)). More formally, define $w_u(S) = \max_{p,q \in S} \ang{p-q,u}$ (recalling that $\ang{\cdot, \cdot}$ denotes inner product). A subset $S' \subseteq S$ is said to be an \emph{$\eps$-coreset} for the extent measure if for all unit vectors $u$, $w_u(S') \ge (1-\eps) w_u(S)$ (see Figure~\ref{fig:coreset}(b)). Agarwal {\etal}~\cite{AHV-coreset} showed that, given a set of $n$ points in $\RE^d$, it is possible to construct an $\eps$-coreset for the extent measure of size $O\big(1/\eps^{(d-1)/2}\big)$. We will employ an improvement of this result due to Chan, who presented an algorithm to compute such a coreset in $O\big(n + (1/\eps)^{d-1}\big)$ time \cite{Chan-coreset}.

\begin{figure}[htbp]
  \centerline{\includegraphics[scale=0.40]{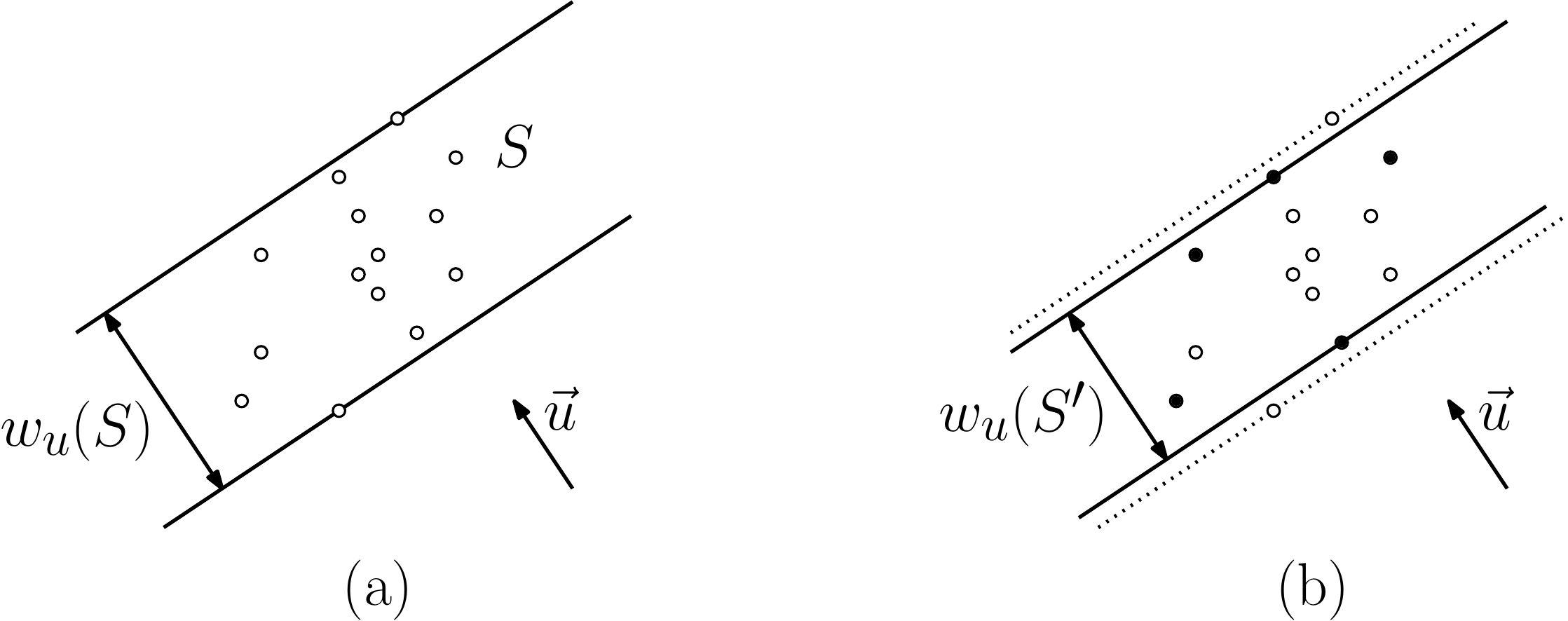}}
  \caption{(a) The extent measure $w_u(S)$ and (b) a coreset.}
  \label{fig:coreset}
\end{figure}

Since $T(K)$ is in $(1/d)$-canonical position, we have $B_0/d \subseteq T(K) \subseteq B_0$. Let $S$ denote the set of $n$ points in $\RE^d$ that result by applying the polar transformation (see Section~\ref{sec:prelim-polar}) to each hyperplane of $T(\mathcal{H})$. It follows from the definition of the polar transformation that $\conv(S) = \polar{T(K)}$ (see Figure~\ref{fig:precondition-2}(b)). As mentioned in Section~\ref{sec:prelim-polar}, the polar transformation maps an origin-centered ball of radius $r$ to a ball of radius $1/r$. Thus, $\conv(S)$ is nested between an inner ball of radius $2 \sqrt{d}$ and an outer ball of radius $2 d \sqrt{d}$. Given $\mathcal{H}$ and $T$, we can easily compute the set $S$ in $O(n)$ time. Let $\eps' = \eps/4 d^2$. We then apply Chan's algorithm to compute an $\eps'$-coreset $S' \subseteq S$ in time $O\big(n + (1/\eps')^{d-1}\big) = O\big(n + (1/\eps)^{d-1}\big)$ (see Figure~\ref{fig:precondition-2}(c)). Let $\mathcal{H'}$ be the subset of $\mathcal{H}$ that results by taking the polar duals of the points of $S'$, and let $K'$ be the convex body that results from intersecting these halfspaces (see Figure~\ref{fig:precondition-2}(d)).

\begin{figure}[htbp]
  \centerline{\includegraphics[scale=0.40]{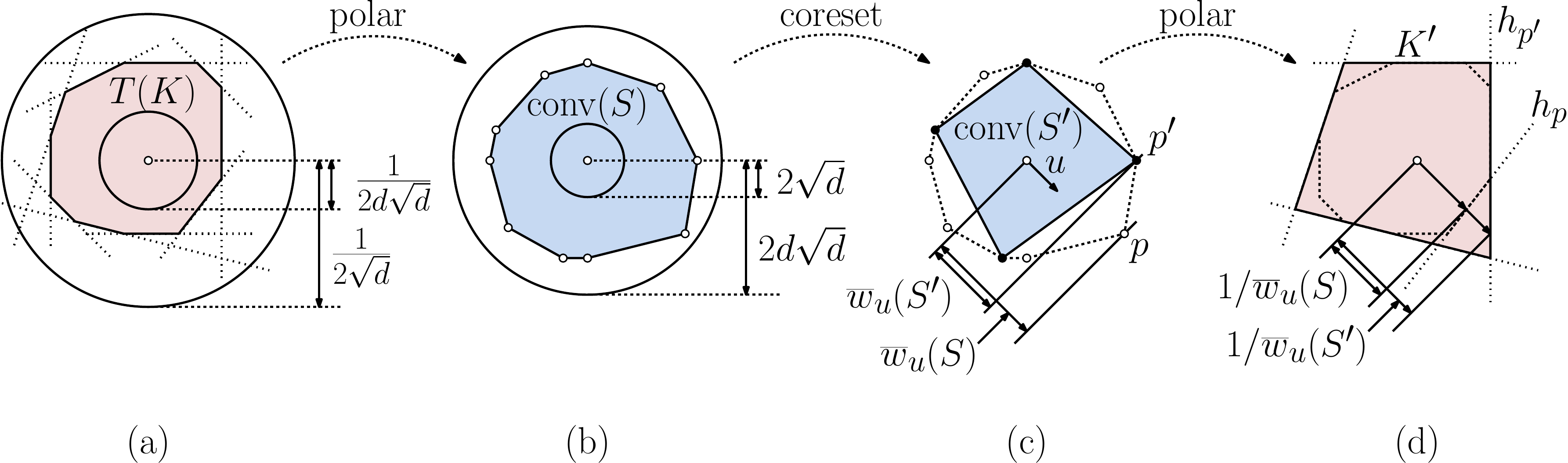}}
  \caption{Proof of Lemma~\ref{lem:precondition-2}. (Not drawn to scale.)}
  \label{fig:precondition-2}
\end{figure}

Clearly, $T(K) \subseteq K'$ and $|\mathcal{H'}| = O\big(1/\eps^{(d-1)/2}\big)$. We assert that the Hausdorff distance between $T(K)$ and $K'$ is at most $\eps/2d\sqrt{d}$. To prove this, we apply an observation due to Chan \cite{Chan-coreset}. Define the \emph{one-sided extent measure}, denoted $\overline{w}_u(S)$ to be $\max_{p \in S} \ang{p, u}$. (This is the distance from the origin to $S$'s closest supporting hyperplane orthogonal to and on the same side as $u$.) In Observation~{1.4} of \cite{Chan-coreset} Chan shows that if $S'$ is an $\eps'$-coreset for the extent measure, then $\overline{w}_u(S') \ge \overline{w}_u(S) - \eps' \cdot w_u(S)$. Given the nesting properties of $S$ and the fact that $u$ is a unit vector we have 
\[
	2 \sqrt{d} ~ \le ~ \overline{w}_u(S) ~ \le ~ 2 d \sqrt{d}
		\qquad\mbox{and}\qquad
	4 \sqrt{d} ~ \le ~ w_u(S) ~ \le ~ 4 d \sqrt{d}. 
\]
Therefore, $w_u(S) \le 2 d \cdot \overline{w}_u(S)$, and with Chan's observation this yields $\overline{w}_u(S') \ge (1 - 2 d \eps') \overline{w}_u(S)$. Under our assumption that $\eps \le 1$, we have $\eps' \le 1/4 d^2$, and so $1 - 2 d \eps' \ge 1 - 1/2 d \ge 1/2$.

Let $p \in S$ be the point that determines $\overline{w}_u(S)$ (see Figure~\ref{fig:precondition-2}(c)). Treating $p$ as a vector, its polar dual is a hyperplane, which we denote by $h_p$ (see Figure~\ref{fig:precondition-2}(d)). It follows directly from the definition of polarity that if we shoot a bullet from the origin parallel to $u$ until it hits $h_p$, the length of the resulting segment is $1/\overline{w}_u(S)$. (To see this, observe that $h_p = \{v : \ang{p,v} = 1\}$, and so $c \cdot u$ lies on $h_p$, for $c = 1/\ang{p,u} = 1/\overline{w}_u(S)$.) Analogously, letting $p' \in S'$ denote the point that determines $\overline{w}_u(S')$, the length of the segment parallel to $u$ that hits the associated polar dual hyperplane $h_{p'}$ is of length $1/\overline{w}_u(S')$. Applying this to every unit vector in $\RE^d$, the Hausdorff distance between $T(K)$ and $K'$ is at most the supremum over all unit vectors of 
\begin{eqnarray*}
	\inv{\overline{w}_u(S')} - \inv{\overline{w}_u(S)}
		& \le & \inv{(1 - 2 d \eps') \cdot \overline{w}_u(S)} - \inv{\overline{w}_u(S)}
		~  =  ~ \frac{2 d \eps'}{(1 - 2 d \eps') \cdot \overline{w}_u(S)} \\
		& \le & \frac{2 d \eps'}{(1/2) \cdot 2 \sqrt{d}}
		~ \le ~ 2 \sqrt{d} \eps'
		~ \le ~ \frac{\eps}{2d \sqrt{d}},
\end{eqnarray*}
which establishes our assertion.

Therefore, if $P$ is any absolute $(\eps/2 d \sqrt{d})$-approximation to $K'$, then by the triangle inequality (applied to the Hausdorff distance) $P$ is an absolute $(\eps/2 d \sqrt{d}) + (\eps/2 d \sqrt{d}) = \eps/d \sqrt{d}$ approximation to $T(K)$. By Lemma~\ref{lem:precondition-1}, $P$ is a relative $\eps$-approximation of $K$.

We are almost done, but the canonical-position condition fails, because $K'$ need not lie within $B_0$ of radius $r_0 = 1/2 \sqrt{d}$ (even though $T(K)$ does). Since the Hausdorff distance between $K'$ and $T(K)$ is at most $\eps/2 d \sqrt{d} \le 1/2\sqrt{d} = r_0$, it follows that $K'$ lies within $2 B_0$. The simple fix is to apply a uniform scaling of space by a factor of $1/2$. In particular, define $T'$ to be the composition of $T$ with such a scaling transformation, and apply the above coreset construction to $\mathcal{H}$ transformed by $T'$, but now with $\eps'$ scaled accordingly to $\eps/8 d^2$. The resulting transformed and reduced polytope is nested between $B_0$ and $B_0/2 d$, and so it is in $(1/2 d)$-canonical position. Also, if $P$ is an absolute $(\eps/4 d \sqrt{d})$-approximation of the transformed and reduced polytope, then ${T'}^{-1}(P)$ is a relative $\eps$-approximation of $K$, as desired.
\end{proof}

\subsection{Efficient Local Approximations} \label{sec:apx-cover}

Next, we consider the implementation of access primitive~(iii), which given a convex body $K$ in $\gamma$-canonical position, a quadtree cell $Q$, and query-time $t$, determines whether there exist $t$ halfspaces whose intersection $\eps$-approximates $K$ within $Q$. The space and query times stated in Theorem~\ref{thm:membership-ub} are based on the assumption that the number of bounding halfspaces of this local approximating polytope is within a constant factor of optimal. However, we know of no efficient algorithm that can achieve this. In this section we show how to efficiently implement Step~(\stepapx) of {\alg} approximately in the sense that the number of halfspaces in the approximation exceeds the optimum (for a slightly smaller approximation parameter) by a factor of $O(\log \inv{\eps} )$. As shown in Lemma~\ref{lem:weak-membership}, this will lead to an increase in the space and query times stated in Theorem~\ref{thm:membership-ub} by a factor of only $O(\log \inv{\eps})$.

A natural approach would be to adapt Clarkson's algorithm for polytope approximation \cite{Clarkson-polytope}. There are a few messy technical issues involved with such an adaptation. (For example, Clarkson's algorithm applies to the convex hull of a set of points, rather than the intersection of halfspaces.) Since we do not require the strong approximation bounds provided by Clarkson's algorithm, we will instead present a simple direct solution based on a reduction to the set-cover problem. Our approach is to construct a set system where the point set consists of a dense set of points of spacing $\Theta(\eps)$ that covers the portion of $Q$ that is external to $K \oplus c' \P \eps$, for a suitable constant $c' < 1$. We associate each bounding halfspace of $K$ with the set of grid points that lie \emph{outside} of this halfspace. We will show that the halfspaces associated with a minimum set cover for this system produces the desired local approximation. We use the greedy set cover heuristic to construct this cover.

Recall that $K \oplus r$ denotes the set of points that lie within Euclidean distance $r$ of $K$. In order to avoid the complexities of determining whether a point lies outside of $K \oplus c' \eps$, it will suffice for our purposes to perform the simpler test of whether a point lie outside a scaled copy of $K$.

\begin{lemma} \label{lem:scaled-growth}
For $0 < \gamma \le 1$ and $0 < \eps \le 1$, let $K$ be a polytope in $\RE^d$ that is in $\gamma$-canonical position, and let $K^+ = \big(1 + 2 \sqrt{d} \eps\big) K$. Then 
\[
	K \oplus \gamma \eps
		~ \subseteq ~ K^+
		~ \subseteq ~ K \oplus \eps.
\]
\end{lemma}

\begin{proof}
By definition of $\gamma$-canonical position, the distance of every point of $\partial K$ to the origin lies between $\gamma/2 \sqrt{d}$ and $1/2 \sqrt{d}$ (see Figure~\ref{fig:scaled-growth}(a)). To prove the first inclusion ($K \oplus \gamma \eps \subseteq K^+$), notice that the distance between any supporting hyperplane of $K$ and its parallel supporting hyperplane of $K^+$ is at least
\[
	\frac{\gamma}{2\sqrt{d}} \cdot 2 \sqrt{d} \eps
		~ = ~ \gamma \eps.
\]
The Minkowski sum $K \oplus \gamma \eps$ is equivalent to translating every supporting hyperplane of $K$ away from the origin by distance $\gamma \eps$ and intersecting the associated (infinite) set of halfspaces. If we apply this to just the (finitely many) defining hyperplanes of $K$, we obtain a polytope that contains $K \oplus \gamma \eps$, and therefore $K \oplus \gamma \eps \subseteq K^+$ (see Figure~\ref{fig:scaled-growth}(b) and~(c)). 

\begin{figure}[htbp]
  \centerline{\includegraphics[scale=0.40]{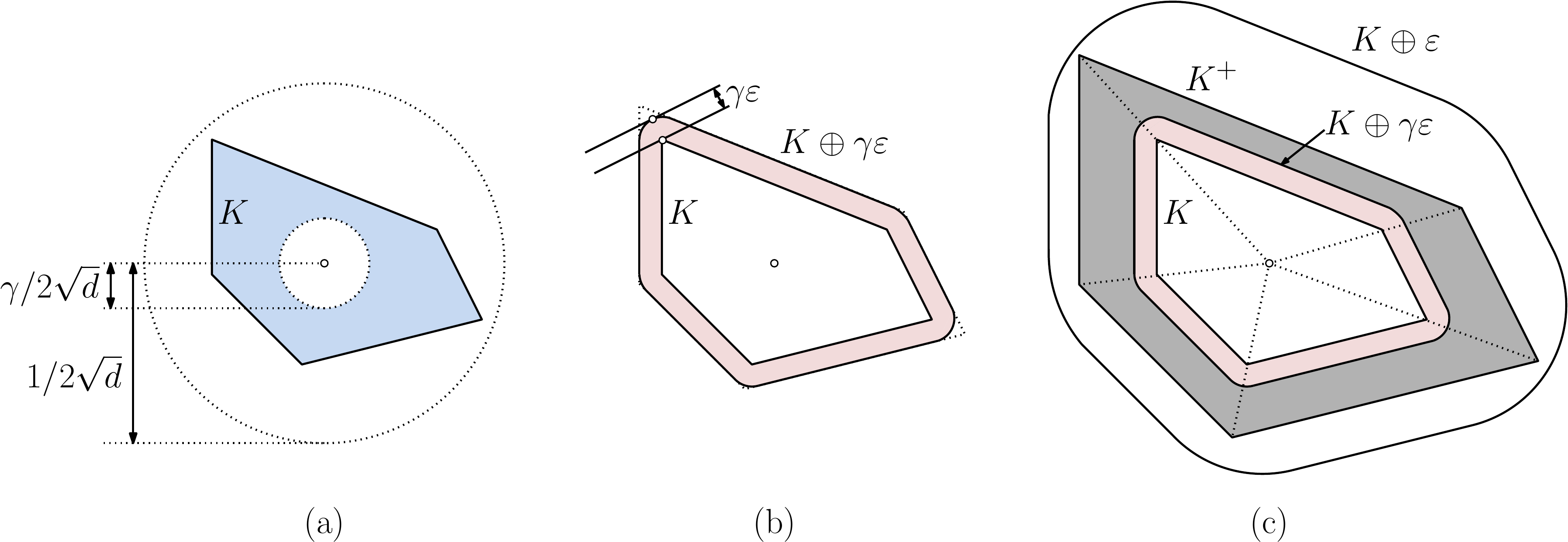}}
  \caption{Proof of Lemma~\ref{lem:scaled-growth}.}
  \label{fig:scaled-growth}
\end{figure}

To establish the second inclusion ($K^+ \subseteq K \oplus \eps$), observe that each point $p \in K$ is in 1--1 correspondence with the point $p' = \big(1 + 2 \sqrt{d} \eps\big) p \in K^+$. Since $p$ is within distance $1/2 \sqrt{d}$ of the origin, $p'$ is within distance $2 \sqrt{d} \eps/2 \sqrt{d} = \eps$ of $p$. This implies that $K^+ \subseteq K \oplus \eps$, as desired.
\end{proof}

While access primitive~(iii) does not place any restrictions on the halfspaces used when computing an $\eps$-approximation to $K$ within $Q$, when the query point $q$ lies outside of $K$, it may be useful to add further restrictions. In particular, when the query point lies outside of $K$, it is desirable to obtain a \emph{witness} to nonmembership in the form of a bounding halfspace of $K$ that does not contain $q$. (This will be exploited in Section~\ref{sec:ann} in the reduction of approximate nearest neighbor searching to approximate polytope membership. The witness hyperplane is used to identify the approximate nearest neighbor.) To achieve this, we would like to use bounding halfspaces from the original polytope in our approximation. By a simple application of Carath\'{e}odory's Theorem, we can show that we sacrifice only a constant factor by adding this restriction. The following is a straightforward generalization of Lemma~{3.1} from Mitchell and Suri \cite{MS}. 

\begin{lemma} \label{lem:discrete-apx}
Let $K$ be a convex polytope in $\RE^d$ defined as the intersection of a set $\mathcal{H}$ of halfspaces, and let $Q \subseteq Q_0$ be a quadtree cell. If there exists an $\eps$-approximation of $K$ within $Q$ bounded by $m$ halfspaces, then there exists a subset of $\mathcal{H}$ of size at most $d \P m$ that $\eps$-approximates $K$ within $Q$. 
\end{lemma}

\begin{proof}
The cases $K \cap Q = \emptyset$ and $Q \subseteq K$ are both trivial, so let us assume that the boundary of $K$ intersects $Q$. Let $P$ be an $\eps$-approximation of $K$ within $Q$ that is bounded by $m$ halfspaces, and let $h$ be any one of these halfspaces. By the definition of  such an approximation, $K \cap Q \subseteq P \cap Q \subseteq (K \oplus \eps) \cap Q$. We may assume that $h$ is a supporting halfspace of $K \cap Q$, for otherwise we can translate it until it is. Let $v$ be any vertex of $K \cap Q$ on $h$'s boundary. Let $\mathcal{H}'$ denote the union of $\mathcal{H}$ and $Q$'s bounding halfspaces. By Carath\'{e}odory's Theorem there exist $d$ halfspaces from $\mathcal{H}'$ such that the complement of $h$ is contained within the union of the complement of these $d$ halfspaces. If we replace $h$ with these $d$ halfspaces, the resulting polytope still $\eps$-approximates $K$ within $Q$. 

After repeating this for each of the $m$ halfspaces bounding $P$, we obtain an $\eps$-approximating polytope by a subset of at most $d m$ halfspaces of $\mathcal{H}'$. Let $P'$ be the result of removing from this polytope all the halfspaces that are not in $\mathcal{H}$ (and hence must bound $Q$). We have
\[
  K \cap Q 
	~ \subseteq ~ P' \cap Q 
	~ \subseteq ~ P \cap Q
	~ \subseteq ~ (K \oplus \eps) \cap Q,
\]
and therefore $P'$ is the desired $\eps$-approximation.
\end{proof}

We are now in a position to present our set-cover-based local approximation. This is a bi-criteria approximation since it is suboptimal with respect to both the number of bounding halfspaces and the approximation parameter.

\begin{lemma} \label{lem:apx-cover}
For $0 < \gamma \le 1$ and $0 < \eps \le 1$, let $K$ be a polytope in $\RE^d$ in $\gamma$-canonical position that is given as the intersection of a set $\mathcal{H}$ of $n$ halfspaces. Let $Q \subseteq Q_0$ be a quadtree cell. In $O(n/\eps^d)$ time, it is possible to compute a subset $\mathcal{H'} \subseteq \mathcal{H}$ such that  
\begin{enumerate}
\setlength{\itemsep}{-0.5ex}%
\setlength{\parsep}{0pt}%
\item[$(i)$] The intersection of the halfspaces of $\mathcal{H'}$ is an $\eps$-approximation of $K$ within $Q$

\item[$(ii)$] If $m$ denotes the minimum number of halfspaces needed to $(\gamma \eps/2)$-approximate $K$ within $Q$, then $|\mathcal{H'}|$ is $O(m \log \inv{\eps})$.
\end{enumerate}
\end{lemma}

\begin{proof}
First, we may assume without loss of generality that $\eps \le 2/\sqrt{d}$. Otherwise, setting $\eps = 2/\sqrt{d}$ will certainly satisfy (i) and will only affect the constant factors in the asymptotic bounds of claim (ii) and the construction time. Define $\beta = \sqrt{d}\eps/2$. By the above assumption, we have
\[
	(1 + \beta)^2
		~  =  ~ \left( 1 + \sqrt{d}\eps + \frac{d \eps^2}{4} \right)
		~ \le ~ 1 + \frac{3}{2}\sqrt{d}\eps.
\]
Let $K^+ = (1 + \beta) K$ and let $K^{++} = (1 + \beta) K^+ = (1 + \beta)^2 K$. By applying Lemma~\ref{lem:scaled-growth} but with $\eps$ taking on the values $\eps/4$ and $3\eps/4$, respectively, we have
\begin{equation}
	K \oplus \frac{\gamma \eps}{4}
		~ \subseteq ~ K^+
		~ \subseteq ~ K^{++}
		~ \subseteq ~ K \oplus \frac{3\eps}{4} \label{eq:apx-cover}
\end{equation}
(see Figure~\ref{fig:apx-cover}(a)). Let $\delta = \gamma \eps/4$ and let $G$ denote the vertices of a hypercube grid of diameter $\delta$. Let $R$ be the set of grid points that lie within $Q$ but outside of $K^{++}$, that is, $R = G \cap (Q \setminus K^{++})$. Since $Q \subseteq Q_0$, the resulting set is of size $O(1/\eps^d)$, and hence it can be computed in time $O(1/\eps^d) \cdot |\mathcal{H}| = O(n/\eps^d)$, by testing each grid point against each halfspace of $\mathcal{H}$. Because $\gamma \le 1$, we have $\delta \le \eps/4$.

\begin{figure}[htbp]
  \centerline{\includegraphics[scale=0.40]{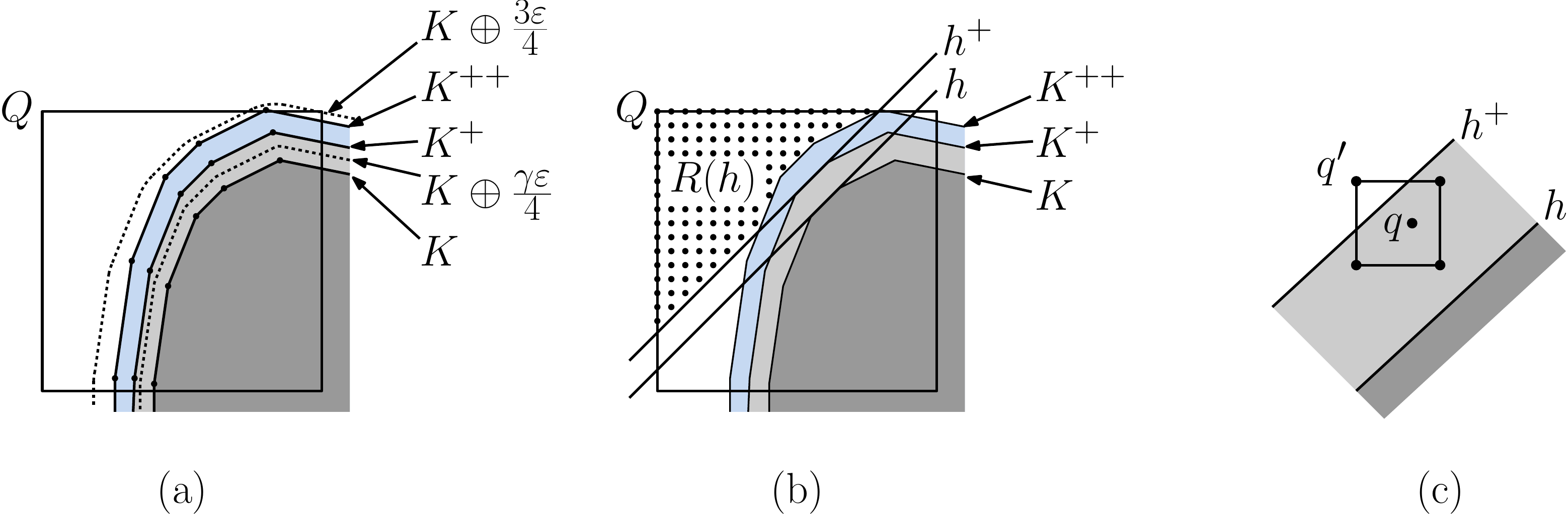}}
  \caption{Proof of Lemma~\ref{lem:apx-cover}.}
  \label{fig:apx-cover}
\end{figure}

Next, we define a set system to model the approximation process. For each $h \in \mathcal{H}$ we define a subset $R(h) \subseteq R$ as follows. First, let $h^+ = (1 + \beta) h$ denote the corresponding bounding halfspace of the scaled body $K^+$ (see Figure~\ref{fig:apx-cover}(b)). Define $R(h)$ to be the subset of points of $R$ that lie outside of $h^+$. Consider a set system consisting of the points of $R$ and the sets $R(h)$ for all $h \in \mathcal{H}$. Since every point of $R$ lies outside of $K^{++}$, and hence outside of $K^+$, together these sets cover $R$. The resulting collection of sets has total cardinality $O(n/\eps^d)$. 

Consider any set cover $C$ of the resulting set system. Let $P(C)$ denote the polyhedron that results by intersecting the halfspaces $h$ whose associated set $R(h)$ is included in this cover. (Note that the sets $R(h)$ are based on the halfspaces bounding the scaled body $K^+$ while $P(C)$ is based on the halfspaces bounding the original body $K$.) We assert that $P(C)$ $\eps$-approximates $K$ within $Q$. It suffices to show that for any point $q \in Q \setminus (K \oplus \eps)$, $q$ is not in $P(C)$. First, observe that for such a point $q$, all the vertices of the grid cell in which it lies are within distance $\delta$ of $q$. Therefore, by the triangle inequality, each such vertex is at distance at least $\eps - \delta \ge 3\eps/4$ from $K$. Since by Eq.~\eqref{eq:apx-cover}, $K^{++} \subseteq K \oplus 3\eps/4$, these vertices are all exterior to $K^{++}$, which implies that they are all members of $R$. Let $q'$ be any of these vertices. Since $C$ is a cover, there exists a halfspace $h \in \mathcal{H}$ such that $R(h)$ is in the cover and contains this point. This implies that $q'$ lies outside the associated halfspace $h^+$ (see Figure~\ref{fig:apx-cover}(c)). Because $K$ is in $\gamma$-canonical position, the minimum distance between $h$'s bounding hyperplane and the origin is at least $\gamma/2\sqrt{d}$. Therefore the distance between any point in $h$ to any point exterior to $h^+$ is at least
\[
	\frac{\gamma}{2\sqrt{d}} ( (1 + \beta) - 1)
		~ = ~ \frac{\gamma}{2\sqrt{d}} \cdot \frac{\sqrt{d}\eps}{2}
		~ = ~ \frac{\gamma\eps}{4}
		~ = ~ \delta.
\]
It follows by the triangle inequality that $q$ is exterior to $h$, and therefore it lies outside of $P(C)$, as desired.

Let $C'$ denote a set cover that results by running the greedy heuristic \cite{CLRS} on the aforementioned set system. By standard results on the greedy heuristic, the size of the resulting cover exceeds that of an optimal cover by a factor of at most $\ln |R| = O(\log \inv{\eps})$. $C'$ can be computed in time that is proportional to the total cardinality of the sets of the set system, which is $O(n/\eps^d)$. Let $\mathcal{H}'$ denote the associated set of halfspaces, and let $P(C')$ denote the intersection of these halfspaces. By the above remarks, $P(C')$ is an $\eps$-approximation to $K$ within $Q$, which establishes claim~(i).

To establish~(ii), consider a $(\gamma \eps/4)$-approximation of $K$ within $Q$ that is bounded by the minimum number $m$ of halfspaces. By Lemma~\ref{lem:discrete-apx} there exists such an approximation that uses only the bounding halfspaces of $K$, such that the number of halfspaces is larger by a factor of at most $d$. Let $P^+$ denote this approximation, and let $\mathcal{H}^+ \subseteq \mathcal{H}$ denote its bounding halfspaces. By Eq.~\eqref{eq:apx-cover}, we have $P^+ \subseteq K \oplus \gamma\eps/4 \subseteq K^+$. Let $P^{++} = (1 + \beta) P^+$. Clearly, $P^{++} \subseteq (1 + \beta) K^+ = K^{++}$. Therefore, every point of $R$ lies outside of $P^{++}$. It follows that the sets $R(h)$ associated with the halfspaces $h$ that bound $P^+$ form a set cover of $R$ within our system. Letting $C^{++}$ denote this cover, we have $|C'| \le O(\log \inv{\eps}) \cdot |C^{++}| \le O(\log \inv{\eps}) \cdot d m = O(m \log \inv{\eps})$, as desired.
\end{proof}

We can now present the main result of this section, which summarizes the preprocessing time.

\begin{lemma} \label{lem:preproc-time}
Given a full-dimensional convex polytope $K$ in $\RE^d$ defined as the intersection of a set of $n$ halfspaces, approximation parameter $0 < \eps \le 1$, and query time parameter $t \ge 1$, there is an algorithm that runs in time $O\big(n + 1/\eps^{c_p \P d}\big)$ for some constant $c_p$ (which does not depend on $d$) that constructs a data structure satisfying Theorem~\ref{thm:membership-ub} but with an additional factor of $O(\log \inv{\eps})$ in both the space and query times.
\end{lemma}

\begin{proof}
Given $K$'s bounding halfspaces, we apply Lemma~\ref{lem:precondition-2}. In $O(n + 1/\eps^{d-1})$ time we obtain a polytope $K'$, such that $K'$ is in $\gamma$-canonical position for $\gamma = 1/2 d$. $K'$ is bounded by a subset $\mathcal{H}'$ of halfspaces of size $n' = O\big(1/\eps^{(d-1)/2}\big)$, and the problem of computing a relative $\eps$-approximation of $K$ reduces to the problem of computing an absolute $\eps'$-approximation of $K'$, where $\eps' = \eps/4 d \sqrt{d}$. 

Ideally, we would like to invoke {\alg} on $K'$ using $\eps'$ as the approximation parameter and $t$ as the query time parameter. Since we do not know how to determine minimum-sized convex approximations efficiently, we will need to relax our expectations. For any quadtree cell $Q$ generated by {\alg}, we apply Lemma~\ref{lem:apx-cover} on the set $\mathcal{H}'$ of halfspaces. By claim~(i) of this lemma, after $O(n'/(\eps')^{d}) = O(1/\eps^{3 d/2})$ time, a subset $\mathcal{H}'' \subseteq \mathcal{H}'$ can be computed that is an $\eps'$-approximation of $K'$ within $Q$. Irrespective of the choice of the query time, the maximum number of quadtree cells generated by {\alg} is $O(1/\eps^d)$, and therefore (after preconditioning) the overall running time of {\alg} is $O(1/\eps^{5 d/2})$. Combined with the $O(n + 1/\eps^d)$ time for preconditioning, the algorithm's overall running time is $O\big(n + 1/\eps^{c_p \P d}\big)$, where $c_p = 5/2$.

Let $\eps'' = \gamma\eps'/2 = \gamma\eps/8 d \sqrt{d}$. By Lemma~\ref{lem:apx-cover}(ii) the number of halfspaces in $\mathcal{H}''$ is within a factor of $\rho = O(\log \inv{\eps'}) = O(\log \inv{\eps})$ of the size of the minimum-sized $\eps''$-approximation of $K'$ within $Q$. Since $\eps'' = \beta \eps'$ for a constant $\beta$, Lemma~\ref{lem:weak-membership} implies that the conclusions of Theorem~\ref{thm:membership-ub} hold but with an additional factor of $\rho = O(\log \inv{\eps})$ in both the space and query times.
\end{proof}

\section{Lower Bound} \label{sec:lb}

In this section, we establish lower bounds on the space-time trade-offs obtained by {\alg} for polytope membership. In particular, we will prove Theorem~\ref{thm:lb} from Section~\ref{sec:intro}. Our approach is similar to the lower bound proof of~\cite{AVD-JACM}. (Note that this is a lower bound on the performance of {\alg}, not on the problem complexity. It applies to the stronger existential version of the algorithm.) It is based on analyzing the performance of the algorithm on a particular convex body, a generalized hypercylinder that is curved in $k+1$ dimensions and flat in $d-1-k$ dimensions. We select the value of $k$ that produces the best lower bound on the storage as a function of $t$, $\eps$, and $d$. Throughout, we use the term $\eps$-approximation in the absolute sense, as defined in Section~\ref{sec:prelim-abs}. 

As mentioned earlier, it is well known that $\Omega\big(1/\eps^{(d-1)/2}\big)$ facets are required to $\eps$-approximate a Euclidean ball of unit radius (see, e.g., \cite{Bro08}), and this holds for any polytope that that is sufficiently close to a ball in terms of Hausdorff distance. The following utility lemma generalizes this observation to different diameters. The proof is straightforward, but for the sake of completeness we include its proof in the appendix.

\begin{restatable}{lemma}{BallLemmaStmt}\label{lem:ball}
Let $\eps$ and $\Delta$ be real parameters, where $0 < \eps \le \Delta/4$. There exists a constant $c_b$ and a polytope $P$ in $\RE^d$ of diameter at most $\Delta$ such that any outer $\eps$-approximation of $P$ requires at least $c_b (\Delta / \eps)^{(d-1)/2}$ facets.
\end{restatable}

Intuitively, in order to produce a polytope that is hard to approximate, it should have high curvature. If the curvature is high in all dimensions, however, the polytope will have a small surface area, and this will make it easier to approximate. Our approach is to consider polytopes based on generalized cylinders, which have constant curvature in some dimensions but are flat in others. Our next lemma introduces such a cylindrical polytope where the number of curved dimensions has been carefully chosen to maximize the space needed by our algorithm for a given query time. Theorem~\ref{thm:lb} is an immediate consequence.

\begin{lemma} \label{lem:cyl}
There exists a polytope $P$ in $\RE^d$ such that for all sufficiently small positive $\eps$ (depending on $d$ and $\alpha$) and  $t = 1/\eps^{(d-1)/\alpha}$, the output of $\alg(K,Q_0)$ on $P$ has total space
\[
	\Omega\left(1/\eps^{(d-1) \left(1 - \frac{2\sqrt{2\alpha} - 3}{\alpha}\right) - 1} \right).
\]
\end{lemma}

\begin{proof}
To start, as a function of $\alpha$, we wish to compute an integer dimension $k$ in order to apply Lemma~\ref{lem:ball}. Define reals $\delta = \sqrt{\alpha/2}/(d-1)$, $\kappa = (d-1)\sqrt{2/\alpha}$ and $\kappa' = \kappa(1+\delta)$. We observe first that
\[
	\kappa' - \kappa
		~  =  ~ \delta (d-1) \sqrt{2/\alpha}
		~  =  ~ 1.
\]
Let $k = \ceil{\kappa}$, implying that $\kappa \le k \le \kappa'$. (Although we do not include the derivation here, $\kappa$ has been chosen to produce the best lower bound, but since it is not necessarily an integer, $k$ is obtained by rounding to a nearby integer.) Since $\alpha \ge 4$ and $d \ge 2$, we have $1 \le k \le d-1$.

Let $c_b$ denote the constant of Lemma~\ref{lem:ball}, and let $\Delta = \eps((2^d+1)t/c_b)^{2/k}$. By our assumptions about $d$ and $\alpha$, we have $t = 1/\eps^{\Theta(1)}$ and $\Delta = \eps \cdot t^{\Theta(1)}$. It follows that for all sufficiently small $\eps$, $\Delta/4 \ge \eps$. Let $h$ denote the linear subspace spanned by the first $k+1$ coordinate axes. We apply Lemma~\ref{lem:ball} in $\RE^{k+1}$ for this value of $\Delta$. The resulting polytope $P$ (lying in $h$) has the property that the number of facets of any $\eps$-approximation is at least
\[
	c_b \left( \frac{\Delta}{\eps} \right)^{\NN k/2}	
		~  =  ~ c_b \left( \frac{ \eps \left( \frac{(2^d+1)t}{c_b} \right)^{\N 2/k} }{\eps} \right)^{\N\NN k/2}
		~  =  ~ (2^d + 1) t.
\]
We can bound $P$'s diameter by observing that for all sufficiently small $\eps$
\begin{eqnarray*}
	\diam(P)
		& \le & \Delta 
		~  =  ~ \eps \left( \frac{(2^d + 1)t}{c_b} \right)^{\NN 2/k}
		~ \le ~ \eps \left( \frac{2^d + 1}{c_b \cdot \eps^{(d-1)/\alpha}} \right)^{\NN 2/\kappa} \\
		&  =  & \eps \left( \frac{2^d + 1}{c_b \cdot \eps^{(d-1)/\alpha}} \right)^{\NN \sqrt{2\alpha}/(d-1)}.
\end{eqnarray*}
(Here we made use of the fact that for all sufficiently small $\eps$, the quantity raised to power of $2/k$ is greater than $1$.) Letting $c_b' = ((2^d + 1)/c_b)^{\sqrt{2\alpha}/(d-1)}$, we obtain
\[
	\diam(P) 
		~ \le ~ c_b' \eps \left( \inv{\eps^{(d-1)/\alpha}} \right)^{\NN\sqrt{2\alpha}/(d-1)}
		~  =  ~ c_b' \eps^{1-\sqrt{2/\alpha}}.
\]
Since $\alpha \ge 4$, for all sufficiently small $\eps$, we have $\diam(P) \le 1/\sqrt{d}$. Therefore, $P$ can be enclosed within $Q_0^{(k+1)}$.

Returning to $\RE^d$, consider an infinite polyhedral hypercylinder whose ``axis'' is the $(d-1-k)$-dimensional orthogonal complement of $h$, and whose ``cross-section'' (i.e., intersection with any $(k+1)$-dimensional hyperplane parallel to $h$) is $P$. Define the polytope $C$ to be the truncated cylinder obtained by intersecting the infinite hypercylinder with hypercube $Q_0^{(d)}$ (see Figure~\ref{fig:cylinder}(a)). Let $T$ denote the output of $\alg(K,Q_0^{(d)})$ for $C$, $\eps$, and $t$. We will show that $T$'s total space satisfies the bound given in the lemma's statement. To do this, let $\Sigma$ denote any set of points placed on $C$'s axis such that the distance between each pair of points is at least $2\Delta\sqrt{d}$. (In the degenerate case where $k = d-1$ the axis is $0$-dimensional and $\Sigma$ degenerates to a single point.) By a simple packing argument, there exists such a set having $\Omega(1/\Delta^{d-1-k})$ points.

\begin{figure}[htbp]
  \centerline{\includegraphics[scale=0.40]{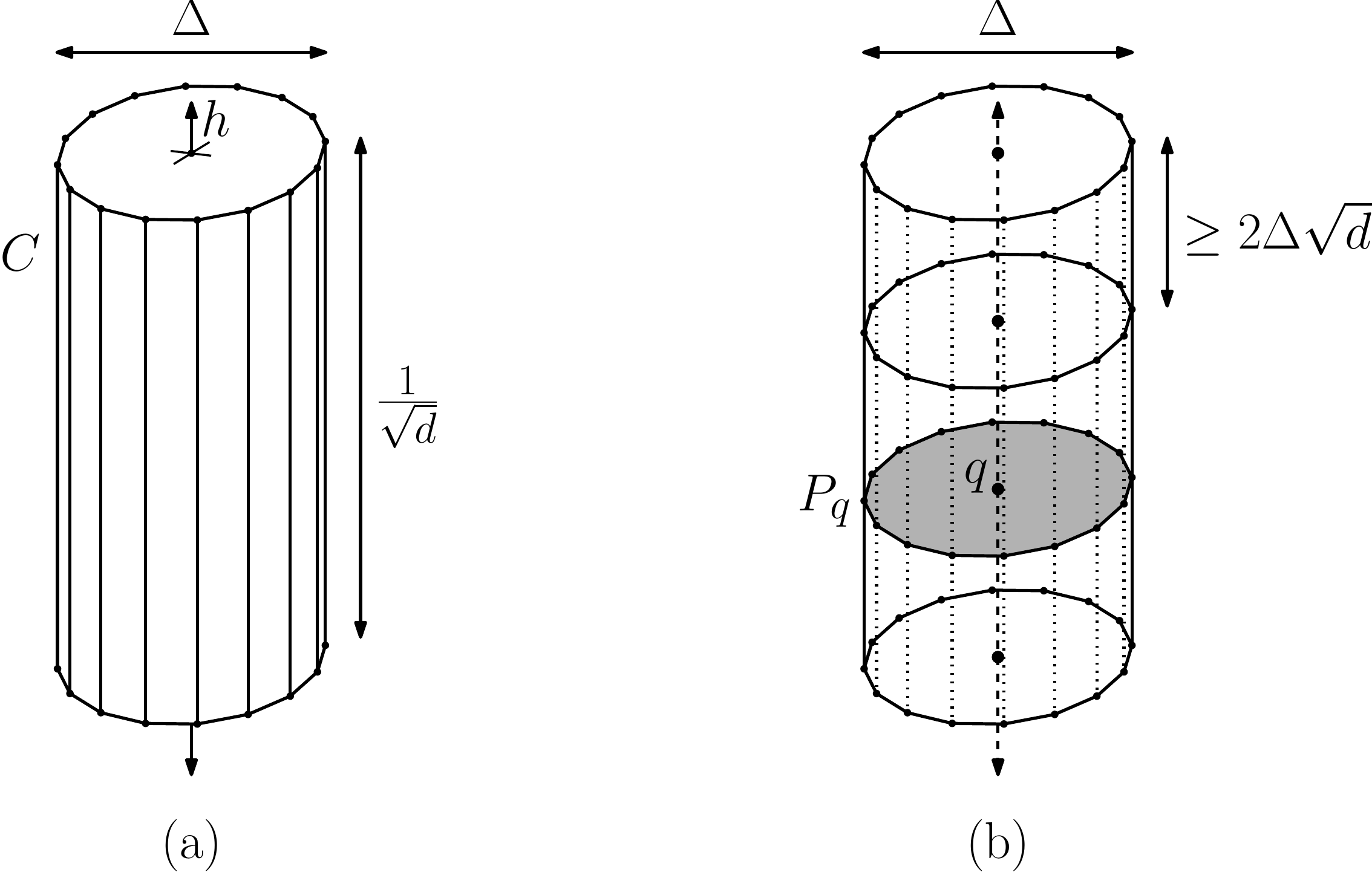}}
  \caption{Lemma~\ref{lem:cyl} for $d = 3$ and $k = 2$.}
  \label{fig:cylinder}
\end{figure}

For any $q \in \Sigma$, let $P_q$ denote the cross-section of $C$ passing through $q$ (see Figure~\ref{fig:cylinder}(b)). Consider the set of leaf cells of $T$ that intersect $P_q$. By applying Lemma~\ref{lem:ball} to the $(k+1)$-dimensional hyperplane on which $P$ lies, it follows that these cells together must contain at least $(2^d + 1) t$ halfspaces. We count the contributions of these cells by classifying them into two types. We say that a leaf cell of $T$ is \emph{large} if its side length is at least $\Delta$, and otherwise it is \emph{small}. By a simple packing argument, the number of large leaf cells intersecting $P_q$ is at most $2^d$. Since each leaf cell contains at most $t$ halfspaces, the large leaf cells can together contain at most $2^d t$ halfspaces. 

Therefore, the small leaf cells intersecting $P_q$ together contain at least $(2^d + 1) t - 2^d t = t$ halfspaces. Because the points of $\Sigma$ are separated from each other by distance at least $2\Delta\sqrt{d}$, which is strictly larger than the diameter of any small leaf cell, each small leaf cell can intersect $P_q$ for at most one $q \in \Sigma$. Therefore, the total space contribution of all the small leaf cells for all points of $\Sigma$ is at least $t \cdot |\Sigma|$. Let $c_b'' = (c_b/(2^d+1))^{2(d-1-k)/k}$. $T$'s total space can be asymptotically bounded from below as
\[
	\frac{t}{\Delta^{d-1-k}}
		~ = ~ \frac{t}{\left(\eps\left(\frac{(2^d+1)t}{c_b}\right)^{2/k}\right)^{d-1-k}}
		~ = ~ \frac{c_b'' \cdot t}{\left(\eps \cdot t^{2/k}\right)^{d-1-k}}
		~ = ~ \frac{c_b'' \cdot t^{1 - 2 (d-1-k)/k}}{\eps^{d-1-k}}.
\]
Clearly, $c_b'' = \Theta(1)$. Recall that $t = 1/\eps^{(d-1)/\alpha}$. Then, $T$'s total space is asymptotically bounded from below as
\begin{equation}
	\left(\inv{\eps}\right)^{(d-1) - k + \frac{d-1}{\alpha}\left(1 - \frac{2 (d-1-k)}{k}\right)}
	~ = ~ \left(\inv{\eps}\right)^{(d-1) - k + \frac{d-1}{\alpha}\left(3 - \frac{2(d-1)}{k}\right)} \label{eq:expo}
\end{equation}
Let $E(\alpha)$ denote this exponent. In order to complete the proof, we provide a lower bound on $E(\alpha)$. We use the fact that $\kappa \le k \le \kappa'$, apply the definitions of $\kappa$, $\kappa'$, and $\delta$, and straightforward manipulations to obtain
\begin{eqnarray*}
	E(\alpha)
		& \ge & (d-1) - \kappa' + \frac{d-1}{\alpha}\left(3 - \frac{2(d-1)}{\kappa}\right)
		~  =  ~ (d-1) \left(1 - \frac{2\sqrt{2\alpha} - 3}{\alpha} \right) - 1.
\end{eqnarray*}
Substituting this value for the exponent in Eq.~(\ref{eq:expo}) completes the proof.
\end{proof}

\section{Approximate Nearest Neighbor Searching} \label{sec:ann}

In this section, we present a reduction from approximate nearest neighbor searching to approximate polytope membership, which will allow us to prove Theorem~\ref{thm:ann-ub} from Section~\ref{sec:intro}. Our reduction will involve the following additional assumptions regarding the implementation of {\alg}. First, (as in Section~\ref{sec:preproc}) we assume that $K$ is presented as the intersection of $n$ halfspaces. Second, we assume that a leaf node is labeled as ``inside'' only if it lies entirely within $K$ (as opposed to lying within $K \oplus \eps$ as described in {\alg}). Third, we assume that leaf cells that store halfspaces use only bounding halfspaces of $K$.

Clearly, these assumptions do not affect the data structure's correctness. We assert that they do not affect the data structure's asymptotic query time or space bounds. Regarding the second assumption, observe that for any cell $Q$ that lies within $K \oplus \eps$, $K$ can be $\eps$-approximated within $Q$ using a single halfspace (any halfspace that contains $Q$ suffices). Regarding the third assumption, recall that Lemma~\ref{lem:discrete-apx} shows that we may assume that the approximating halfspaces for each node are drawn from the input halfspaces at the expense of a constant factor increase in the query time.

The reduction from approximate nearest neighbor searching to approximate polytope membership is based on the approximate Voronoi diagram (AVD) construction from~\cite{AVD-JACM}. The AVD employs a height balanced variant of a quadtree, a balanced box decomposition (BBD) tree~\cite{ARS} to be precise. Each cell of a BBD tree corresponds to the set theoretic difference of two quadtree cells, an \emph{outer box} and an optional \emph{inner box}. Each leaf cell of the tree stores a set of \emph{representative points} with the property that for any query point $q$ lying within this cell, at least one of these representatives is an $\eps$-nearest neighbor of $q$. A query is answered by locating the leaf cell that contains the query point and then determining the nearest representative from this cell (by brute force). The AVD's space is dominated by the total number of representatives over all the leaf cells. The query time is the height of the tree plus the number of representatives in the leaf cell. A data structure for nearest neighbor searching is said to be in the \emph{AVD model} if it has this general form, that is, a covering of the query region by hyperrectangles of bounded aspect ratio, each of which is associated with a set of representative points~\cite{AVD-JACM}. Lower bounds on the performance of any data structure in the AVD model were given in~\cite{AVD-JACM}.

The reader need not be familiar with the details of the AVD data structure. The next lemma encapsulates the important technical information needed for our reduction. It follows easily from the proofs of Lemmas~{6.1} and~{8.1} in~\cite{AVD-JACM}. Given a cell $Q$ in a BBD tree, let $B_Q$ denote the ball of radius $2 \cdot \diam(Q)$ whose center coincides with the center of $Q$'s outer box (see Figure~\ref{fig:separation}(a)). Given a Euclidean ball $B$ of radius $r$ and positive $c$, let $c \P B$ denote the ball concentric with $B$ of radius $c \P r$.

\begin{lemma} \label{lem:avd}
Let $0 < \eps \leq 1/2$ be a real parameter and $X$ be a set of $n$ points in $\RE^d$. It is possible to construct a BBD tree $T$ with $O(n \cdot \log \inv{\eps})$ nodes, where each leaf cell $Q$ stores a subset $R_Q \subset X$ satisfying the following properties:

\begin{enumerate}
\setlength{\itemsep}{-0.5ex}%
\setlength{\parsep}{0pt}%
\item[$(i)$] For any point $q$ in $Q$, one of the points in $R_Q$ is an $\eps$-approximate nearest neighbor of $q$.

\item[$(ii)$] At most one point of $R_Q$ is contained in the ball $B_Q$, and the remaining points of $R_Q$ are contained in $c_q B_Q \setminus B_Q$ for some constant $c_q$ (which depends on the dimension).

\item[$(iii)$] The total number of representative points over all the leaf cells of $T$ is $O(n \cdot \log \inv{\eps})$.
\end{enumerate}

Moreover, it is possible to compute the tree $T$ and the sets $R_Q$ for all the leaf cells in total time $O(n \cdot \log n \cdot \log \inv{\eps})$, and the cell that contains a query point can be located in time $O(\log n + \log \log \inv{\eps})$.
\end{lemma}

\begin{figure}[htbp]
  \centerline{\includegraphics[scale=0.40]{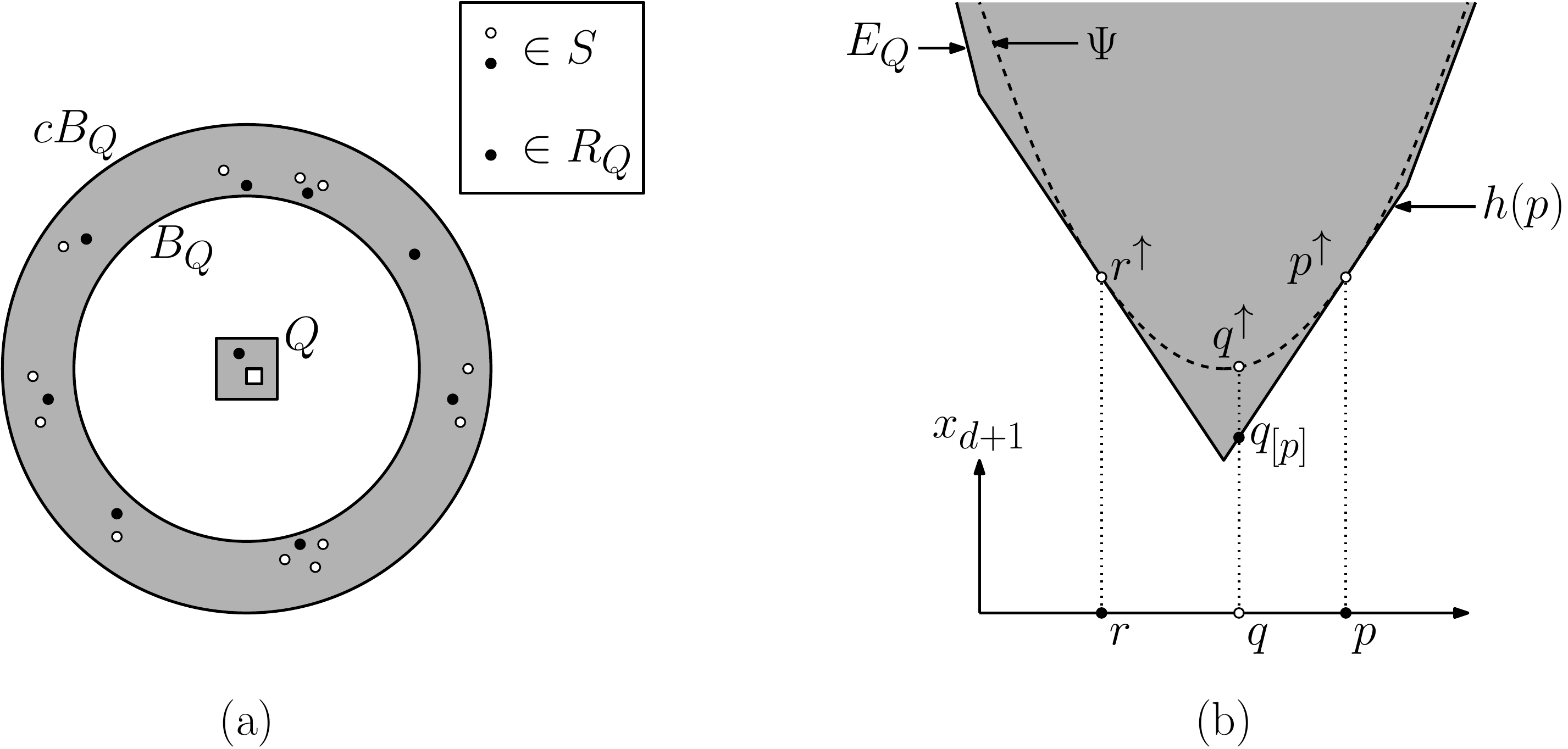}}
  \caption{Approximate nearest neighbor searching: (a) Lemma~\ref{lem:avd} (black points are members of $R_Q$), (b) the lifting transformation. (Note that the figure is not drawn to scale, and the paraboloid in (b) has been translated to aid legibility.)}
  \label{fig:separation}
\end{figure}

In the AVD data structure of \cite{AVD-JACM} the closest representative point to a query point is determined by brute-force enumeration of the elements of $R_Q$. We consider whether it is possible search them more efficiently by reduction to polytope approximation. The following lemma explains how to connect Lemma~\ref{lem:avd} with approximate polytope membership queries. Our construction uses the well known \emph{lifting transformation}~\cite{ray-shooting-NN,edels}. Let $(x_1, \ldots, x_{d+1})$ denote the coordinates of $\RE^{d+1}$, and let us think of $(d+1)$st coordinate axis as being directed vertically upwards. Let $\Psi$ denote the paraboloid $x_{d+1} = \sum_{i=1}^d x_i^2$. Given a point $p \in \RE^d$, let $p^{\uparrow}$ denote the vertical projection of $p$ onto $\Psi$ (see Figure~\ref{fig:separation}(b)), and let $h(p)$ denote the hyperplane tangent to $\Psi$ at $p^{\uparrow}$. That is, the points of $h(p)$ satisfy $x_{d+1} = \sum_{i=1}^d 2 p_i x_i - \|p\|^2$. Given $q \in \RE^d$, let $q_{[p]}$ denote the point on $h(p)$ hit by a vertical ray shot downwards from $q^{\uparrow}$. A straightforward consequence of the definition of $\Psi$ is that the squared distance between $q$ and $p$ in $\RE^d$ is equal to the length of this vertical segment, that is, $\|q p\|^2 = \|q^{\uparrow} q_{[p]}\|$. 

This suggests the following approach to computing the closest representative point through vertical ray shooting. Consider the (unbounded) convex polyhedron that results by taking the upper envelope of the hyperplanes $h(p)$ associated with the lifted representatives. Given the query point $q \in \RE^d$, a ray shot vertically downward from $q^{\uparrow}$ hits some facet of this polyhedron. It follows from the above remarks, that the representative associated with this hyperplane is the closest to $q$. We can simulate ray shooting by applying polytope membership queries in concert with binary search. Of course, some care will be needed to map this problem into our context, which assumes a bounded polytope and approximation.

\begin{lemma} \label{lem:mini-reduction}
Let $0 < \eps \leq 1/2$ be a real parameter and consider a quadtree cell $Q$ and a set of representative points $R_Q$ as in Lemma~\ref{lem:avd}. Given a data structure for $\eps$-approximate polytope membership in $d$-dimensional space with query time $t_{d}(\eps)$ and space $s_{d}(\eps)$, it is possible to preprocess $R_Q$ into an ANN data structure for query points in $Q$ with query time $O(t_{d+1}(\eps) \cdot \log \inv{\eps})$ and space $O(s_{d+1}(\eps))$.
\end{lemma}

\begin{proof}
Since at most one point of $R_Q$ is contained in $B_Q$, the corresponding point may be inspected separately without increasing the complexity bounds. Therefore, we may assume that all points of $R_Q$ are contained in $c_q B_Q \setminus B_Q$. 

Although we assume that the errors in polytope membership are absolute (because of standardization), errors in approximate nearest neighbor searching are relative. That is, a point $r$ is an $\eps$-approximate nearest neighbor of $q$ if $\|q r\| \le (1+\eps) \|q p\|$, where $p$ is $q$'s true nearest neighbor. Because errors are relative, we may assume that space has been translated and uniformly scaled so that $Q$ is mapped to $Q_0^{(d)}$, the hypercube of unit diameter centered at the origin in $\RE^d$. As a result, $B_Q$ is mapped to a ball of radius $2$. It follows that the distance from any point of $Q$ to any point of $R_Q$ is greater than $1$. Therefore, an absolute error of $\eps$ implies a relative error of at most $\eps$.

In order to reduce nearest neighbor searching among the points of $R_Q$ to polytope membership, let $E_Q$ denote the upper envelope, that is, the intersection of the upper halfspaces, of the hyperplanes $h(p)$, for all $p \in R_Q$ (the shaded region in Figure~\ref{fig:separation}(b)). As mentioned above, the facet of $E_Q$ hit by shooting a ray vertically downward from $q^{\uparrow}$ corresponds to the closest point of $R_Q$ to $q$. 

Since the upper envelope is unbounded, we first compute a bounded convex polytope on which to perform approximate membership queries. Because the query points lie in $Q$, we are only interested in the portion of $E_Q$ that projects vertically onto $Q$. Given that the distance of any point $p \in R_Q$ to the origin is at most $2 c_q = O(1)$, it follows that the portion of $E_Q$ of interest fits within an axis-aligned $(d+1)$-dimensional hypercube of constant diameter that is centered at the origin. Let $Q'$ denote such a hypercube, let $K_Q = E_Q \cap Q'$, and let $\eps' = \eps/6 c_q$. We invoke {\alg} to construct an $\eps'$-approximate membership data structure for $K_Q$. (More formally, we first scale $Q'$ into standard form, and we scale $\eps'$ by the same factor. We then apply {\alg} with the scaled value of $\eps'$. Since $Q'$ is of constant diameter, the scale factor will also be a constant, and therefore only the constant factors in the analysis will be affected. We then apply an inverse scaling to obtain the desired $\eps'$-approximating polytope for $K_Q$.)

We simulate the ray shooting process by a binary search to locate the contact point approximately. Consider the vertical segment formed by intersecting $Q'$ with the vertical line passing through $q^{\uparrow}$. The upper endpoint of this segment is clearly inside $K_Q$ and its lower endpoint is outside. We repeatedly split the segment at its midpoint, perform an approximate polytope membership query, and retain the subsegment whose upper endpoint is (approximately) inside $K_Q$ and whose lower endpoint is (approximately) outside. We terminate the search when the length of the segment falls below $\eps'$. Since $Q'$ is of constant diameter, the search terminates after $O(\log \inv{\eps})$ membership queries. Let us denote the endpoints of this final segment as $q^+$ (upper) and $q^-$ (lower).

Recall our assumption that cells are labeled by {\alg} as ``inside'' or ``outside'' only if they lie entirely inside or outside $K_Q$, respectively. It follows that as we traverse the cells that intersect the segment $q^+ q^-$ from top to bottom, we cannot transition directly from an ``inside'' cell to an ``outside'' cell. Therefore, at least one of these cells must contain a set of representative hyperplanes. Let $h(r)$ denote the hyperplane having the topmost intersection with the vertical ray. We return $r$ as the approximate nearest neighbor (see Figure~\ref{fig:paraboloid}). It is easy to see that this algorithm satisfies the desired time and space bounds. 

\begin{figure}[htbp]
  \centerline{\includegraphics[scale=0.40]{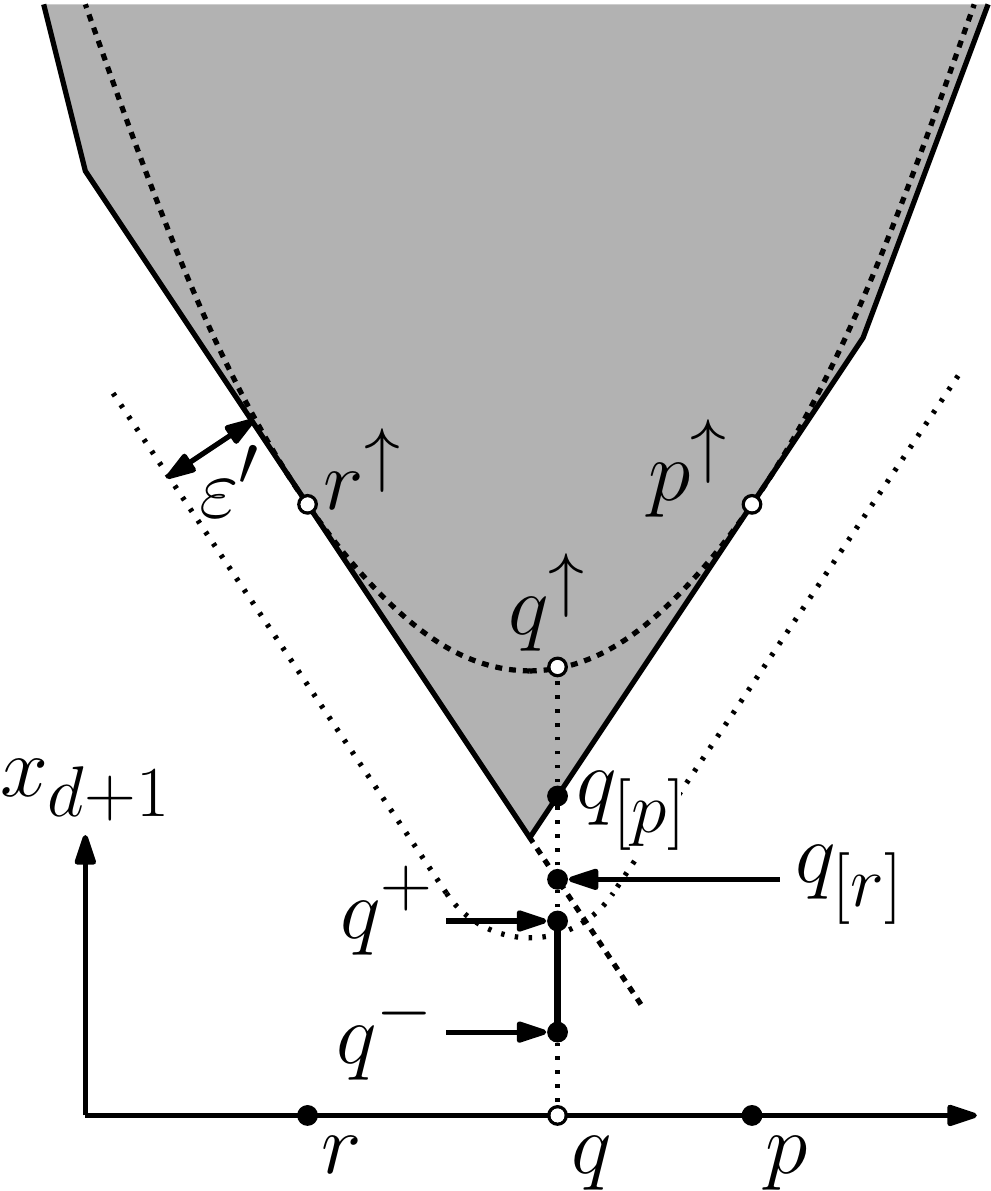}}
  \caption{Proof of Lemma~\ref{lem:mini-reduction}. (Not drawn to scale.)}
  \label{fig:paraboloid}
\end{figure}

All that remains is to establish correctness, by showing that $r$ is indeed an $\eps$-approximate nearest neighbor of $q$. In order to do this, let $p$ be $q$'s true nearest neighbor in $R_Q$. Due to the nature of the binary search, $q^+$ lies within distance $\eps'$ of $K_Q$. (Note that it might lie within $K_Q$.) Thus, the distance from $q^+$ to the upper halfspace bounded by $h(p)$ is at most $\eps'$. By the triangle inequality, the distance from $q^-$ to this halfspace is at most $\eps' + \eps' = 2 \eps'$. Since $p$ is $q$'s true nearest neighbor, $q_{[p]}$ lies on $\partial K_Q$, and so the hyperplane $h(p)$ separates $q^-$ from $K_Q$. This implies that the distance from $q^-$ to $h(p)$ is also not greater than $2 \eps'$. 

We claim that the vertical distance from $q^-$ to $q_{[p]}$ is at most $\eps$. To see why, recall that $p$ lies within a ball of radius $2 c_q$ centered at the origin. This implies that $h(p)$ cannot be too steep, that is, the angle formed between $h(p)$'s normal vector and the vertical axis can be bounded away from $\pi/2$ by a constant. By basic linear algebra, it can be shown that the ratio of the vertical and orthogonal distances of any point to $h(p)$ is bounded above by $\sqrt{4 c_q^2 + 1} < 3 c_q$. Therefore, we have $\|q_{[p]} q^-\| \le 3 \P c_q \P (2 \eps') = \eps$, as desired. 

Because $r$ is the witness produced by the algorithm, $h(r)$ separates $q^-$ from $K_Q$, which implies that $q_{[r]}$ lies above $q^-$. Thus, we have $\|q_{[p]} q_{[r]}\| \le \|q_{[p]} q^-\| \le \eps$. Therefore,
\[
	\|q r\|^2
		~  =  ~ \|q^{\uparrow} q_{[r]}\|
		~  =  ~ \|q^{\uparrow} q_{[p]}\| + \|q_{[p]} q_{[r]}\|
		~ \le ~ \|q^{\uparrow} q_{[p]}\| + \eps.
\]
By the lifting transformation, we have $\|q^{\uparrow} q_{[p]}\| = \|q p\|^2$, and combining this with the fact that $\|q p\| \ge 1$, we have
\[
	\|q r\|^2
		~ \le ~ \|q p\|^2 + \eps 
		~ \le ~ \|q p\|^2 + \|q p\|^2 \eps
		~  =  ~ \|q p\|^2 (1 + \eps)
		~ \le ~ \left( \|q p\| (1 + \eps) \right)^2.
\]
Therefore, $r$ is an $\eps$-approximate nearest neighbor of $p$, which completes the proof.
\end{proof}

The above lemma shows how to apply approximate polytope membership to efficiently answer approximate nearest neighbor queries within each cell of the AVD. To obtain a complete data structure for approximate nearest neighbor searching we apply this to every leaf cell of the AVD.

\begin{lemma} \label{lem:reduction}
Let $0 < \eps \leq 1/2$ be a real parameter and $X$ be a set of $n$ points in $\RE^d$. Given a data structure for approximate polytope membership in $d$-dimensional space with query time at most $t_{d}(\eps)$ and storage $s_{d}(\eps)$, it is possible to preprocess $X$ into an ANN data structure with query time $O(\log n + t_{d+1}(\eps) \cdot \log \inv{\eps})$ and space
\[
	\OO{n\,\log\inv{\eps} + n \, \frac{s_{d+1}(\eps)}{t_{d+1}(\eps)}}.
\]
\end{lemma}

\begin{proof}
Following Lemma~\ref{lem:avd}, construct a BBD-tree $T$, and for each leaf cell $Q$ of $T$, construct the set of representative points $R_Q$. For each leaf cell such that $|R_Q| \leq t_{d+1}(\eps) \cdot \lg \inv{\eps}$, simply store the set $R_Q$ and answer the corresponding queries by brute force. For the nodes with $|R_Q| > t_{d+1}(\eps) \cdot \lg \inv{\eps}$, use the construction from Lemma~\ref{lem:mini-reduction}. 

To answer an ANN query we search the AVD of Lemma~\ref{lem:avd} to find the leaf cell containing the query point and then apply Lemma~\ref{lem:mini-reduction}. Thus, the query time is 
\[
	O\left(\log n + \log \log \inv{\eps} + t_{d+1}(\eps) \cdot \log \inv{\eps}\right) 
		~ = ~ O\left(\log n + t_{d+1}(\eps) \cdot \log \inv{\eps}\right).
\]
To bound the total space, observe from Lemma~\ref{lem:avd}(iii) that the total number of representative points is $O(n \log \inv{\eps})$. Thus, by a simple counting argument, the number of leaf cells with more than $t_{d+1}(\eps) \cdot \lg \inv{\eps}$ representatives is $O(n/t_{d+1}(\eps))$. Therefore, the total space of the data structure is $O(n \log \inv{\eps} + n (s_{d+1}(\eps)/t_{d+1}(\eps)))$.
\end{proof}

Because of its reliance on binary search, the generic reduction given in Lemmas~\ref{lem:mini-reduction} and~\ref{lem:reduction} is not formally in the AVD model. Recall that the AVD model is important because lower bounds have been established in this model~\cite{AVD-JACM}, and thus these lower bounds do not apply here. However, by sacrificing generality and a factor of $O(\log \inv{\eps})$ in the space bound, we can exploit the properties of {\alg} to obtain a data structure that is in the AVD model.

\begin{lemma} \label{lem:avd-model}
Let $0 < \eps \leq 1/2$ be a real parameter and $X$ be a set of $n$ points in $\RE^d$. Given a split-reduce data structure for approximate polytope membership in $d$-dimensional space with query time at most $t_{d}(\eps)$ and storage $s_{d}(\eps)$, it is possible to preprocess $X$ into an ANN data structure in the AVD model with query time $O(\log n + t_{d+1}(\eps) \cdot \log \inv{\eps})$ and space
\[
	\OO{n \left( 1 + \frac{s_{d+1}(\eps)}{t_{d+1}(\eps)} \right) \log \inv{\eps}}.
\]
\end{lemma}

\begin{proof}
As in Lemma~\ref{lem:reduction}, construct a BBD-tree $T$, and for each leaf cell $Q$ of $T$, construct the set of representative points $R_Q$. We may assume that $|R_Q| > t_{d+1}(\eps) \cdot \lg \inv{\eps}$, since otherwise we just use the points of $R_Q$ as the representatives. In order to handle query points lying within $Q$, we apply Lemma~\ref{lem:mini-reduction}, where queries are answered using the tree produced by {\alg}. Let $T_Q$ denote the resulting tree. We exploit the fact that the {\alg} data structure associates a collection of hyperplanes with each leaf cell of $T_Q$, and by the nature of our reduction, each of these hyperplanes corresponds to a lifted point of $R_Q$. These lifted points will play the role of nearest neighbor representatives. Intuitively, our approach is to ``undo'' the lifting transformation by projecting the leaf cells of $T_Q$ vertically from $\RE^{d+1}$ down to $\RE^d$ and then building a $d$-dimensional AVD structure based on this projection. 

The projection of the cells of $T_Q$ onto $\RE^d$ naturally defines a quadtree subdivision of $\RE^d$, which we denote by $T'_Q$ (see Figure~\ref{fig:avd-model}(a)). For each leaf cell $Q'$ of $T'_Q$, let $C_{Q'}$ denote the infinite vertical cylinder in $\RE^{d+1}$ whose cross section is $Q'$ (see Figure~\ref{fig:avd-model}(b)). Because $Q'$ is a leaf, any leaf cell of $T_Q$ that intersects this cylinder projects onto a hypercube that contains $Q'$. 

\begin{figure}[htbp]
  \centerline{\includegraphics[scale=0.40]{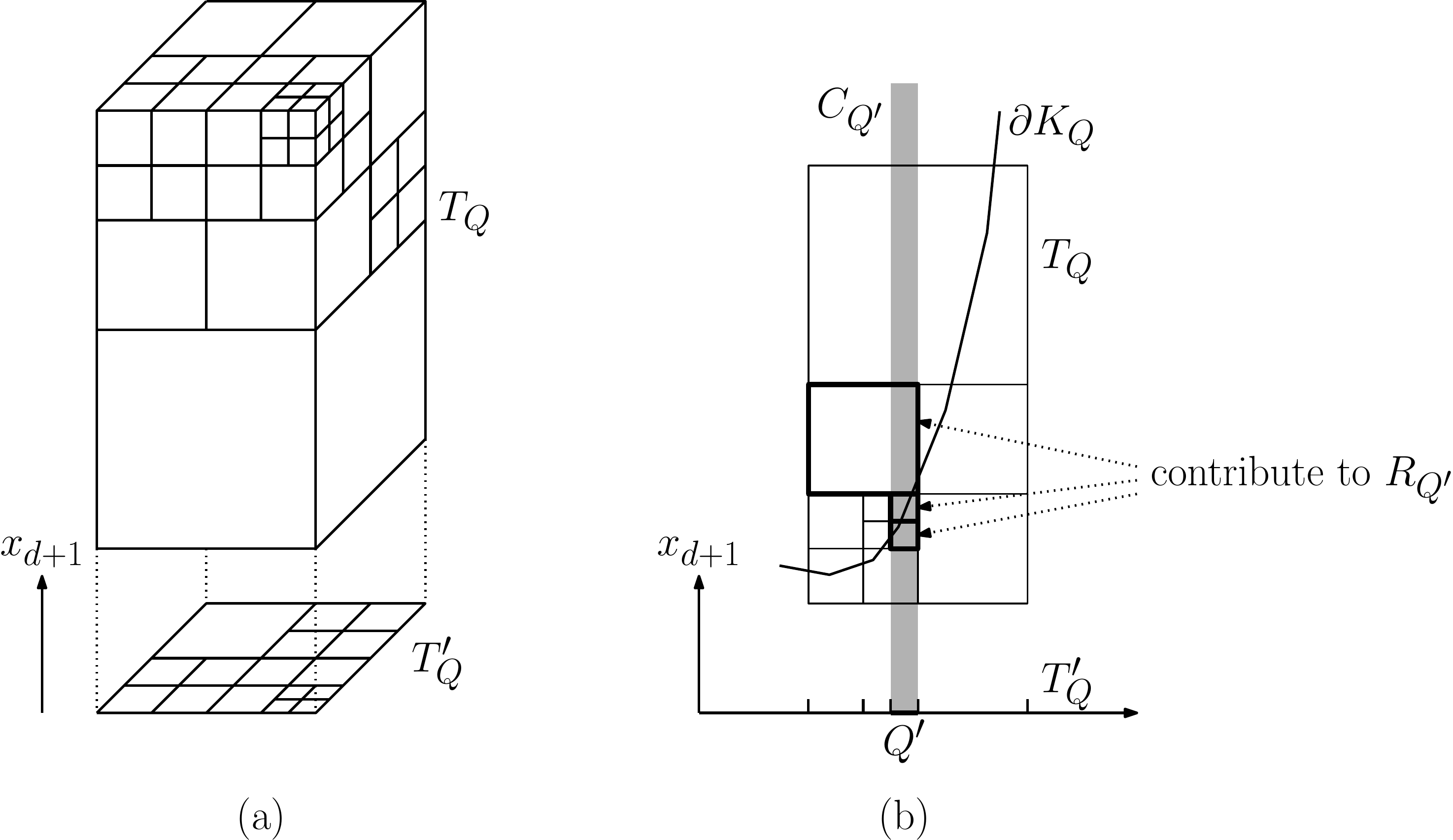}}
  \caption{Producing an ANN data structure in the AVD model.}
  \label{fig:avd-model}
\end{figure}

Recall the lifted polytope $K_Q$ of Lemma~\ref{lem:mini-reduction}. For each leaf cell of $T_Q$ that contains a point whose vertical distance from $\partial K_Q$ is at most $\eps$, we create a representative point corresponding to each of the hyperplanes that {\alg} associates with this leaf cell. We denote the resulting collection of representatives by $R_{Q'}$. These are the only hyperplanes that are relevant to the binary search of Lemma~\ref{lem:mini-reduction}, and therefore one of them will provide the final witness in the binary search (the point $r$ in the proof of Lemma~\ref{lem:mini-reduction}). This implies that $R_{Q'}$ constitutes a valid representative set for $\eps$-approximate nearest neighbor searching for any query point that lies in $Q'$. Thus, the resulting data structure is a valid AVD structure. 

In order to bound the query time we recall some of the observations made in the proof of Lemma~\ref{lem:mini-reduction}. Since $K_Q$ is contained within a hypercube of constant diameter centered at the origin, the absolute slopes of the hyperplanes of the approximating polytope are bounded above by some constant. Recall that the leaf cells of $T_Q$ that contribute a point to $R_{Q'}$ have side lengths at least as large as that of $Q'$. By the same reasoning used in Lemma~{3} of \cite{ARS}, the number of such quadtree leaf cells that can intersect $\partial K_Q$ is bounded by a constant, which we denote by $c_{\ell}$. (This constant depends on the dimension $d$ and the largest possible slope.) Therefore, the total number of cells contributing a representative to $R_{Q'}$ is at most $c_{\ell}$. Since each cell contributes at most $t_{d+1}(\eps)$ representatives, the total number of representatives associated with any leaf cell of $T'_Q$ is at most $c_{\ell} \cdot t_{d+1}(\eps) = O(t_{d+1}(\eps))$.

The bound on the total space is complicated by the fact that a large cell that intersects $\partial K_Q$ may overlap the columns of many small leaf cells, and hence a large cell's representatives may be replicated many times. Let $M$ denote the set of internal nodes of $T_Q$ all of whose children are leaves. We encountered this set earlier in the proof of Lemma~\ref{lem:total-space}. As we saw in that earlier lemma, because each node of $M$ was split by {\alg}, it follows that each such cell requires more than $t_{d+1}(\eps)$ halfspaces to approximate $K(Q)$, and thus, the children of $M$ together require at least as many representatives. Therefore we have $|M| \cdot t_{d+1}(\eps) \le s_{d+1}(\eps)$. Reasoning as we did in Lemma~\ref{lem:total-space}, every internal node of $T_Q$ is either in $M$ or is an ancestor of a node in $M$. Thus, the number of internal nodes is at most $|M| \cdot \textrm{height}(T_Q)$. Since every internal node has $2^d$ children, the total number of nodes in $T_Q$ is at most $2^d \cdot |M| \cdot \textrm{height}(T_Q)$. Clearly, the number of leaf cells of $T'_Q$ can be no larger. As we saw in the previous paragraph, each leaf cell of $T'_Q$ is associated with at most $c_{\ell} \cdot t_{d+1}(\eps)$ representatives. Since the tree is of height $O(\log \inv{\eps})$, the total number of representatives over all these cells is at most
\begin{eqnarray*}
	(2^d \cdot |M| \cdot \textrm{height}(T_Q)) (c_{\ell} \cdot t_{d+1}(\eps))
		&  =  & c_{\ell} \cdot 2^d \cdot \textrm{height}(T_Q) \cdot (|M| \cdot t_{d+1}(\eps)) \\
		& \le & (c_{\ell} \cdot 2^d \cdot {\textstyle \log \inv{\eps}}) \cdot s_{d+1}(\eps) \\
		&  =  & O(s_{d+1}(\eps) \cdot {\textstyle \log \inv{\eps}}).
\end{eqnarray*}
By Lemma~\ref{lem:avd}(iii), the total number of representatives in $T_Q$ is $O(n \log \inv{\eps})$. By a counting argument, the number of leaf cells with more than $t_{d+1}(\eps) \cdot \log \inv{\eps}$ representatives is $O(n/t_{d+1}(\eps))$. Therefore, the total space is
\[
	\OO{n \log \inv{\eps} + \frac{n}{t_{d+1}(\eps)} \cdot s_{d+1}(\eps) \cdot \log \inv{\eps}}
	~ = ~ \OO{n  \left( 1 + \frac{s_{d+1}(\eps)}{t_{d+1}(\eps)} \right) \log\inv{\eps} }
\]
as desired.
\end{proof}

By combining this with Theorem~\ref{thm:membership-ub} (applying the more accurate space bounds from Lemma~\ref{lem:trade-off-ub}) we obtain the main result of this section.

\begin{lemma} \label{lem:ann-ub}
Let $0 < \eps \leq 1$ be a real parameter, $\alpha \geq 1$ be a real constant, and $X$ be a set of $n$ points in $\RE^d$. There is a data structure in the AVD model for approximate nearest neighbor searching that achieves
\begin{eqnarray*}
\text{Query time:} & & O\big(\log n + (1/\eps^{d/2 \alpha}) \cdot \log^2 {\textstyle \inv{\eps}}\big) \\
\text{Space:} & & 
  \OO{n \cdot \max\left(\log {\textstyle \inv{\eps}}, 1/\eps^{d\left(\inv{2} -\inv{2\alpha}\right)}\right)}\N , \text{~for $1 \leq \alpha < 2$ and,} \\
  & &
  \OO{n/\eps^{d \left(1 - \frac{\floor{\lg\alpha}}{\alpha} - \inv{2^{\floor{\lg \alpha}}} + \inv{2 \alpha}\right)}}\N , \text{~for $\alpha \geq 2$}.
\end{eqnarray*}
The constant factors in the space and query time depend only on $d$ and $\alpha$ (not on $\eps$). At the expense of increasing the query time and space by a factor of $O(\log \inv{\eps})$ it is possible to construct the data structure in time $O(n (\log n + 1/\eps^{c \P d}) \log\inv{\eps})$, for some constant $c$ (that does not depend on $d$ or $\alpha$).
\end{lemma}

\begin{proof}
Given $X$ and $\eps$, we first observe that if $1/16 < \eps \le 1$, we may set $\eps = 1/16$, since this will only affect the constant factors in the asymptotic bounds. We consider two cases based on the value of $\alpha$. 

If $1 \le \alpha < 2$, we will apply Theorem~\ref{thm:membership-ub} with the values of $d$ and $\alpha$ of the theorem set to $d' = d+1$ and $\alpha' = 4$, respectively. The theorem states that there is a data structure that achieves query time 
\begin{equation}
	O\left(({\textstyle \log \inv{\eps}})/\eps^{\frac{d'-1}{\alpha'}}\right)
		~ = ~ O\left((1/\eps^{d/4}) \cdot \log {\textstyle \inv{\eps}}\right)
		~ = ~ O\left((1/\eps^{d/2 \alpha}) \cdot \log {\textstyle \inv{\eps}}\right) \label{eq:ann-ub-time}
\end{equation}
and space
\begin{equation}
	\OO{1/\eps^{(d'-1)\left(1 - \frac{2\floor{\lg \alpha'} - 2}{\alpha'} \right)}}
		~ = ~ O\big( 1/\eps^{d/2} \big). \label{eq:ann-ub-space}
\end{equation}
Letting $t_{d+1}(\eps)$ and $s_{d+1}(\eps)$ denote the quantities of Eqs.~\eqref{eq:ann-ub-time} and~\eqref{eq:ann-ub-space}, respectively, we apply Lemma~\ref{lem:avd-model} to obtain a data structure in the AVD-model with query time $O\big(\log n + (1/\eps^{d/2 \alpha}) \cdot \log^2 \inv{\eps}\big)$ and space
\[
	\OO{n \left( 1 + \frac{1/\eps^{d/2}}{(1/\eps^{d/2 \alpha}) \cdot \log \inv{\eps}} \right) \log \inv{\eps}}
	~ = ~
	\OO{n \cdot \max\left(\log {\textstyle \inv{\eps}}, 1/\eps^{d\left(\inv{2} -\inv{2\alpha}\right)}\right)},
\]
as desired.

Otherwise, if $\alpha \ge 2$, we apply Theorem~\ref{thm:membership-ub} (but using the more accurate space bounds from Lemma~\ref{lem:trade-off-ub}) in dimension $d' = d+1$ and with trade-off parameter $\alpha' = 2 \alpha$. (Observe that $\alpha' \ge 4$, as required by Theorem~\ref{thm:membership-ub} and Lemma~\ref{lem:trade-off-ub}.) This yields an approximate polytope membership data structure with query time $t_{d+1}(\eps) = O\big((1/\eps^{d/2 \alpha}) \cdot \log \inv{\eps} \big)$ and space
\[
  s_{d+1}(\eps) 
	~ = ~ \OO{1/\eps^{d \left(1 - 2 \left(\frac{\floor{\lg (2\alpha)} - 2}{2\alpha} + \inv{2^{\floor{\lg(2\alpha)}}} \right) \right)}}
	~ = ~ \OO{1/\eps^{d \left(1 - \frac{\floor{\lg \alpha} - 1}{\alpha} - \inv{2^{\floor{\lg \alpha}}} \right)}}.
\]
By Lemma~\ref{lem:avd-model} this implies the existence of a data structure in the AVD-model with the desired query time of $O\big(\log n + (1/\eps^{d/2 \alpha}) \cdot \log^2 \inv{\eps}\big)$ and space
\[
	\OO{n \left( 1 + \frac{1/\eps^{d \left(1 - \frac{\floor{\lg \alpha} - 1}{\alpha} - \inv{2^{\floor{\lg \alpha}}} \right)}}{({\textstyle \log \inv{\eps}})/\eps^{d/2 \alpha}} \right) \log\inv{\eps}}.
\]
Since $\alpha \ge 2$, we may ignore the ``$1+$'' term in the inner parenthetical factor. After some simplification we obtain the desired space bound of
\[
	\OO{n/\eps^{d \left(1 - \frac{\floor{\lg \alpha}}{\alpha} - \inv{2^{\floor{\lg \alpha}}} + \inv{2 \alpha} \right)}}.
\]

The preprocessing involves first computing the AVD, which by Lemma~\ref{lem:avd} takes $O(n \cdot \log n \cdot \log \inv{\eps})$ time. For each of the $O(n \log \inv{\eps})$ leaf cells $Q$ of the AVD, we apply {\alg} in dimension $d+1$ to its associated set $R_Q$ of representatives. By Lemma~\ref{lem:preproc-time} this takes $O(n_Q + 1/\eps^{c_p(d+1)})$ time, where $n_Q = |R_Q|$, and $c_p$ is a constant that does not depend on $d$. Summing over all the leaf cells of the AVD and recalling that the total number of representatives is $O(n \cdot \log \inv{\eps})$, it follows that the total preprocessing time is on the order of
\begin{eqnarray*}
  n \cdot \log n \cdot \log \inv{\eps} ~+~ \sum_Q \big(n_Q + 1/\eps^{c_p(d+1)}\big)
	& = & n \cdot \log n \cdot \log \inv{\eps} + n \cdot \log \inv{\eps} \cdot \left(\inv{\eps}\right)^{\N c_p(d+1)} \\
	& = & n \left(\log n + \left( \inv{\eps}\right)^{\N c \P d} \right) \log \inv{\eps},
\end{eqnarray*}
where $c = c_p(d+1)/d$, as desired. Because of the reliance on approximate set cover in the processing of Lemma~\ref{lem:preproc-time}, the query time and space are larger by a factor of $O(\log \inv{\eps})$.
\end{proof}

Note that the above proof uses the AVD-based reduction given in Lemma~\ref{lem:avd-model}. If instead we had used Lemma~\ref{lem:reduction}, we would obtain a slight improvement in the space, by a factor of $\Theta(\log \inv{\eps})$, at the loss of having a data structure in the AVD model. By the simple observation that $1/2^{\floor{\lg \alpha}} \ge 1/\alpha$, the above space bound for the $\alpha \ge 2$ case simplifies to $O\big( n/\eps^{d (1 - \frac{\floor{\lg\alpha}}{\alpha} - \inv{2 \alpha})}\big)$, and this establishes Theorem~\ref{thm:ann-ub}.

\section{Proof of the Area-Product Bound} \label{sec:proof}

In this section, we present lower bounds for the product of the area of (restricted) $\eps$-dual caps and the associated Voronoi patches, and in particular, we present a proof of Lemma~\ref{lem:dual-basic}, which appeared at the end of Section~\ref{sec:cap}.

We begin by recalling some notation. We are given a convex body $K$ in $\RE^d$, and a pair $(p,h(p))$ where $p \in \partial K$ and $h(p)$ is a supporting hyperplane passing through $p$, such that $p$ lies within a unit ball centered at the origin. Also recall that $p_{\eps}$ denotes the point lying at distance $\eps$ from $p$ in the direction of the outward normal orthogonal to $h(p)$ at $p$. $S$ denotes the Dudley hypersphere, which is centered at the origin and is of radius $3$. For $y \ge 1$, let $H^{(y)}(p)$ be any hyperplane that is parallel to $h(p)$ and translated away from $K$ by distance $y$. (This is illustrated in Figure~\ref{fig:area-bound-setup}. Note that the figures of this section are not drawn to scale.) To simplify our descriptions, we consider the directed line segment from $p$ to $p_{\eps}$ to be ``vertically downwards,'' so that the hyperplanes $h(p)$ and $H^{(y)}(p)$ are ``horizontal'' with $h(p)$ above $H^{(y)}(p)$.

\begin{figure}[htbp]
  \centerline{\includegraphics[scale=0.40]{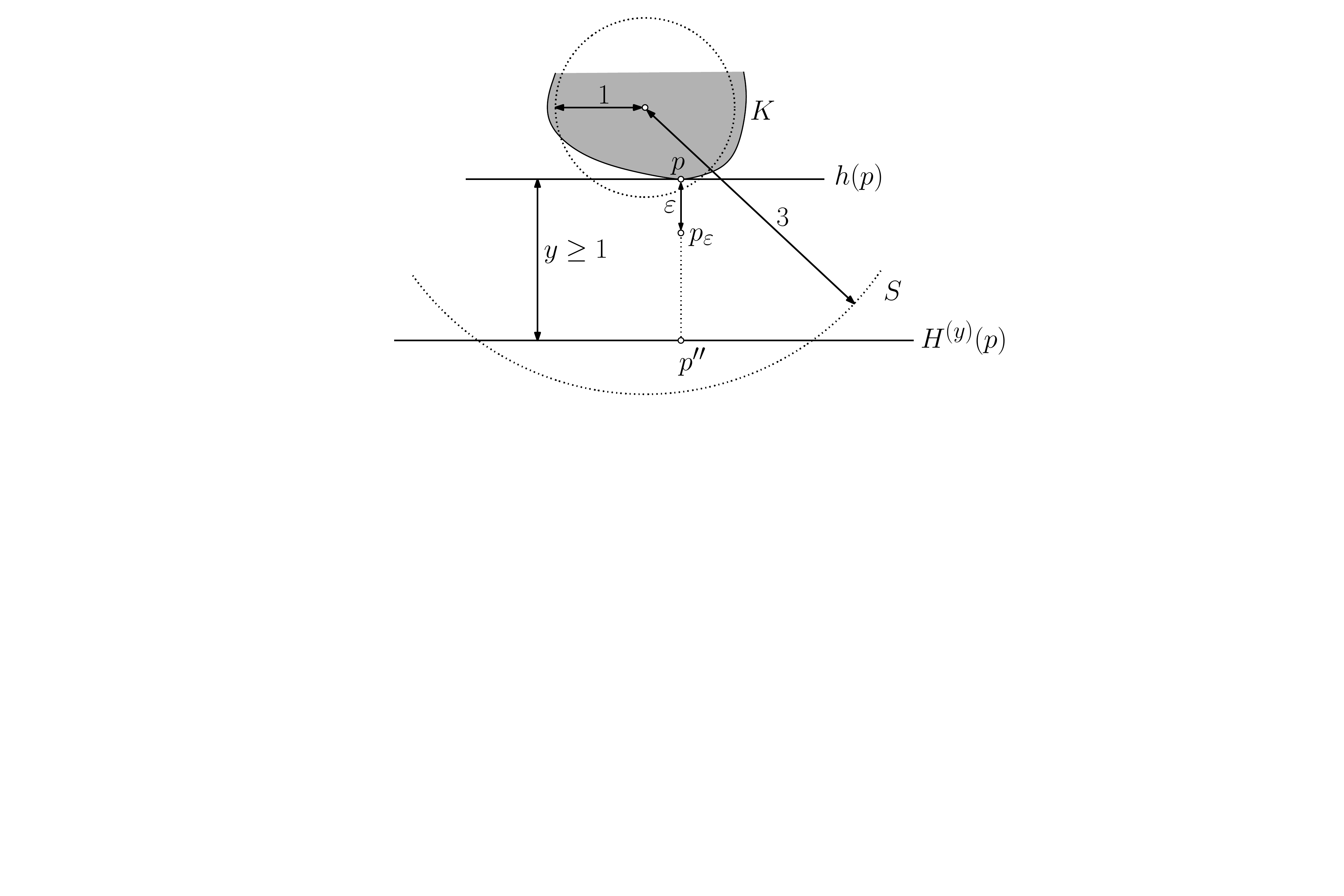}}
  \caption{Definitions of $h(p)$, $H^{(y)}(p)$ and $S$.}
  \label{fig:area-bound-setup}
\end{figure}

Recall that the \emph{$\eps$-dual cap} defined by $p$, denoted $D(p)$, is the portion of $\partial K$ that is visible from $p_{\eps}$ (see Figure~\ref{fig:voronoi}(a)). Also, recall that $\Vor(D(p))$ consists of the points that are exterior to $K$ whose closest point on $\partial K$ lies within $D(p)$. Define the \emph{base} of $D(p)$, denoted $\Gamma(p)$, to be the intersection of $h(p)$ with the convex hull of $K \cup \{p_{\eps}\}$. 

For $\delta > 0$, recall that the \emph{$\delta$-restricted $\eps$-dual cap} defined by $p$, denoted $D_{\delta}(p)$, is $D(p) \cap B_{\delta}(p)$, where $B_{\delta}(p)$ is the Euclidean ball of radius $\delta$ centered at $p$ (see Figure~\ref{fig:voronoi}(b)). As before, $\Vor(D_{\delta}(p))$ is the set of points that are exterior to $K$ whose closest point on $\partial K$ lies within $D_{\delta}(p)$. Also, the \emph{$\delta$-restricted base}, denoted $\Gamma_{\delta}(p)$ is $\Gamma(p) \cap B_{\delta}(p)$.

Our objective in this section is to establish bounds on the product of the area of a $\sqrt{\eps}$-restricted $\eps$-dual cap and its Voronoi patch on the Dudley hypersphere. It will be easier to start with hyperplane patches on $H^{(y)}(p)$ and then generalize to spherical patches on $S$. The main result of this section is given in the following lemma. Part~(ii) is equivalent to Lemma~\ref{lem:dual-basic}, which is our main objective. Part~(i) is a useful intermediate result.

\begin{figure}[htbp]
  \centerline{\includegraphics[scale=0.38]{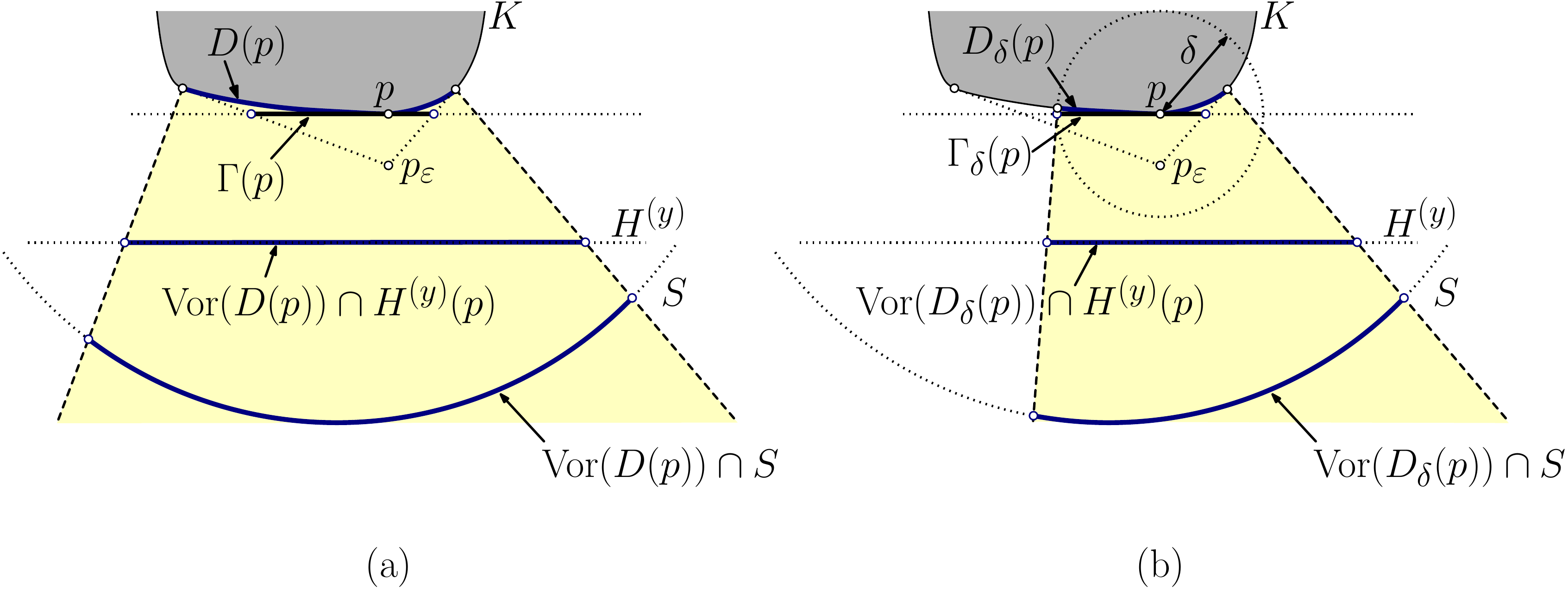}}
  \caption{Dual caps, bases, and Voronoi regions for the (a) unrestricted and (b) restricted cases.}
  \label{fig:voronoi}
\end{figure}

\begin{lemma} \label{lem:dual}
Let $K$ be a convex body in $\RE^d$, and let $0 < \eps \le 1/8$ and $\delta = \sqrt{\eps}$. There are constants $c_a$ and $c'_a$ (depending only on $d$) such that for any point $p \in \partial K$:
\begin{enumerate}
\item[$(i)$]  Given any $y \ge 1$, $\area(D_{\delta}(p)) \cdot \area(\Vor(D_{\delta}(p)) \cap H^{(y)}(p)) \ge c'_a \cdot \eps^{d-1}$.

\item[$(ii)$] If $K$ is fat and has diameter at least $2 \eps$, and $p$ lies within a unit ball centered at the origin, then $\area(D_{\delta}(p)) \cdot \area(\Vor(D_{\delta}(p)) \cap S) \ge c_a \cdot \eps^{d-1}$.
\end{enumerate}
\end{lemma}

This lemma holds generally for any $\delta \ge \sqrt{\eps}$, but it suffices for our purposes to consider the restricted case of $\delta = \sqrt{\eps}$. Note that the additional assumptions on fatness and diameter of part~(ii) are necessary for establishing a lower bound. If $K$ is not fat or not of sufficiently large diameter, then $\area(D_{\delta}(p))$ can be arbitrarily small. Since the Dudley hypersphere is bounded, it would not be possible to establish any lower bound on the product of their areas. 

The remainder of this section is devoted to proving this lemma. Because $p$ will be fixed throughout, in order to simplify the notation, we will drop references to $p$. For example, we will use $h$, $H^{(y)}$, $D_{\delta}$, $\Gamma_{\delta}$, and $B_{\delta}$ in place of $h(p)$, $H^{(y)}(p)$, $D_{\delta}(p)$, $\Gamma_{\delta}(p)$, and $B_{\delta}(p)$, respectively.

Since it will be useful to relate sets on $h$ with sets on $H^{(y)}$, we observe that each of these hyperplanes can be consistently identified with $\RE^{d-1}$ by endowing them with parallel coordinate frames, one centered at $p$ (for $h$) and one centered at $p$'s orthogonal projection onto $H^{(y)}$. Thus, a point on $h$ and its vertical projection onto $H^{(y)}$ have the same coordinates.

We start by proving Lemma~\ref{lem:dual}(i). Since the value of $y$ will be fixed throughout this part of the proof, we refer to $H^{(y)}$ simply as $H$. Let $p''$ denote the origin of $H$'s coordinate system (the vertical projection of $p$ onto $H$). (See Figure~\ref{fig:area-bound-setup}.) In order to exploit Lemma~\ref{lem:mahler} on the Mahler volume, rather than considering $\Vor(D_{\delta}) \cap H$ directly, we will find it convenient to instead analyze the polar dual of the base $\Gamma_{\delta}$. Using the aforementioned coordinate frame, we can think of $\Gamma_{\delta}$ as a body in $\RE^{d-1}$. For $r = \sqrt{\eps/8}$, consider the generalized polar of the dual base, $\polarX{r}{\Gamma_{\delta}}$, which we can think of as a convex subset of $H$. Because $\Gamma_{\delta}$ contains the origin of $h$ (namely, $p$), it follows directly that $\polarX{r}{\Gamma_{\delta}}$ is bounded, convex, and also contains the origin of $H$ (namely, $p''$). In order to obtain a lower bound on $\area(\Vor(D_{\delta}) \cap H)$, we will first show that $\polarX{r}{\Gamma_{\delta}}$ is a subset of $\Vor(D_{\delta}) \cap H$ and then derive a lower bound on $\area(\polarX{r}{\Gamma_{\delta}})$. The first assertion is established by the following lemma.

\begin{lemma} \label{lem:dual-subset}
Given the preconditions of Lemma~\ref{lem:dual} and $r = \sqrt{\eps/8}$, we have $\polarX{r}{\Gamma_{\delta}} \subseteq \Vor(D_{\delta}) \cap H$.
\end{lemma}

The proof is rather technical and involves a reduction to the problem in 2-dimensional space. Before giving the proof, it will help to provide some intuition regarding the relationship between $\Vor(D_{\delta}) \cap H$ and the polar of $\Gamma_{\delta}$. 

For the sake of simplicity, let us consider just the 2-dimensional setting. Let $t$ denote a point of tangency on $\partial K$ with respect to $p_{\eps}$ (see Figure~\ref{fig:dual-subset-intuition}), and let $v$ be the intersection of the line segment $p_{\eps} t$ with $h$. Shoot a ray from $t$ perpendicular to $\partial K$ until it intersects $H$. Let $q$ denote this intersection point. Since $K$ is convex, all the points on the segment $p'' q$ have their nearest neighbor on the portion of $\partial K$ between $p$ and $t$, that is, they all lie within $\Vor(D)$. Observe that if we translate this perpendicular line so that it emanates from $p_{\eps}$ instead of $t$, it will hit $H$ at a point $q'$ that is closer to $p''$. Therefore, the segment $p'' q'$ also lies within $\Vor(D)$. Let $\ell$ denote the distance between $p_{\eps}$ and $p''$. By similar triangles, it is easy to see that the length of $p'' q'$ is $\ell \cdot \eps/\|p v\|$. Since $v \in \Gamma$, $q'$ lies within $\polarX{r'}{\Gamma}$, where $r' = \sqrt{\ell \cdot \eps}$. Because $y \ge 1$ and $\eps \le 1/8$, we have $r' = \Omega(\sqrt{\eps})$. This observation generalizes readily to higher dimensions, and it follows that $\polarX{r'}{\Gamma} \subseteq \Vor(D) \cap H$. We will show how to generalize this intuition to higher dimensions and the $\delta$-restricted setting.

\begin{figure}[htbp]
  \centerline{\includegraphics[scale=0.40]{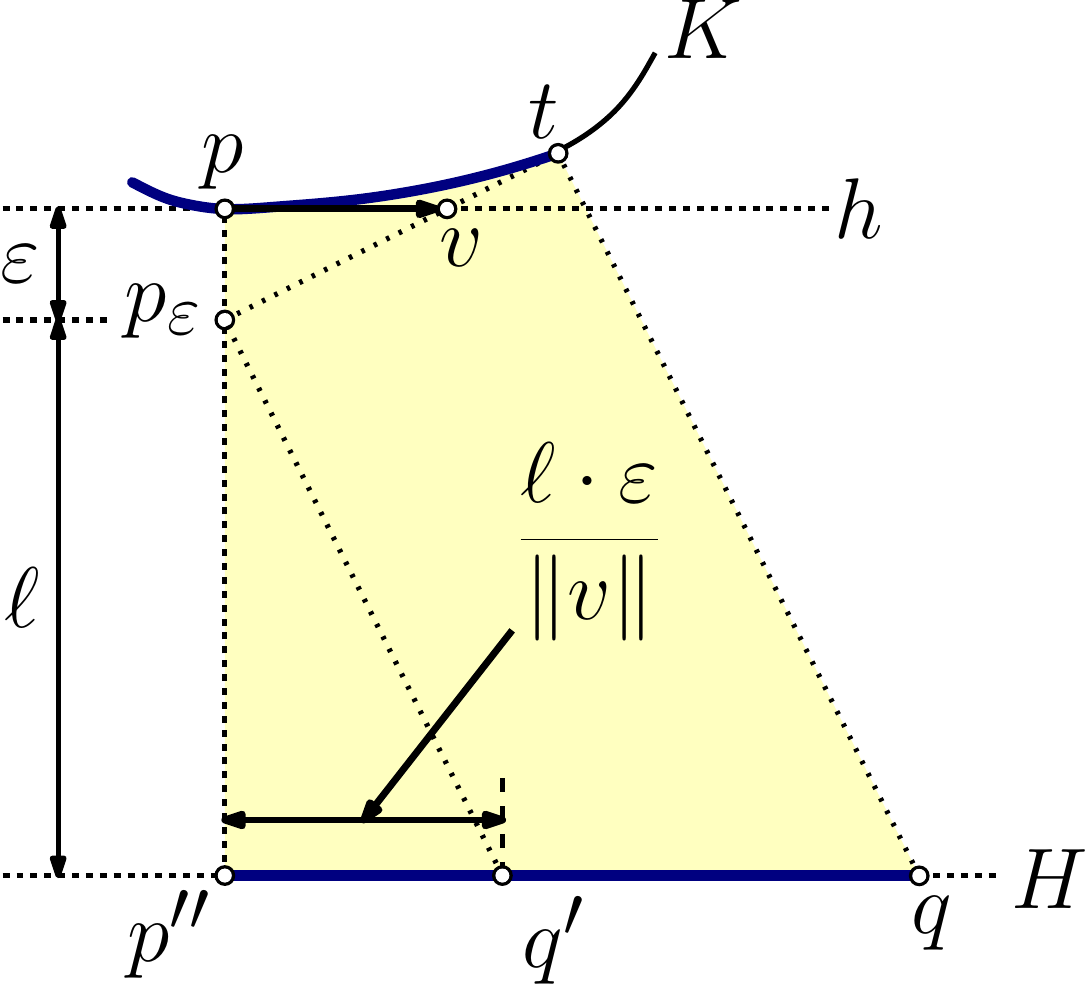}}
  \caption{The relationship between $\Vor(D_{\delta}) \cap H$ and the polar of $\Gamma_{\delta}$.}
  \label{fig:dual-subset-intuition}
\end{figure}

For any $z \in H$ let $w$ denote its nearest neighbor on $\partial K$. In order to prove Lemma~\ref{lem:dual-subset}, it suffices to show that if $w \notin D_{\delta}$ (implying that $z \notin \Vor(D_{\delta}) \cap H$), then $z \notin \polarX{r}{\Gamma_{\delta}}$. By our assumption that $H$ lies below $K$ it follows that $w$ lies on the ``lower surface'' of $\partial K$ (meaning that a vertical ray directed downwards from $w$ does not intersect the interior of $K$). Since $w$ is not in the restricted cap, we know that either $w \notin D$ or $w \notin B_{\delta}$.

It will simplify the analysis to reduce the problem to a 2-dimensional setting. Consider the plane $\Phi$ that contains the points $p$, $p_{\eps}$, and $w$. (Note that these points are not collinear.) Let $t$ be the point of tangency on $\partial K \cap \Phi$ with respect to $p_{\eps}$ that lies on the same side as $w$ (see Figure~\ref{fig:dual-subset}(a)). Let $v$ be the intersection of the line segment $p_{\eps} t$ with $h$. We may identify $\Phi$ with $\RE^2$ by imposing a coordinate system on $\Phi$ where the origin is at $p$, the $y$-axis is directed upwards away from $p_{\eps}$ and the $x$-axis is parallel to the vector from $p$ to $v$. Given a point $u \in \Phi$, let $u_x$ and $u_y$ denote its coordinates relative to this coordinate system. Further, if $u \in \partial K \cap \Phi$, let $u_{\theta}$ denote the slope of the (unique) supporting line on $\Phi$ passing through $u$. Note that $z$ need not lie on $\Phi$. Let $z'$ be the orthogonal projection of $z$ onto $\Phi$. Observe that $t_{\theta} = \eps/\|p v\|$, and therefore $\|p v\| = \eps/t_{\theta}$. By our choice of coordinate system and assumptions about orientations, the coordinates of $w$, $t$, and the slopes $w_{\theta}$ and $t_{\theta}$ are all nonnegative quantities. 

\begin{figure}[htbp]
  \centerline{\includegraphics[scale=0.40]{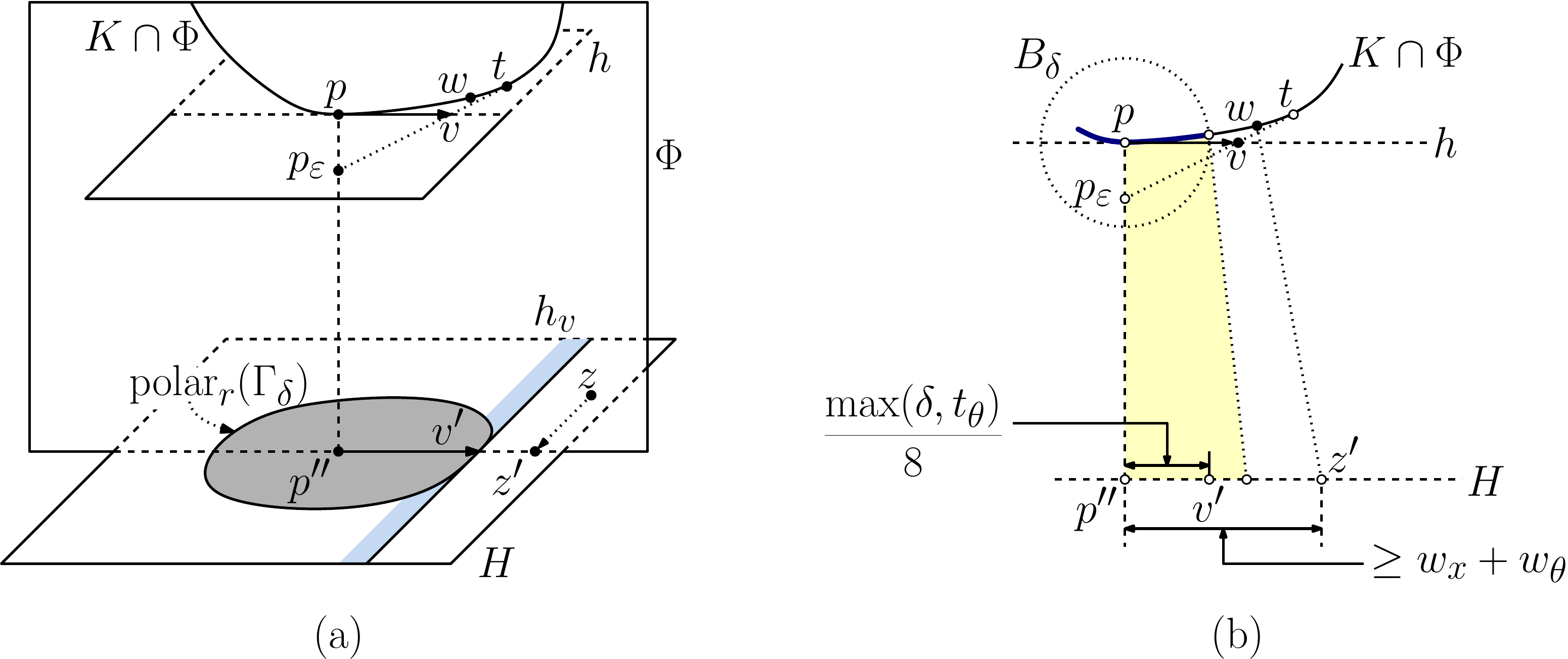}}
  \caption{The reduction to the plane $\Phi$.}
  \label{fig:dual-subset}
\end{figure}

The point $v$ lies on the base $\Gamma$ of $p$'s unrestricted dual cap. By employing our coordinate system on $h$, we can identify $v$ with a vector in $\RE^{d-1}$ (emanating from $p$). If $\|p v\| \le \delta$, then $v$ contributes a bounding halfspace to $\polarX{r}{\Gamma_{\delta}}$. This halfspace is bounded by a hyperplane that is orthogonal to $v$ and lies at distance $r^2/\|p v\|$ from the origin. Let us think of this halfspace, which we denote by $h_v$, as lying on $H$ (see Figure~\ref{fig:dual-subset}(a)). Recalling that $r = \sqrt{\eps/8}$, the distance of $h_v$'s bounding hyperplane to the origin $p''$ is $(\eps/8)/\|p v\| = t_{\theta}/8$. On the other hand, if $\|p v\| > \delta$, then $v$ lies outside the restricted base. In this case $v$'s subvector of length $\delta$ lies on the boundary of the restricted base and contributes to $\polarX{r}{\Gamma_{\delta}}$ a halfspace whose bounding hyperplane is at distance $(\eps/8)/\delta$ from the origin. Recalling that $\delta = \sqrt{\eps}$, this is equal to $\delta/8$. Thus, in either case, $\polarX{r}{\Gamma_{\delta}}$ is bounded by a halfspace whose defining hyperplane is orthogonal to $v$ and lies at distance $\max(\delta, t_{\theta})/8$ from the origin. This hyperplane intersects the horizontal line $H \cap \Phi$ at some point $v'$ that lies to the right of $p''$ at distance $v'_x = \max(\delta, t_{\theta})/8$ (see Figure~\ref{fig:dual-subset}(b)). 

Because the hyperplane passing through $v'$ is orthogonal to $v$, in order to show that $z \notin \polarX{r}{\Gamma_{\delta}}$, it suffices to show that $z'$ does not lie within $h_v$, which is equivalent to showing that $z'_x > v'_x$. We have thus reduced the problem to a two-dimensional setting.

Recall that $w$ is the closest point to $z$ on $\partial K$. We assert that $w$ is also the closest point to $z'$ on $\partial K \cap \Phi$. The reason is that the squared distance from $z$ to any point on $\partial K \cap \Phi$ can be expressed as the sum of the squared distance from $z'$ to this point and the squared distance from $z$ to $z'$. Since the latter quantity is the same for all points on $\Phi$, the closest point to $z$ is also the closest point to $z'$. From basic properties of convexity, it follows that the line $w z'$ is orthogonal to the support line passing through $w$ on $\partial K \cap \Phi$. Therefore, the slope of $w z'$ (in $\Phi$'s coordinate system) is $-1/w_{\theta}$, and in particular we have $(z'_x - w_x)/(z'_y - w_y) = -w_{\theta}$. Since $h$ and $H$ are separated by at least unit distance (with $h$ above $H$), we have $w_y - z'_y \ge 1$, and so $z'_x \ge w_x + w_{\theta}$. 

Thus, to complete the proof of Lemma~\ref{lem:dual-subset}, it suffices to show that if $w \notin D_{\delta}$ then $v'_x < w_x + w_{\theta}$. We first establish two useful technical results. These results will be applied in a context where $w$ lies within the unrestricted dual cap but outside the restricted dual cap. That is, when $w_x \le t_x$ but $w \notin B_{\delta}$. The first result shows that if $t_{\theta}$ is sufficiently small, the slope of the line $p w$ is at most unity. The second shows that if $t_{\theta}$ is sufficiently large, the slope of $p w$ is not much smaller than the slope of $t$'s supporting line.

\begin{restatable}{lemma}{DualSubsetHelperStmt}\label{lem:dual-subset-helper}
Given the preconditions of Lemma~\ref{lem:dual} and the aforementioned 2-dimensional reduction, and given $w$ and $t$ as introduced above, where $w_x \le t_x$ and $w \notin B_{\delta}$:
\begin{enumerate}
\item[$(i)$] if $t_{\theta} \le \delta\sqrt{8}$, then $w_y/w_x \le 1$, and

\item[$(ii)$] if $t_{\theta} > \delta\sqrt{8}$, then $w_y/w_x \ge t_{\theta}/2$.
\end{enumerate}
\end{restatable}

The proof involves simple geometry and is given in the appendix.

\medskip

We are now in position to complete the proof of Lemma~\ref{lem:dual-subset}. Recall that our objective is to show that if $w \notin D_{\delta}$ then $v'_x < w_x + w_{\theta}$, where $v'_x = \max(\delta, t_{\theta})/8$. We consider two cases, depending on $t_{\theta}$. First, if $t_{\theta} \le \delta\sqrt{8}$, then $v'_x \le \max(\delta, \delta\sqrt{8})/8 = \delta/\sqrt{8}$. Since the line $p_{\eps} t$ has slope $t_{\theta}$ and $t_y \ge 0$, we have $t_x = (t_y + \eps)/t_{\theta} \ge \eps/t_{\theta} \ge \delta/\sqrt{8}$. We consider two subcases. If $w_x > t_x$, then we have $w_x + w_{\theta} > t_x \ge \delta/\sqrt{8} \ge v'_x$, as desired. On the other hand, if $w_x \le t_x$, then $w$ is inside the unrestricted cap $D$. Since by our hypothesis, $w$ is not in the restricted cap, it must be that $w \notin B_{\delta}$, that is, $w_x^2 + w_y^2 > \delta^2$. By Lemma~\ref{lem:dual-subset-helper}(i), we have $w_x \ge w_y$. Therefore, $2 w_x^2 \ge w_x^2 + w_y^2 > \delta^2$, which implies that $w_x > \delta/\sqrt{2}$. Therefore, $w_x + w_{\theta} \ge w_x > \delta/\sqrt{2} > v'_x$, as desired.

For the second case, assume that $t_{\theta} > \delta\sqrt{8}$. In this case $v'_x = t_{\theta}/8$. As before, we consider two subcases. If $w_x > t_x$, then by convexity $w_{\theta} \ge t_{\theta}$, and so $w_x + w_{\theta} \ge t_{\theta} > v'_x$, as desired. On the other hand, if $w_x \le t_x$, then since $w$ lies within the unrestricted cap, we may infer that $w \notin B_{\delta}$. By Lemma~\ref{lem:dual-subset-helper}(ii), we have $w_y/w_x \ge t_{\theta}/2$. Because the support line at $w$ passes below the origin, we also have $w_{\theta} \ge w_y/w_x$. Therefore $w_x + w_{\theta} \ge w_y/w_x \ge t_{\theta}/2 > v'_x$. This completes the proof of Lemma~\ref{lem:dual-subset}.

\medskip

Because it is easier to deal with flat objects than curved ones, before returning to the proof of Lemma~\ref{lem:dual}(i), we show that the area of the restricted dual cap is, up to a constant factor, bounded below by the area of its base. This result is straightforward for unrestricted caps, since it is easy to show that the base is contained within the orthogonal projection of the dual cap onto $h$. However, restriction complicates the analysis. The proof involves a technical geometric argument and is presented in the appendix.

\begin{restatable}{lemma}{BaseCapAreaStmt}\label{lem:base-cap-area}
Given the preconditions of Lemma~\ref{lem:dual}, it follows that $\area(D_{\delta}) \ge \area(\Gamma_{\delta}) / 2^{d-1}$.
\end{restatable}

\medskip

We are now ready to prove Lemma~\ref{lem:dual}(i). Recall that $r = \sqrt{\eps/8}$. As observed earlier, $\polarX{r}{\Gamma_{\delta}}$ is a scaled copy of $\polar{\Gamma_{\delta}}$ by a factor of $r^2$, and therefore (since these are $(d-1)$-dimensional bodies) we have $\area(\polarX{r}{\Gamma_{\delta}}) = r^{2(d-1)} \cdot \area(\polar{\Gamma_{\delta}})$. By applying Lemma~\ref{lem:dual-subset}, we have
\[
	\area(\Vor(D_{\delta}) \cap H)
		~ \ge ~ \area(\polarX{r}{\Gamma_{\delta}})
		~  =  ~ r^{2(d-1)} \cdot \area(\polar{\Gamma_{\delta}}).
	\label{eqn:dual-subset-bound}
\]
By Lemma~\ref{lem:base-cap-area}, $\area(D_{\delta}) \ge \area(\Gamma_{\delta}) / 2^{d-1}$, and therefore
\begin{eqnarray*}
	\area(D_{\delta}) \cdot \area(\Vor(D_{\delta}) \cap H)
		& \ge & \frac{\area(\Gamma_{\delta})}{2^{d-1}} \cdot r^{2(d-1)} \cdot \area(\polar{\Gamma_{\delta}}) \\
		& \ge & \left( \frac{r^2}{2} \right)^{\N d-1} \area(\Gamma_{\delta}) \cdot \area(\polar{\Gamma_{\delta}}).
\end{eqnarray*}
We now apply the Mahler-volume bound. By Lemma~\ref{lem:mahler} (in $\RE^{d-1}$), there exists a constant $c_m$ (depending only on $d$) such that $\area(\Gamma_{\delta}) \cdot \area(\polar{\Gamma_{\delta}}) \ge c_m$. Therefore,
\[
	\area(D_{\delta}) \cdot \area(\Vor(D_{\delta}) \cap H)
		~ \ge ~ c_m \left( \frac{r^2}{2} \right)^{d-1} 
		~  =  ~ c_m \left( \frac{\eps}{16} \right)^{d-1}.
\]
Selecting any $c'_a \le c_m/16^{d-1}$ establishes Lemma~\ref{lem:dual}(i).

\bigskip

Next, let us establish Lemma~\ref{lem:dual}(ii). Recall that we assume that $K$ is fat and of diameter at least $2 \eps$. In particular, let us assume that $K$ is $\gamma$-fat, where $\gamma$ is a constant independent of $n$ and $\eps$ that lies in the interval $(0,1]$. (As a result of Lemma~\ref{lem:fat}, we may assume that $\gamma$ is $1/d$ when applying this result.)

It is natural to try to generalize the approach used in part~(i). First, we would show that 
\[
	\area(\Vor(D_{\delta}) \cap S) = \Omega\big(r^{2(d-1)} \cdot \area(\polar{\Gamma_{\delta}})\big)
	\quad\hbox{and}\quad
	\area(D_{\delta}) = \Omega(\area(\Gamma_{\delta})),
\]
and then we would apply the Mahler-volume bound to yield a lower bound on the product $\area(\Gamma_{\delta}) \cdot \area(\polar{\Gamma_{\delta}})$. A problem arises, however, if $K$ is not smooth. In particular, if some portion of the boundary of $K$ in $p$'s vicinity is nearly vertical, then the boundary of $\Gamma_{\delta}$ can be arbitrarily close to the origin (namely $p$), implying that $\polar{\Gamma_{\delta}}$ cannot be bounded, and hence its area can be arbitrarily large. This was not an issue in part~(i), because $H$ is also unbounded. But since $S$ is bounded, $\area(\Vor(D_{\delta}) \cap S)$ cannot be arbitrarily large. We will remedy this by smoothing $K$ by taking its Minkowski sum with a small Euclidean ball of radius $O(\eps)$. We shall see (in the proof of Lemma~\ref{lem:smooth-to-polar}) that this allows us to constrain the area of $\polar{\Gamma_{\delta}}$. This smoothing operation requires us to adapt many of the prior results of this section to this new context.

To construct the smoothed body, for the remainder of this section define $\eps' = \eps/2$, and let $K' = K \oplus \eps'$ (see Figure~\ref{fig:smooth}(a)). Recall that $h$ denotes the supporting hyperplane at $p$ and $p_{\eps}$ is the point at distance $\eps$ from $p$ in the direction orthogonal to $h$. As before, for the sake of illustration, let us assume that $p_{\eps}$ is vertically below $p$. Let $p'$ be the midpoint of the segment $p p_{\eps}$. Clearly, $p' \in \partial K'$, and the parallel hyperplane $h'$ passing through $p'$ is a supporting hyperplane for $K'$.

\begin{figure}[htbp]
  \centerline{\includegraphics[scale=0.40]{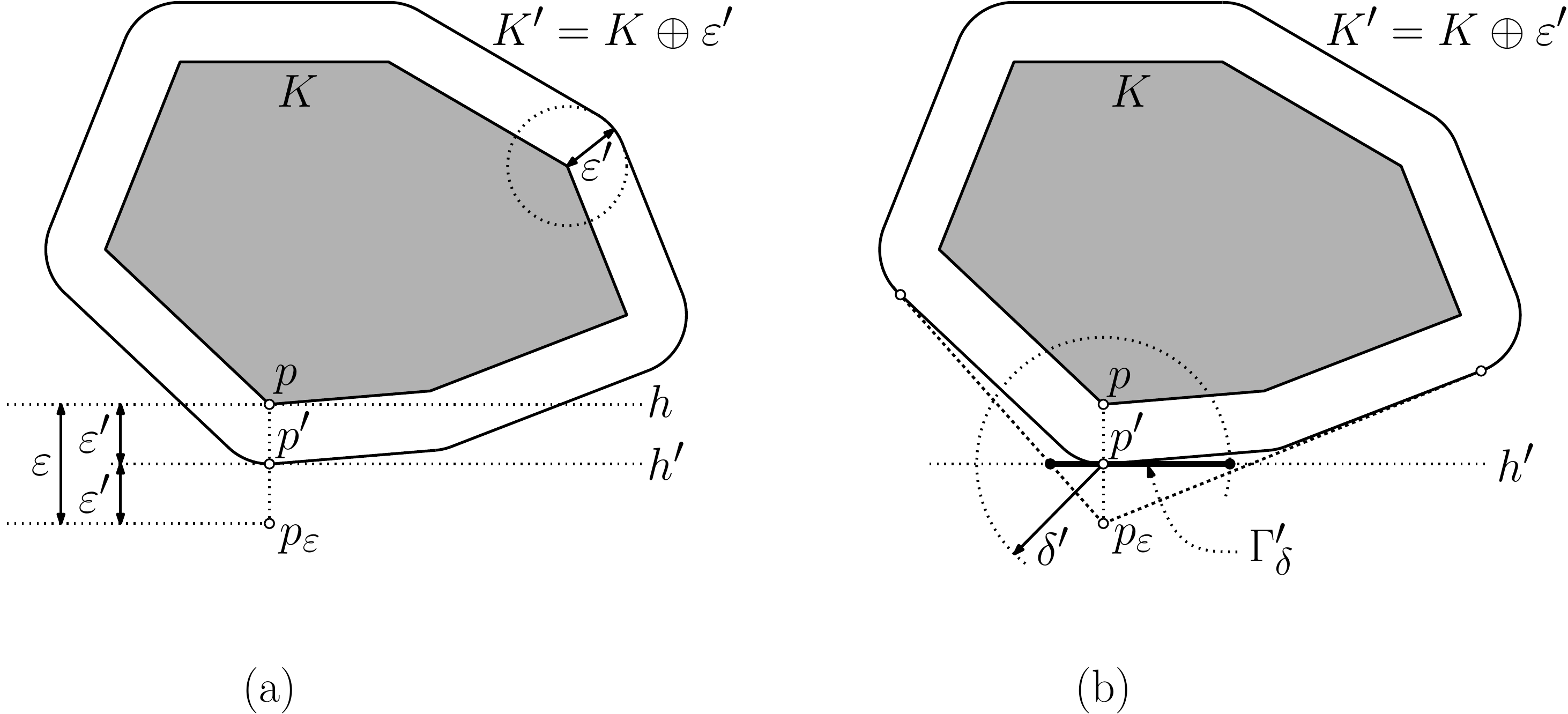}}
  \caption{The smoothed body $K'$.}
  \label{fig:smooth}
\end{figure}

Let us also define the dual base in this smoothed context. Define $\Gamma'$ to be the intersection of $h'$ and $\conv(K' \cup \{p_{\eps}\})$. Let $\delta' = \sqrt{\eps'} = \delta/\sqrt{2}$, and define the restricted base $\Gamma'_{\delta}$ to be the intersection of $\Gamma'$ and a ball of radius $\delta'$ centered at $p'$ (see Figure~\ref{fig:smooth}(b)). Our analysis will be based on $K'$ and $\Gamma'_{\delta}$, as opposed to $K$ and $\Gamma_{\delta}$. Our first objective will be to show that the area of $\Gamma'_{\delta}$ is not significantly larger than that of $\Gamma_{\delta}$. As before, we endow $h$ and $h'$ with parallel coordinate frames whose origins are located at $p$ and $p'$, respectively. Then we can think of $\Gamma_{\delta}$ and $\Gamma'_{\delta}$ as convex sets in $\RE^{d-1}$. The following lemma relates these two bodies.

\begin{restatable}{lemma}{SmoothAreaStmt}\label{lem:smooth-area}
Given a convex body $K$ that is $\gamma$-fat and of diameter at least $2 \eps$ and given $\Gamma_{\delta}$ and $\Gamma'_{\delta}$ as defined above, there exists a constant $c$ (depending on $\gamma$ and the dimension $d$) such that $\area(\Gamma'_{\delta}) \le c \cdot \area(\Gamma_{\delta})$.
\end{restatable}

The proof is rather technical, but it involves simple geometric reasoning. It is given in the appendix.

Recall that $\Vor(D_{\delta}) \cap S$ consists of the set of points on the sphere $S$ whose closest point on $\partial K$ lies within the restricted dual cap $D_{\delta}$. Let $D'_{\delta}$ denote the corresponding restricted dual cap for $K'$, that is, the set of points of $\partial K'$ that are visible from $p_{\eps}$ and lie within the ball $B_{\delta'}(p')$. Our analysis will be based on establishing a lower bound on the area of $\Vor(D'_{\delta}) \cap S$. The following lemma shows that this will provide a lower bound on the area of $\Vor(D_{\delta}) \cap S$.

\begin{lemma} \label{lem:smooth-voronoi}
Given the preconditions of Lemma~\ref{lem:dual}(ii), $\area(\Vor(D'_{\delta}) \cap S) \le \area(\Vor(D_{\delta}) \cap S)$.
\end{lemma}

\begin{proof}
We will prove the stronger result that $\Vor(D'_{\delta}) \cap S \subseteq \Vor(D_{\delta}) \cap S$. We begin by observing that both $D_{\delta}$ and $D'_{\delta}$ lie within the ball bounded by the Dudley hypersphere $S$. To see this, recall that by the conditions of Lemma~\ref{lem:dual}(ii), $p$ lies within unit distance of the origin, and so by the triangle inequality $p'$ lies within distance $1 + \eps'$ of the origin. The points of $D_{\delta}$ and $D'_{\delta}$ lie within distances $\delta$ and $\delta'$ of $p$ and $p'$, respectively. Therefore, the distance of any point of $D_{\delta}$ or $D'_{\delta}$ from the origin is at most $1 + \max(\delta, \eps' + \delta')$. As shown in the proof of Lemma~\ref{lem:smooth-area}, $\delta' + \eps' \le \delta$, and therefore this distance is at most $1 + \delta \le 1 + 1/\sqrt{8} < 3$. Therefore, both caps lie within $S$. 

Consider any point $q'' \in \Vor(D'_{\delta}) \cap S$. It suffices to show that $q'' \in \Vor(D_{\delta}) \cap S$. Let $q'$ be the closest point to $q''$ on $\partial K'$. By convexity, $q'$ is the closest point of $q''$ on $\partial K'$ if and only if the supporting hyperplane at $q'$, denoted $h(q')$, is orthogonal to the line $q' q''$.  Let $q$ be the closest point to $q'$ on $\partial K$. By basic properties of Minkowski sums, the segment $q q'$ is orthogonal to both the supporting hyperplanes $h(q)$ and $h(q')$ at $q$ and $q'$, respectively (see Figure~\ref{fig:smooth-voronoi}). It follows that all three points $q$, $q'$ and $q''$ are collinear, and $h(q)$ is orthogonal to the segment $q q''$. Therefore, $q$ is the closest point to $q''$ on $\partial K$.

\begin{figure}[htbp]
  \centerline{\includegraphics[scale=0.40]{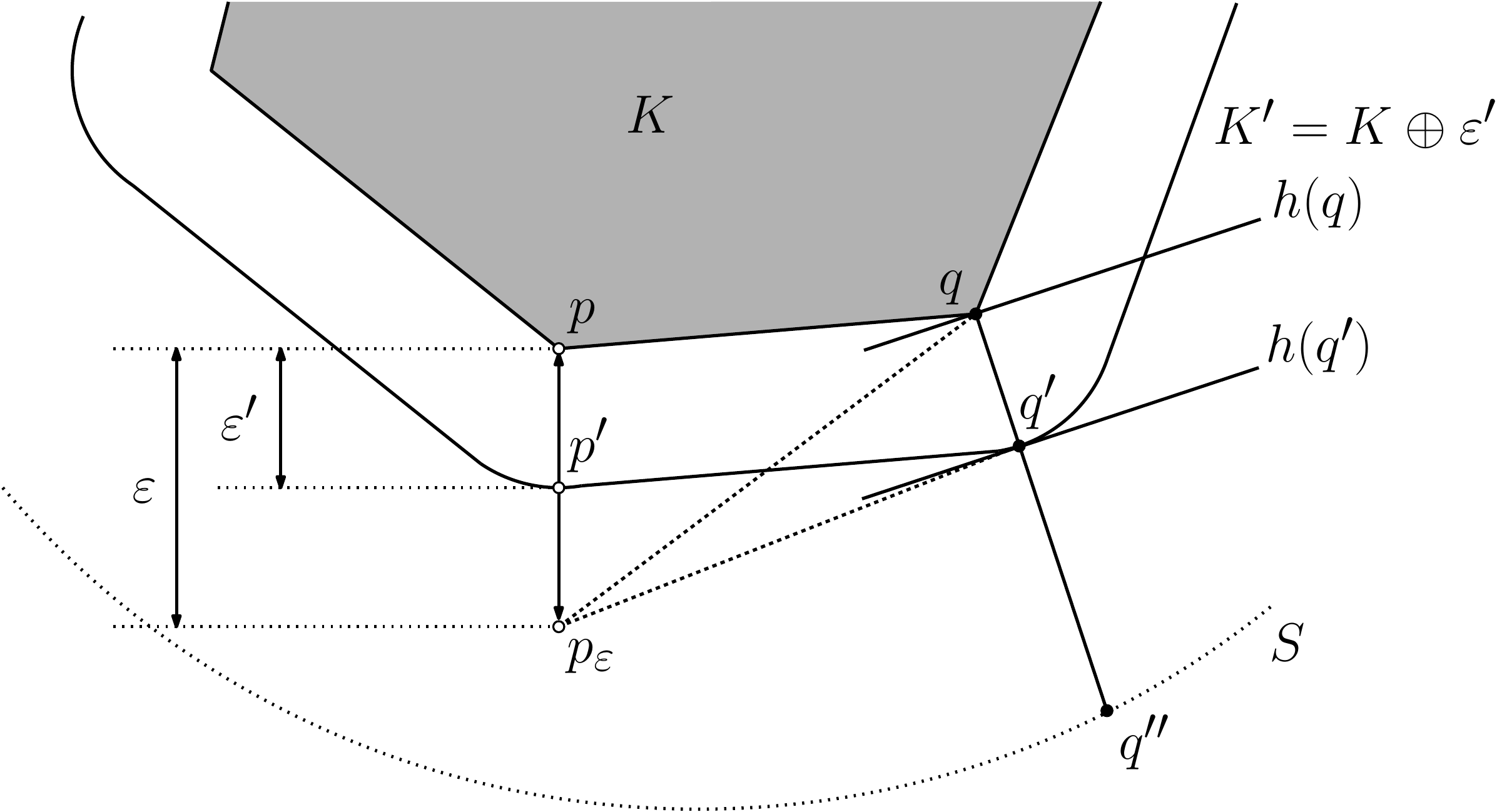}}
  \caption{Proof of Lemma~\ref{lem:smooth-voronoi}.}
  \label{fig:smooth-voronoi}
\end{figure}

Since $q'' \in \Vor(D'_{\delta}) \cap S$, we have $q' \in D'_{\delta}$, which means that $q'$ is visible from $p_{\eps}$ and $q'$ lies within distance $\delta'$ of $p'$. Because $h(q)$ and $h(q')$ are parallel, it follows that $q$ is also visible from $p_{\eps}$. Therefore, $q$ lies in $p$'s unrestricted dual cap $D$. To prove that $q$ lies within the restricted cap $D_{\delta}$, it suffices to show that $\|p q\| \le \delta$. We apply a standard result from convexity theory which states that for any convex body $K$ and two points $p'$ and $q'$ that are exterior to $K$, if $p$ and $q$ are their respective nearest neighbors on $\partial K$, then $\|p q\| \le \|p' q'\|$ (see, e.g., Lemma~{4.3} in \cite{Dudley}). Clearly, this applies in our situation, and so $\|p q\| \le \|p' q'\| \le \delta' < \delta$, which implies that $q \in D_{\delta}$. In summary, the closest point to $q''$ on $\partial K$ lies within $D_{\delta}$, implying that $q'' \in \Vor(D_{\delta}) \cap S$, as desired.
\end{proof}

Before completing the proof of Lemma~\ref{lem:dual}(ii), we exploit the smoothness of $K'$ to establish a relationship between the areas of $\Vor(D'_{\delta}) \cap S$ and $\polarX{r'}{\Gamma'_{\delta}}$, where the polar radius $r'$ is suitably modified for the smoothed context. This is given in the next lemma. The proof involves basic geometric reasoning and is given in the appendix.

\begin{restatable}{lemma}{SmoothToPolarStmt}\label{lem:smooth-to-polar}
Given the preconditions of Lemma~\ref{lem:dual}(ii) and $r' = \sqrt{\eps'/8}$, we have $\area(\Vor(D'_{\delta}) \cap S) \ge \area(\polarX{r'}{\Gamma'_{\delta}})$.
\end{restatable}

\bigskip

We are now ready to prove Lemma~\ref{lem:dual}(ii). Recall that $r' = \sqrt{\eps'/8}$. By Lemmas~\ref{lem:smooth-voronoi} and~\ref{lem:smooth-to-polar}, we have
\[
	\area(\Vor(D_{\delta}) \cap S)
		~ \ge ~ \area(\Vor(D'_{\delta}) \cap S)
		~ \ge ~ \area(\polarX{r'}{\Gamma'_{\delta}}).
\]
As observed earlier, $\polarX{r'}{\Gamma'_{\delta}}$ is a scaled copy of $\polar{\Gamma'_{\delta}}$ by a factor of $(r')^2 = \eps'/8 = \eps/16$, and therefore (since these are $(d-1)$-dimensional bodies) we have
\[
	\area(\Vor(D_{\delta}) \cap S)
		~ \ge ~ \left(\frac{\eps}{16}\right)^{\N d-1} \cdot \area(\polar{\Gamma'_{\delta}}).
\]
By Lemma~\ref{lem:base-cap-area}, $\area(D_{\delta}) \ge \area(\Gamma_{\delta}) / 2^{d-1}$. Also, by Lemma~\ref{lem:smooth-area} there is a constant $c''$ (depending on the fatness parameter $\gamma$ and $d$) such that $\area(\Gamma'_{\delta}) \le c'' \cdot \area(\Gamma_{\delta})$. Therefore, we have
\[
	\area(D_{\delta})
		~ \ge ~ \frac{\area(\Gamma_{\delta})}{2^{d-1}}
		~ \ge ~ \frac{\area(\Gamma'_{\delta})}{c'' \cdot 2^{d-1}}.
\]
Combining these, we obtain
\begin{eqnarray*}
	\area(D_{\delta}) \cdot \area(\Vor(D_{\delta}) \cap S)
		& \ge & \frac{\area(\Gamma'_{\delta})}{c'' \cdot 2^{d-1}} \cdot \left(\frac{\eps}{16}\right)^{\N d-1} \cdot \area(\polar{\Gamma'_{\delta}}) \\
		&  =  & \inv{c''} \left( \frac{\eps}{32} \right)^{\N d-1} \area(\Gamma'_{\delta}) \cdot \area(\polar{\Gamma'_{\delta}}).
\end{eqnarray*}
By applying Lemma~\ref{lem:mahler} (in $\RE^{d-1}$) to $\Gamma'_{\delta}$, there exists a constant $c_m$ (depending on $d$) such that
\[
	\area(\Gamma'_{\delta}) \cdot \area(\polar{\Gamma'_{\delta}}) 
		~ \ge ~ c_m.
\]
Therefore, we have
\[
	\area(D_{\delta}) \cdot \area(\Vor(D_{\delta}) \cap S)
		~ \ge ~ \frac{c_m}{c''} \left( \frac{\eps}{32} \right)^{\N d-1}.
\]
Selecting any $c_a \le (c_m/c'') (1/32)^{d-1}$ establishes Lemma~\ref{lem:dual}(ii). This concludes our proof of the area bounds. 

\section{Concluding Remarks}

In this paper we have presented an efficient data structure for determining approximately whether a given query point lies within a convex body. Our solution is based on a simple and natural quadtree-based algorithm, called {\alg}. Our principal technical contribution has been an analysis of the space-time trade-offs for this algorithm. These are the first nontrivial space-time trade-offs for this problem. We do not know whether this analysis is tight, but we presented a lower bound example that demonstrates the limits of possible improvements. We also demonstrated the value of approximate polytope membership by showing that our data structure can be combined with an AVD data structure to produce significant improvements to the space-time trade-offs of approximate nearest neighbor searching in Euclidean space.

Our analysis of the trade-offs involved a combination of a number of novel techniques, which may be of broader interest. One notable example is the application of the Mahler volume as a means of analyzing the local structure of a convex body through consideration of both its primal and dual representations. This resulted in an efficient two-pronged sampling strategy for computing hitting sets of low cardinality for $\eps$-dual caps. The Mahler volume has also been applied in \cite{AFM12} to derive an optimal area-sensitive bound on the number of facets needed to approximate a convex body.

This work provokes a number of questions for further research. The first involves extending approximate polytope membership queries to other approximate query problems involving convex bodies. For example, in Section~\ref{sec:ann} we showed how to reduce approximate nearest neighbor searching in dimension $d$ to vertical ray shooting queries in dimension $d+1$. However, the polytope involved had a very restricted structure. It would be interesting to know whether there is a data structure exhibiting similar trade-offs for answering approximate ray-shooting queries for general convex bodies. Another example is answering approximate linear-programming queries, where a convex body is preprocessed, and the problem is to determine an  extreme point of the body approximately in a given query direction. A further generalization of this would be to extend the work of Barba and Langerman \cite{BaL15} to an approximate setting. It particular, is it possible to preprocess convex bodies so that given two such bodies that have been translated and rotated, it can be decided efficiently whether they intersect each other approximately.

Our result on approximate nearest neighbor searching relies on the lifting transformation to reduce the problem to approximate polytope membership. As a consequence, this approach is applicable only to Euclidean distances. This raises the question of whether there exists a more direct route to approximate nearest neighbor searching that achieves similar space-time improvements and yet avoids reliance on lifting. For example, Arya and Chan \cite{ArC14} have presented improvements to approximate nearest-neighbor searching that do not involve lifting. This raises the hope that generalizations to other norms may be possible. While their focus was different from ours (for example, space-time trade-offs are not considered), their results are inferior to our best bounds. These better bounds arise explicitly from concepts like the Mahler volume, which are only applicable in the context of convex approximation, and hence they rely crucially on lifting. A major challenge is whether it is possible to bypass this intermediate step in order to obtain analogous improvements for approximate nearest neighbor searching.

\section{Acknowledgments}

The authors would like to thank the anonymous reviewers for the journal version of the paper for their many insightful comments.

\subsection*{Note Added in Proof}

After the original submission of this paper, the authors have discovered a new approach to polytope membership that achieves query time $O(\log \frac{1}{\eps})$ with storage of only $O(1/\eps^{(d-1)/2})$~\cite{AFM17}. As a consequence, it is possible to answer $\eps$-approximate nearest neighbor queries for a set of $n$ points in $O(\log \frac{n}{\eps})$ time with storage of only $O(n/\eps^{d/2})$. While these new results surpass the results of this paper theoretically, the data structure presented there involves significantly larger constant factors and lacks the simplicity and practicality of the approach described here.


\bibliographystyle{abbrv}
\bibliography{shortcuts,polytope-journal}


\appendix

\section{Omitted Proofs}

{\VCLemmaStmt*}

\begin{proof}
We may assume that $K$ contains the origin, since this clearly does not affect the VC-dimension of these range spaces. To prove~(1), consider the set of augmented points $(p,h) \in \partial K$ (recalling that $h$ is any supporting hyperplane passing through $p$). Under the polar transformation (recall Section~\ref{sec:prelim-polar}), these supporting hyperplanes are mapped to a set of points that form the boundary of $\polar{K}$, which is a convex body. Consider any point $q$ that is external to $K$ (see Figure~\ref{fig:cap-duality}(a)). If we treat $q$ as a vector, $\polar{q}$ is a hyperplane that intersects $\polar{K}$ (see Figure~\ref{fig:cap-duality}(b)). 

\begin{figure}[htbp]
  \centerline{\includegraphics[scale=0.40]{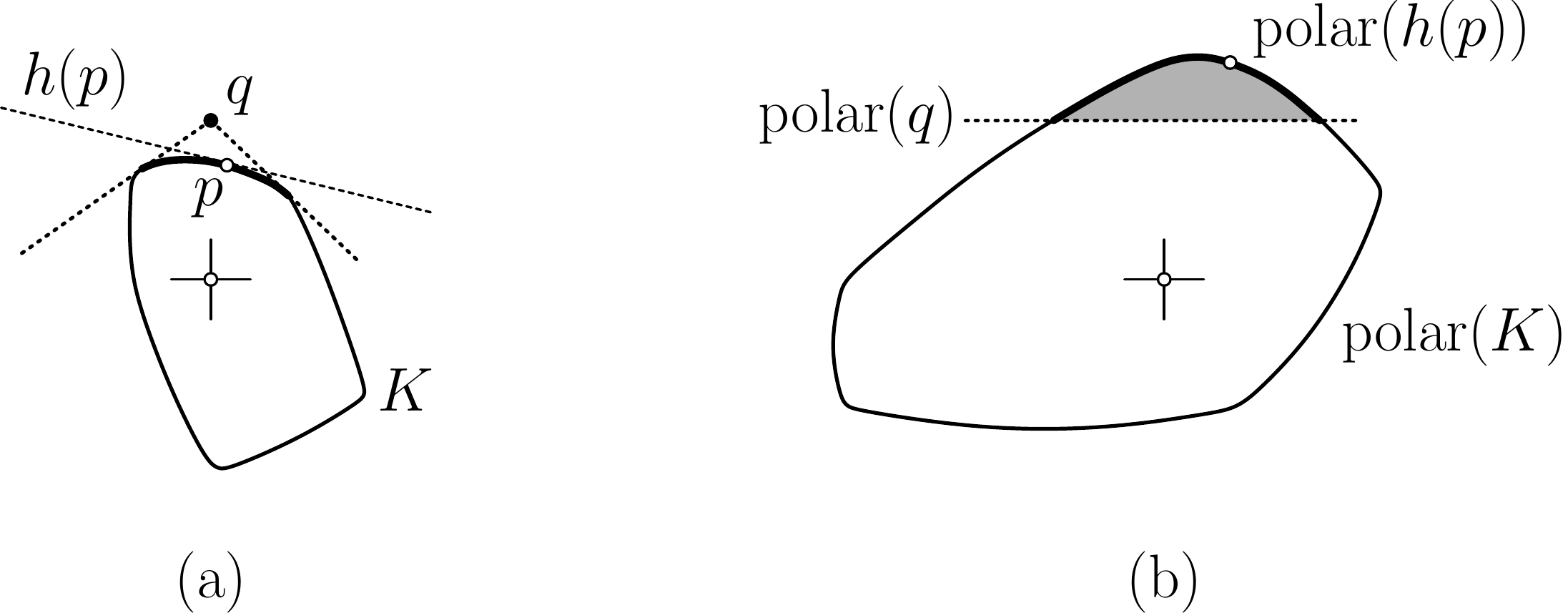}}
  \caption{Proof of Lemma~\ref{lem:VC}.}
  \label{fig:cap-duality}
\end{figure}

Consider the \emph{cap} of $\polar{K}$ induced by $\polar{q}$, which we define to be the points of the boundary of $\polar{K}$ that are separated from the origin by $\polar{q}$. By the inclusion-reversing properties of the polar, the points of this cap are in 1--1 correspondence with the supporting hyperplanes of $K$ that separate $q$ from $K$. It follows that the set of dual caps of $K$ is equivalent, through polarity, to the set of caps of $\polar{K}$. The range space of caps is equivalent to the range space of halfspaces, which is known to have constant VC-dimension, and therefore the VC-dimension of dual caps is equal.

To establish~(2), consider the function that maps each point $q \in S$ to its closest augmented point $(p,h) \in \partial K$, where $h$ is chosen to be orthogonal to the line segment $pq$. This function is a bijection (in fact, a homeomorphism) from the points of $S$ to the augmented points of $\partial K$. This function induces a 1--1 correspondence between the set of $\eps$-dual caps of $K$ (in fact, any set of surface patches) and the Voronoi patches associated with these dual caps. Therefore, $(X_2,R_2)$ and $(X_1,R_1)$ have the same VC-dimension.
 
To establish~(3) and~(4), observe that each range from $(X_3,R_3)$ (resp., $(X_4,R_4)$) results from intersecting a range of $(X_1,R_1)$ (resp., $(X_2,R_2)$) and a Euclidean ball. It is well known (see, e.g.,~\cite{Mat02}) that a range space that results by taking the intersection of ranges from two range spaces of constant VC-dimension has itself constant VC-dimension.
\end{proof}

{\BallLemmaStmt*}

\begin{proof}
Consider a $d$-dimensional ball $B$ of radius $\Delta/4$, and let $P$ be any polytope that is an outer $\eps$-approximation of $B$, that is, $B \subseteq P \subseteq B \oplus \eps$. Since $\eps \le \Delta/4$, $P$ is of diameter at most $\diam(B) + 2 \eps \le \Delta$. We will show that $P$ satisfies the conditions of the lemma.

Since the Hausdorff distance is a metric, it follows from the triangle inequality that any polytope $P'$ that is an $\eps$-approximation of $P$ is a $2 \eps$-approximation of $B$. To complete the proof, it suffices to show that any $2 \eps$-approximation of $B$ has at least the desired number of facets. To do this, define $\delta = 4\sqrt{\Delta \eps}$, and let $\Sigma$ denote a set of points on $\partial B$ such that the distance between any two points of $\Sigma$ is at least $\delta$. By a simple packing argument, for a suitable constant $c_b$ there exists such a set of size at least $c_b \big(\Delta / \sqrt{\Delta \eps} \big)^{d-1} = c_b (\Delta / \eps)^{(d-1)/2}$. It is easy to see that the function that maps each point of $\partial P'$ to its closest point on $\partial B$ induces a 1--1 correspondence between these two sets. Let $\Sigma'$ be the points corresponding to $\Sigma$ on $\partial P'$.

\begin{figure}[htbp]
	\centerline{\includegraphics[scale=0.40]{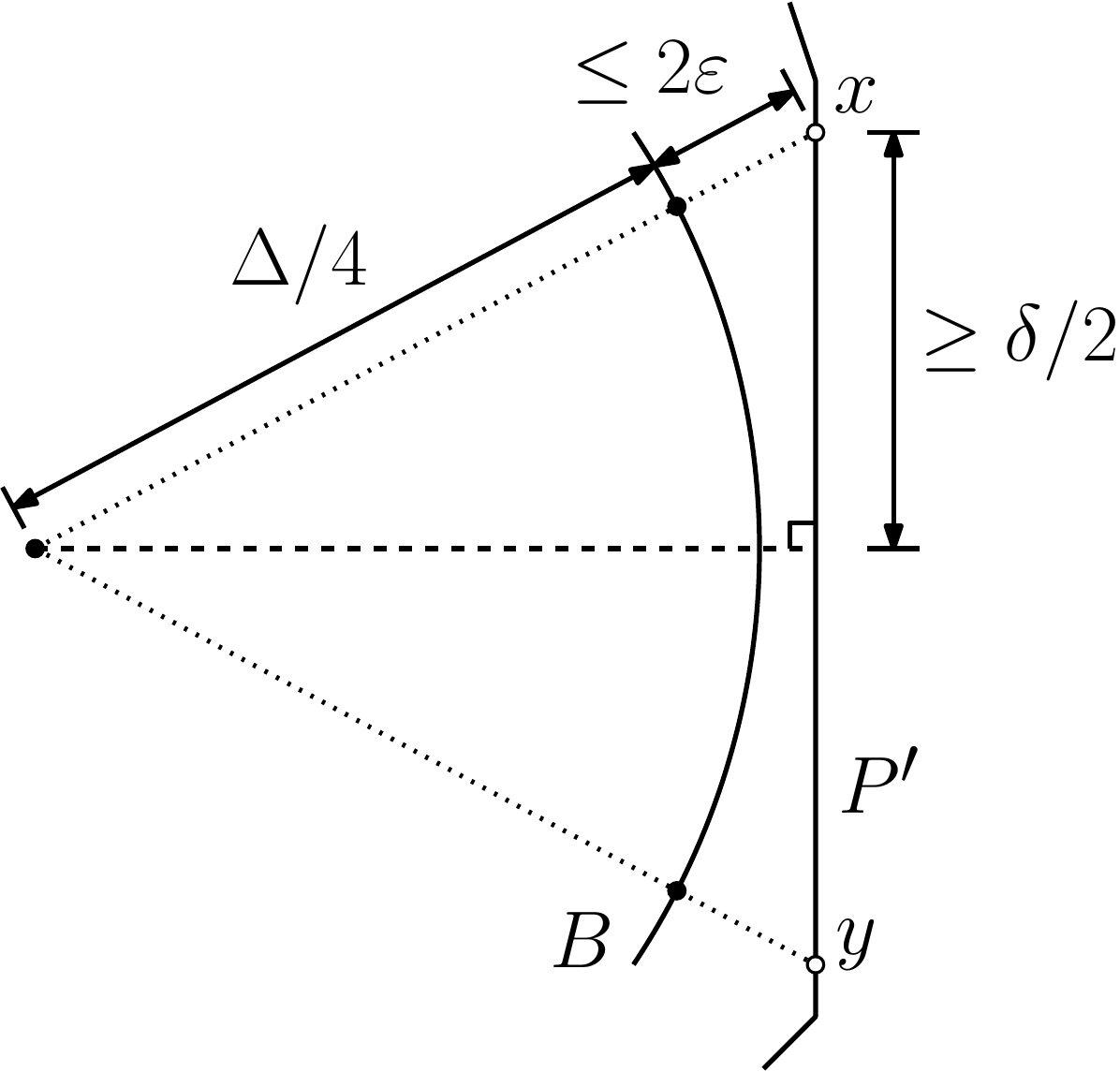}}
	\caption{Proof of Lemma~\ref{lem:ball}.}
	\label{fig:pythag}
\end{figure}

We assert that the points of $\Sigma'$ lie on distinct facets of $P'$. To see this, suppose to the contrary that two points $x,y \in \Sigma'$ were on the same facet of $P'$ (see Figure~\ref{fig:pythag}). Then the line segment $\overline{x y}$ lies on $\partial P'$. Because $P'$ is a $(2\eps)$-approximation of $B$, $x$ and $y$ are both within distance $\Delta/4 + 2\eps$ of $B$'s center. Therefore, by the Pythagorean theorem and the fact that $\eps \le \Delta/4$, it follows that the distance from the midpoint of the line segment $\overline{x y}$ to $B$'s center is at most
\[
	\sqrt{\left(\frac{\Delta}{4} + 2\eps\right)^2 - \left(\frac{\delta}{2}\right)^2}
		~  =  ~ \sqrt{(\Delta/4)^2 + \Delta \eps + 4 \eps^2 - 4 \Delta \eps}
		~ \le ~ \sqrt{(\Delta/4)^2 - 2 \Delta \eps}
		~  <  ~ \Delta/4.
\]

This implies that the line segment $\overline{x y}$ intersects $B$'s interior, which contradicts the hypothesis that $P'$ is an outer approximation. Therefore, $P'$ must have at least $|\Sigma'| \ge c_b (\Delta / \eps)^{(d-1)/2}$ facets, and this completes the proof.
\end{proof} 

{\DualSubsetHelperStmt*}

\begin{proof}
To prove part~(i), observe first that, by the preconditions of Lemma~\ref{lem:dual}, $\delta = \sqrt{\eps}$ and $\eps \le 1/8$, and so $\delta \sqrt{8} \le 1$. Since the supporting line passing through $t$ passes on or below the origin, we have $t_y/t_x \le t_{\theta}$. (It may be helpful to refer to Figure~\ref{fig:dual-subset}.) Since $w_x \le t_x$, by convexity we have $w_y/w_x \le t_y/t_x$. Combining this, we have $w_y/w_x \le t_{\theta} \le \delta \sqrt{8} \le 1$, as desired.

To prove part~(ii), observe that because $t$ is a point of tangency with respect to $p_{\eps}$, $w$ lies on or above the line $p_{\eps} t$. This implies that $w_y \ge -\eps + t_{\theta} w_x$. Since $t_{\theta}$ is positive, $w_x \le (w_y + \eps)/t_{\theta}$. Because $w \notin B_{\delta}$, we know that $w_x^2 + w_y^2 > \delta^2$, which implies that $\max(w_x,w_y) > \delta/\sqrt{2}$. If $w_x > \delta/\sqrt{2}$, then
\[
	\frac{w_y}{w_x} 
		~ \ge ~ \frac{-\eps + t_{\theta} w_x}{w_x} 
		~  =  ~ t_{\theta} - \frac{\eps}{w_x} 
		~  >  ~ t_{\theta} - \frac{\eps\sqrt{2}}{\delta}
		~  =  ~ t_{\theta} - \delta \sqrt{2}.
\]
By our bound on $t_{\theta}$, we have $\delta\sqrt{2} < t_{\theta}/2$, implying that $w_y/w_x > t_{\theta} - t_{\theta}/2 = t_{\theta}/2$, as desired. 

On the other hand, if $w_y > \delta/\sqrt{2}$, then 
\[
	\frac{w_y}{w_x} 
		~ \ge ~ \frac{w_y}{(w_y + \eps)/t_{\theta}}
		~  =  ~ \frac{t_{\theta}}{1 + \eps/w_y}
		~  >  ~ \frac{t_{\theta}}{1 + \eps\sqrt{2}/\delta}
		~  =  ~ \frac{t_{\theta}}{1 + \delta\sqrt{2}}.
\]
The definition of $\delta$ and our bounds on $\eps$ imply that $1 + \delta\sqrt{2} < 2$, and so $w_y/w_x \ge t_{\theta}/2$, which completes the proof of~(ii).
\end{proof}

{\BaseCapAreaStmt*}

\begin{proof}
Let $D^{\downarrow}_{\delta}$ denote the orthogonal projection of $D_{\delta}$ onto $h$. Since the area of the orthogonal projection of a set cannot exceed the area of the original set, $\area(D^{\downarrow}_{\delta}) \le \area(D_{\delta})$. Given a set $X$ in real space, let $\inv{2} X$ denote the set $\{v/2 : v \in X\}$. We assert that $\inv{2} \Gamma_{\delta} \subseteq D^{\downarrow}_{\delta}$ (using the aforementioned coordinate system on $h$ whose origin is at $p$). Since $\Gamma_{\delta}$ lies in $\RE^{d-1}$, from this assertion we will have
\[
	\area(\Gamma_{\delta})
		~  =  ~ 2^{d-1} \cdot \area\left(\textstyle \inv{2} \Gamma_{\delta}\right)
		~ \le ~ 2^{d-1} \cdot \area(D^{\downarrow}_{\delta})
		~ \le ~ 2^{d-1} \cdot \area(D_{\delta}),
\]
from which the result will follow.

It remains to prove the assertion. Consider any $v \in \Gamma_{\delta}$, and recall that $B_{\delta}$ denotes the ball of radius $\delta = \sqrt{\eps}$ centered at $p$. By definition of $\Gamma_{\delta}$, $v \in B_{\delta}$, and so $v_x \le \delta$. It suffices to show that $v/2 \in D^{\downarrow}_{\delta}$, that is, there exists a point $w' \in D_{\delta}$ whose projection onto $h$ is $v/2$. Note that this is trivially true if $v = p$, and so we may assume that $v \ne p$. Under this assumption, consider the unique plane $\Phi$ defined by the points $p$, $p_{\eps}$, and $v$. We will impose the same coordinate system on $\Phi$ as in the proof of Lemma~\ref{lem:dual-subset}, with $p$ as origin and $p_{\eps}$ on the negative $y$-axis. Henceforth, we restrict our attention to this plane.

We may assume by symmetry that $v$ lies in the positive $x$-halfplane. By definition of $\Gamma_{\delta}$, there exists a point $t$ on the lower surface of $\partial K$ such that the line segment $p_{\eps} t$ passes through $v$. Let $w$ be the point of $\partial K$ whose orthogonal projection onto $h$ equals $v$ (see Figure~\ref{fig:base-cap-area}). Note that $w$ exists along the portion of $\partial K \cap \Phi$ between $p$ and $t$, and therefore $0 < w_x \le t_x$. Recall that $w_{\theta}$ and $t_{\theta}$ denote the slopes of the support lines at $w$ and $t$, respectively. By basic properties of convexity, we have $w_y/w_x \le w_{\theta} \le t_{\theta}$.

\begin{figure}[htbp]
  \centerline{\includegraphics[scale=0.40]{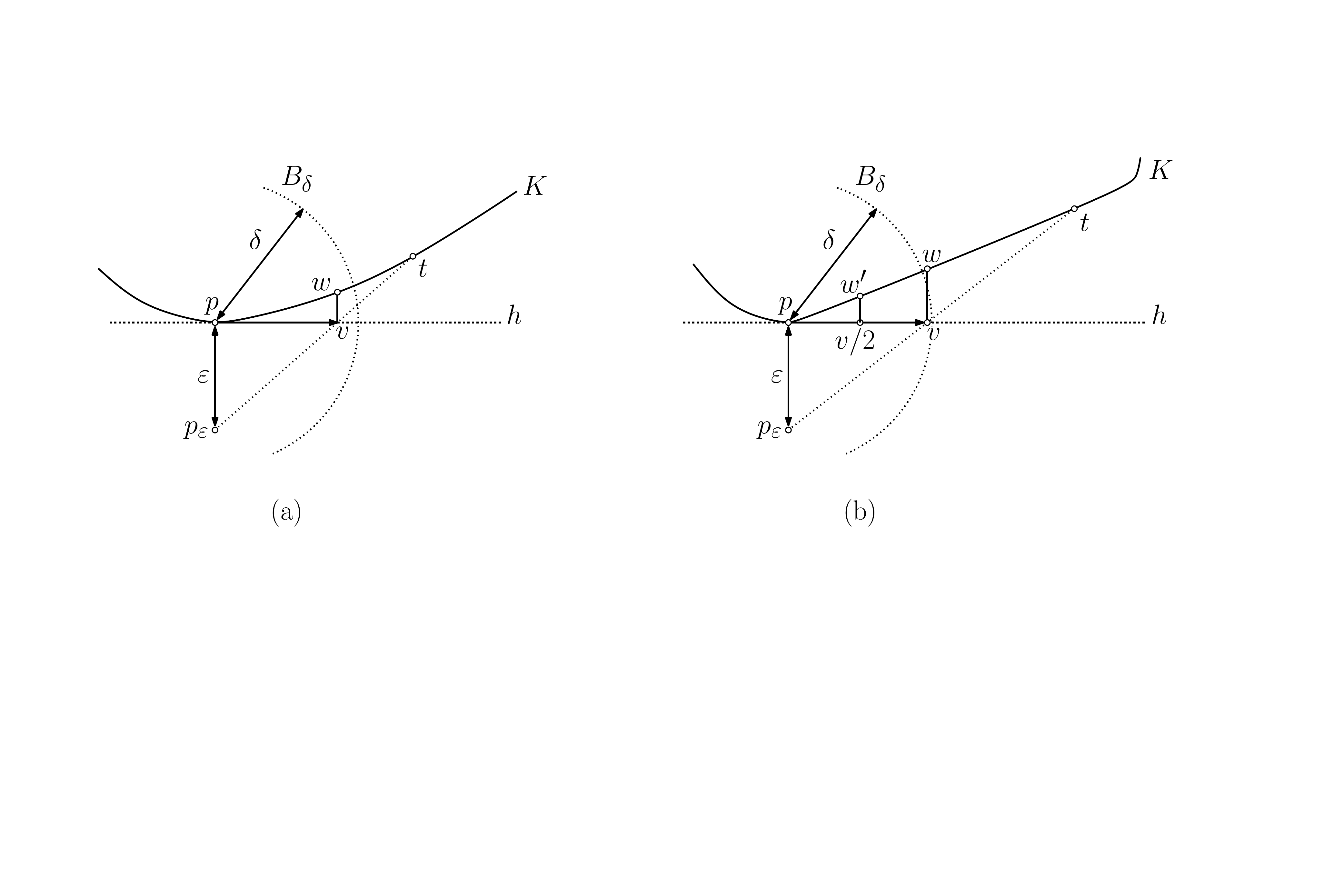}}
  \caption{Proof of Lemma~\ref{lem:base-cap-area}.}
  \label{fig:base-cap-area}
\end{figure}

We consider two cases, based on the location of $w$. First, if $w \in B_{\delta}$, then $w \in D_{\delta}$ (see Figure~\ref{fig:base-cap-area}(a)). This implies that $v \in D^{\downarrow}_{\delta}$. Clearly, $D^{\downarrow}_{\delta}$ is star-shaped with respect to $p$, and therefore $v/2 \in D^{\downarrow}_{\delta}$.

On the other hand, if $w \notin B_{\delta}$, then $w_x^2 + w_y^2 > \delta^2$ (see Figure~\ref{fig:base-cap-area}(b)). Because $v$ lies on the line $p_{\eps} t$, we have $v_y = -\eps + t_{\theta} v_x$. Also, since $v_x = w_x$ and $v_y = 0$, we have $t_{\theta} w_x = \eps$. Combined with the fact that $w_y/w_x \le t_{\theta}$, this yields $w_y \le \eps$. Let $w'$ be the point of $\partial K$ whose orthogonal projection onto $h$ is $v/2$. By convexity, we have $w'_y \le w_y \le \eps$ and $w'_x = v_x/2 \le \delta/2$. Because $\delta^2 = \eps \le \inv{8}$, we also have $\eps^2 \le \delta^2/8$. Thus, we obtain
\[
	\|w'\|^2
		~  =  ~ (w'_x)^2 + (w'_y)^2 
		~ \le ~ \left( \frac{\delta}{2} \right)^{\N 2} + \eps^2
		~ \le ~ \delta^2 \left( \inv{4} + \inv{8} \right) 
		~  <  ~ \delta^2,
\]
and therefore $w' \in B_{\delta}$, which implies that $v/2 \in D^{\downarrow}_{\delta}$, as desired.
\end{proof}

{\SmoothAreaStmt*}
\begin{proof}
Recall that $\eps' = \eps/2$. The proof is based on two assertions: 
\begin{enumerate}
\setlength{\itemsep}{-0.5ex}%
\setlength{\parsep}{0pt}%
\item[(1)] $\Gamma_{\delta}$ contains a $(d-1)$-dimensional Euclidean ball of radius $\eps'\gamma/2$.

\item[(2)] $\Gamma'_{\delta} \subseteq \Gamma_{\delta} \oplus \eps'$ (treating $\Gamma'_{\delta}$ and $\Gamma_{\delta}$ as subsets of $\RE^{d-1}$).
\end{enumerate}
To see why this suffices, observe that by (1), if we scale $\Gamma_{\delta}$ by a factor of $1 + 2/\gamma$ about the center of this ball, the scaled copy contains $\Gamma_{\delta} \oplus \eps'$. (To see why, observe that each supporting hyperplane of the original body is at distance at least $\eps'\gamma/2$ from the ball's center, and so the scaled body has a parallel supporting hyperplane at distance at least $(\eps'\gamma/2)(2/\gamma) = \eps'$ from the original supporting hyperplane.) Scaling increases $\Gamma_{\delta}$'s area by a factor of $(1 + 2/\gamma)^{d-1}$. By (2), $\Gamma'_{\delta}$ is contained within this scaled copy, and therefore its area cannot be larger. The result follows by setting $c = (1 + 2/\gamma)^{d-1}$.

We first prove assertion (1). Our approach will be to show that $K$ contains a ball that is sufficiently large and sufficiently close to $p$ that it contributes a ball of the desired radius to $\Gamma_{\delta}$. Since $K$ is $\gamma$-fat, there exist two concentric balls, $B^-$ and $B^+$, whose radii differ by a factor of $\gamma$ such that $B^- \subseteq K \subseteq B^+$ (see Figure~\ref{fig:smooth-area-1}(a)). Let $r$ and $r/\gamma$ denote the radii of $B^-$ and $B^+$, respectively, and let $y$ denote the distance of their common center to $h$. $B^-$ is the natural candidate for the desired ball, but it may be too far from $p$ to contribute to $\Gamma_{\delta}$ (due to restriction). Since $\diam(K) \ge 2 \eps$, $K$ does not lie entirely within a ball of radius $\eps$ centered at $p$. Let us scale space uniformly about $p$ so that $K$ just barely fits within this ball. Call the scaled body $K_0$, and let $B_0^-$ and $B_0^+$ denote the scaled copies of $B^-$ and $B^+$, respectively (see Figure~\ref{fig:smooth-area-1}(b)). We will show that $B_0^-$ is the desired ball. 

\begin{figure}[htbp]
  \centerline{\includegraphics[scale=0.40]{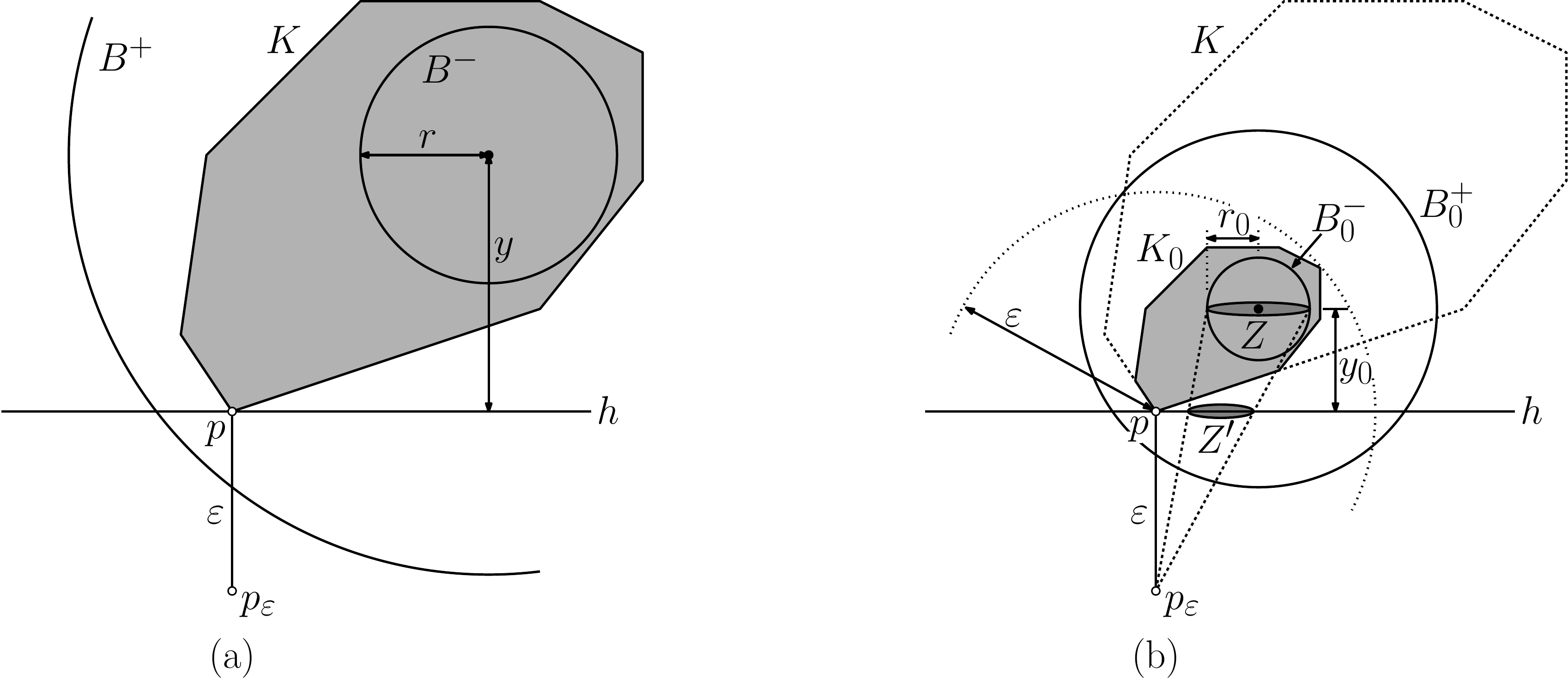}}
  \caption{Proof of assertion (1) of Lemma~\ref{lem:smooth-area}.}
  \label{fig:smooth-area-1}
\end{figure}

Since the scale factor is not greater than unity, $B_0^- \subseteq K_0 \subseteq K$. Let $r_0$ denote the radius of $B_0^-$, and let  $y_0$ denote its center's distance to $h$. Because the center of $B^-$ lies within $K$, we have
\[
	y
		~ \le ~ \diam(K)
		~ \le ~ \diam(B^+)
		~  =  ~ \frac{2 r}{\gamma},
\]
which implies that $y/r \le 2/\gamma$. Because the scaling is uniform, it follows that $y_0/r_0 = y/r \le 2/\gamma$. Another consequence of scaling is that
\[
	\eps
		~ \le ~ \diam(K_0) 
		~ \le ~ \diam(B_0^+) 
		~  =  ~ \frac{\diam(B_0^-)}{\gamma} 
		~  =  ~ \frac{2 r_0}{\gamma},
\]
which implies that $\eps/r_0 \le 2/\gamma$. Let $Z$ denote the $(d-1)$-dimensional horizontal diametrical disk of radius $r_0$ that lies within $B_0^-$ (see Figure~\ref{fig:smooth-area-1}(b)). It is easy to verify that if we project $Z$ centrally towards $p_{\eps}$ onto $h$, the result is a $(d-1)$-dimensional ball of radius $r' = r_0 \eps/(y_0 + \eps)$, which we denote by $Z'$. Clearly $Z'$ lies within the unrestricted dual-cap base. As a result of scaling, both $Z$ and $Z'$ lie within a ball of radius $\eps$ centered at $p$. Since $\eps < 1$, we have $\delta > \eps$, and therefore $Z'$ lies within the restricted base, that is $Z' \subseteq \Gamma_{\delta}$. By the above inequalities, the radius of $Z'$ is
\[
	\frac{r_0 \eps}{y_0 + \eps}
		~  =  ~ \frac{\eps}{y_0/r_0 + \eps/r_0}
		~ \ge ~ \frac{\eps}{2/\gamma + 2/\gamma}
		~  =  ~ \frac{\eps}{4/\gamma}
		~  =  ~ \frac{\eps' \gamma}{2},
\]
which establishes assertion~(1).

Next, we prove assertion~(2). Consider any point $v' \in \Gamma'_{\delta}$. By definition of $\Gamma'_{\delta}$ (and thinking of $v'$ as a vector in $\RE^d$), $v'$ is naturally associated with a point $q' \in \partial K'$ by shooting a ray from $p_{\eps}$ through $v'$ until it hits $\partial K'$ (see Figure~\ref{fig:smooth-area-2}). Since $K' = K \oplus \eps'$, there exists a unique point $q \in \partial K$ that is at distance $\eps'$ from $q'$. Let $u \in \RE^d$ denote the vector $q' - q$. Similarly, $q$ is associated with a point $v$ located where the line segment $q p_{\eps}$ intersects $h$. (The point $v$ lies on the base of the unrestricted dual cap $\Gamma$, but not necessarily on the restricted base $\Gamma_{\delta}$.) 

Now, thinking of $v$ and $v'$ as vectors in $\RE^{d-1}$, we claim (i) that there exists a scalar $0 \le \alpha' \le 1$ such that $v'$ lies within distance $\eps'$ of $\alpha' v$, and (ii) $\alpha' v$ is of length at most $\delta$. From (i) and the fact that $\Gamma$ is star-shaped with respect to the origin (namely $p$) it follows that $\alpha' v \in \Gamma$. From (ii), we have $\alpha' v \in \Gamma_{\delta}$. From these two claims we conclude that each point $v' \in \Gamma'_{\delta}$ is within distance $\eps'$ of a point in $\Gamma_{\delta}$, which implies assertion~(2). The remainder of the proof involves proving these two claims.

\begin{figure}[htbp]
  \centerline{\includegraphics[scale=0.40]{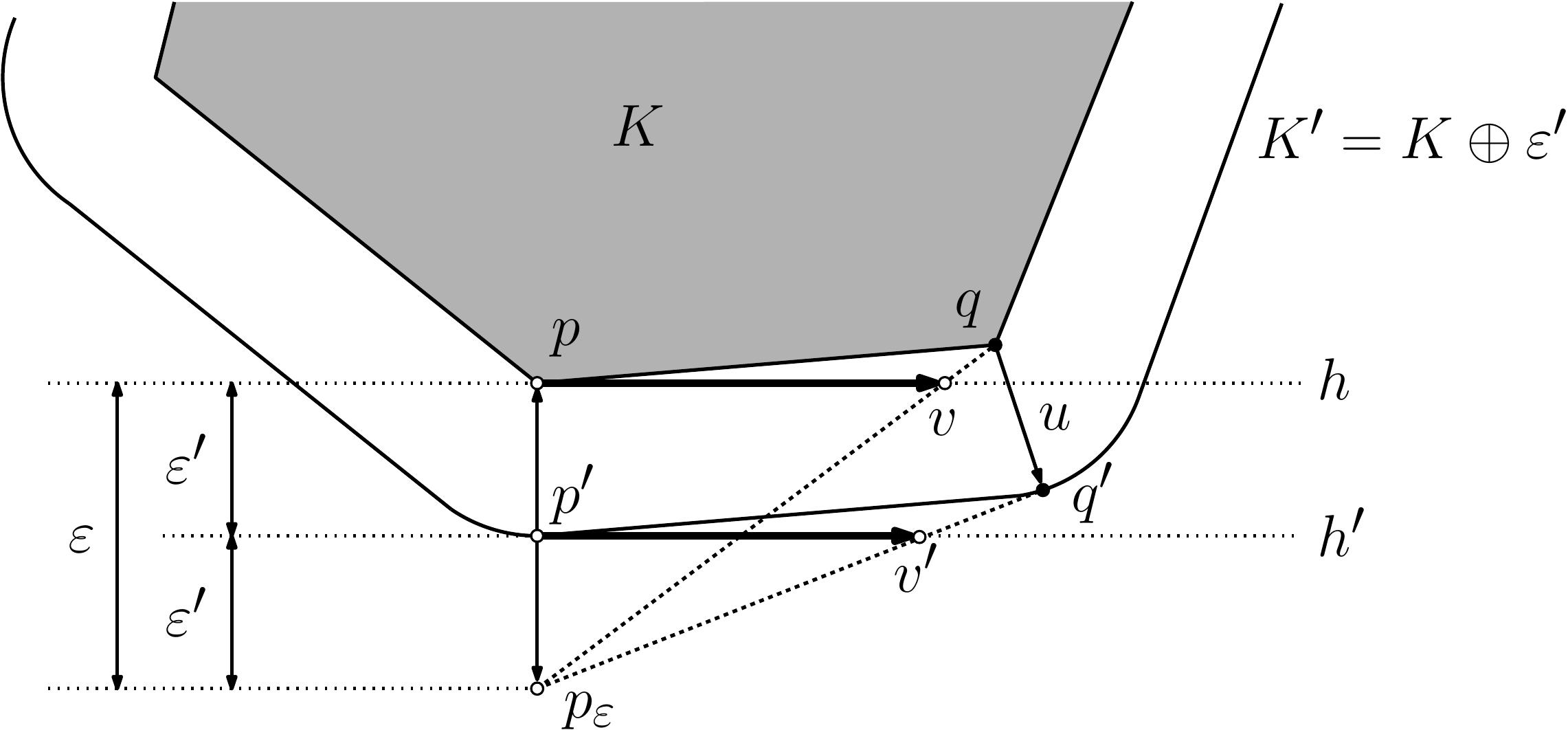}}
  \caption{Proof of assertion (2) of Lemma~\ref{lem:smooth-area}.}
  \label{fig:smooth-area-2}
\end{figure}

To establish claim~(i), let us consider a coordinate frame whose origin is $p_{\eps}$, whose $d$th coordinate axis points vertically upwards, and whose other $d-1$ coordinate vectors are taken from $h$'s coordinate frame. We may express any point $q \in \RE^d$ as a pair $(\overline{q}, q_d)$, where $\overline{q} \in \RE^{d-1}$ and $q_d$ is the vertical distance from $p_{\eps}$ to $q$. Define the transformation $T : \RE^d \rightarrow \RE^{d-1}$ that projects a point $q \in \RE^d$ centrally towards $p_{\eps}$ until it intersects $h$. Also define an analogous transformation $T'$ that projects a point $q'$ onto $h'$. It is easy to verify that
\[
	T(q) ~ = ~ \frac{\eps}{q_d} \cdot \overline{q}
	\qquad\mbox{and}\qquad
	T'(q') ~ = ~ \frac{\eps'}{q'_d} \cdot \overline{q}'.
\]
Since $q' = q + u$ and by definition of $\eps'$, we have
\begin{eqnarray*}
	T'(q')
		& = & T'(q + u)
		~ = ~ \frac{\eps'}{q_d + u_d} (\overline{q} + \overline{u}) 
		~ = ~ \frac{q_d}{2(q_d + u_d)} \left( \frac{\eps}{q_d} \cdot \overline{q} \right) + \frac{\eps'}{q_d + u_d} \cdot \overline{u} \\
		& = & \inv{2(1 + u_d/q_d)} T(q) + \frac{\eps'}{q_d + u_d} \cdot \overline{u}.
\end{eqnarray*}
Since $\|u\| = \eps'$, we have $-\eps' \le u_d \le \eps'$, and since $q \in \partial K$, we have $q_d \ge \eps$. It follows that 
\[
	\inv{3} ~\le~ \inv{2(1 + u_d/q_d)} ~\le~ 1
	\qquad\mbox{and}\qquad
	0 ~\le~ \frac{\eps'}{q_d + u_d} ~\le~ 1.
\]
Therefore, for some scalars $\inv{3} \le \alpha' \le 1$ and $0 \le \alpha'' \le 1$, we have
\[
	v'
		~ = ~ T'(q')
		~ = ~ \alpha' T(q) + \alpha'' \P \overline{u}
		~ = ~ \alpha' v + \alpha'' \P \overline{u}.
\]
Since $\|\overline{u}\| \le \eps'$ and $\alpha'' \le 1$, it follows that $v'$ lies within distance $\eps'$ of $\alpha' v$, which establishes claim~(i). 

To establish claim~(ii), observe that because $v' \in \Gamma'_{\delta}$ it lies within distance $\delta'$ of the origin on $h'$ (namely, $p'$). Therefore, $\alpha' v$ lies within distance $\delta' + \eps'$ of the origin on $h$ (namely, $p$). Since $\eps' = \eps/2 \le 1/16$ and $\delta' = \delta/\sqrt{2} = \sqrt{\eps/2}$, it is easy to verify that $\delta' + \eps' \le \delta$. Since $p$ is the origin, this implies that $\|\alpha' v\| \le \delta$ which establishes claim~(ii) and completes the proof.
\end{proof}

{\SmoothToPolarStmt*}
\begin{proof}
We begin by showing that $\Gamma'_{\delta}$ contains a $(d-1)$-dimensional Euclidean ball (centered at $p'$) of radius $\eps'/\sqrt{3}$. To see this, observe that because $K' = K \oplus \eps'$, $K'$ contains a ball $B$ of radius $\eps'$ that has $p'$ on its boundary (see Figure~\ref{fig:dual-2}(a)). By similar triangles it is easy to show that the hyperplane $h'$ intersects the ``ice cream cone'' shaped structure $\conv(B \cup \{p_{\eps}\})$ in a $(d-1)$-dimensional Euclidean ball of radius $\eps'/\sqrt{3}$.  This ball clearly lies within the unrestricted dual base, and since $\delta' = \sqrt{\eps'} < 1$, we have $\eps'/\sqrt{3} < (\delta')^2 < \delta'$, and so this ball also lies within the restricted dual base, $\Gamma'_{\delta}$.

\begin{figure}[htbp]
  \centerline{\includegraphics[scale=0.40]{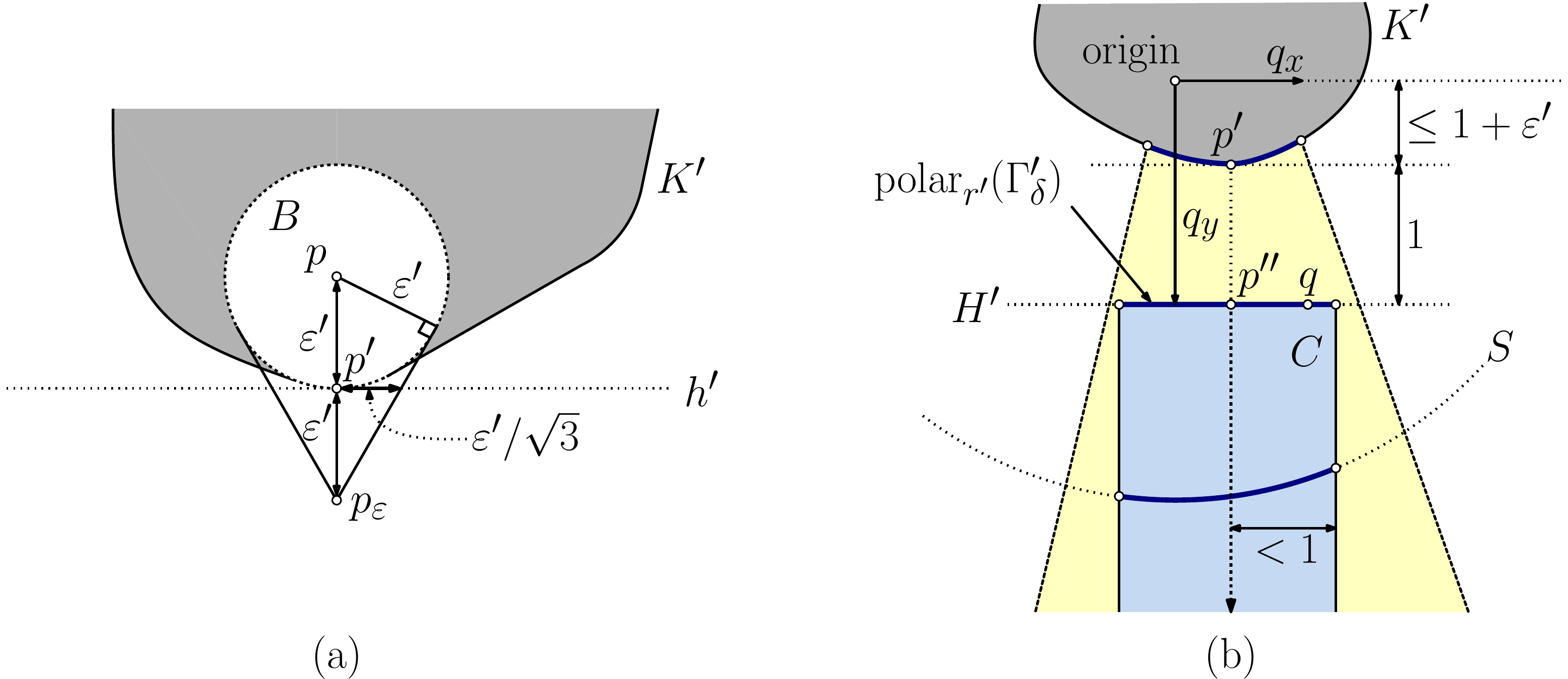}}
  \caption{Proof of Lemma~\ref{lem:smooth-to-polar}.}
  \label{fig:dual-2}
\end{figure}

For any $y \ge 1$, recall that $H^{(y)}$ is the hyperplane that is parallel to $h$ and at distance $y$ below $h$. Let $H' = H^{(1+\eps')}$ denote the hyperplane that is at unit distance below $p'$. Let $p''$ denote the vertical projection of $p'$ onto $H'$. Recall that we endow $h$ and $H'$ with parallel coordinate frames with origins at $p$ and $p''$, respectively. As a consequence of the above observation, for each vector $v$ of length $\eps'/\sqrt{3}$ in $\RE^{d-1}$, there is a halfspace bounding $\polarX{r'}{\Gamma'_{\delta}}$ that is orthogonal to $v$ and lies at distance $(\eps'/8)/\|v\| = (\eps'/8)/(\eps'/\sqrt{3}) = \sqrt{3}/8 < 1$ from the origin. Thus, $\polarX{r'}{\Gamma'_{\delta}}$ (when viewed as a subset of $H'$) is contained within a $(d-1)$-dimensional unit ball centered at $p''$ (see Figure~\ref{fig:dual-2}(b)).

Let $C$ denote the semi-infinite generalized cylinder whose horizontal cross section is $\polarX{r'}{\Gamma'_{\delta}}$, whose upper surface lies on $H'$, and which extends vertically downwards  (see Figure~\ref{fig:dual-2}(b)). Lemma~\ref{lem:dual-subset} (applied now to $K'$, $\eps'$, $\polarX{r'}{\Gamma'_{\delta}}$ and $\Vor(D'_{\delta}) \cap H'$) implies that $\polarX{r'}{\Gamma'_{\delta}} \subseteq \Vor(D'_{\delta}) \cap H'$. Since this applies not only to $H'$ but to any hyperplane lying below $H'$, it follows that $C \subseteq \Vor(D'_{\delta})$.

We will show that the orthogonal projection of $S \cap C$ onto $H'$ is equal to $\polarX{r'}{\Gamma'_{\delta}}$. It suffices to show that the base of $C$ (which lies on $H'$) lies entirely within $S$. We can express any point $q$ on $C$'s base as the sum of two perpendicular vectors $q_x + q_y$, where $q_x$ is horizontal and $q_y$ is vertical. Since $p$ lies within unit distance of the origin, $p'$ lies below $p$ at distance $\eps' = \frac{\eps}{2} \le \inv{16}$, and $H'$ lies at unit distance below $p'$, we have $\|q_y\| \le 1 + \inv{16} + 1 = 2 + \inv{16}$. As observed above, every point of $\polarX{r'}{\Gamma'_{\delta}}$ lies within unit distance of $p''$, and since $p$ lies within unit distance of the origin, we have $\|q_x\| \le 1 + 1 = 2$. Therefore, the squared distance from $q$ to the origin is
\[
	\|q_x\|^2 + \|q_y\|^2
		~ \le ~ \left( 2 + \inv{16} \right)^{\N 2} + 2^2
		~  <  ~ 9,
\]
which implies that $q$ lies within the sphere $S$ of radius $3$. Therefore, $C$'s base lies entirely within $S$.

Since $S \cap C \subseteq \Vor(D'_{\delta}) \cap S$, and since the area of the orthogonal projection of a set cannot be larger than the area of the original set, we have
\[
	\area(\Vor(D'_{\delta}) \cap S)
		~ \ge ~ \area(S \cap C)
		~ \ge ~ \area(\polarX{r'}{\Gamma'_{\delta}}),
\]
as desired.
\end{proof}
\end{document}